\documentclass[pre,twocolumn,showpacs,preprintnumbers,amsmath,amssymb,notitlepage
]{revtex4-1}
\usepackage{mathrsfs}
\usepackage{graphicx}
\usepackage{amssymb}
\usepackage{amsmath}

\usepackage{dcolumn}
\usepackage{bm}
\usepackage{color,hyperref}
\usepackage{xr}
\newcounter{one}
\def\QED{\mbox{\rule[0pt]{1.5ex}{1.5ex}}}

\def\endproof{\hspace*{\fill}~\QED\par\endtrivlist\unskip}

\newenvironment{proofof}[1]{\vspace*{5mm} \par \noindent
         {\bf Proof of #1:\hspace{2mm}}}{\endproof
}

\newenvironment{derivationof}[1]{\vspace*{5mm} \par \noindent
         {\bf Derivation of #1:\hspace{2mm}}}{\endproof
}

\newtheorem{lemma}{Lemma}
\newtheorem{theorem}{Theorem}
\newtheorem{corollary}{Corollary}
\newtheorem{definition}{Definition}

\newcommand\calA{{\cal A}}
\newcommand\calB{{\cal B}}
\newcommand\calC{{\cal C}}
\newcommand\calD{{\cal D}}
\newcommand\calE{{\cal E}}
\newcommand\calF{{\cal F}}

\newcommand\calH{{\cal H}}
\newcommand\calI{{\cal I}}

\newcommand\calM{{\cal M}}
\newcommand\calN{{\cal N}}

\newcommand\calR{{\cal R}}
\newcommand\calS{{\cal S}}

\newcommand\calU{{\cal U}}
\newcommand\calV{{\cal V}}

\newcommand\calX{{\cal X}}

\newcommand\id{\mathrm{id}}
\newcommand\Tr{\mathrm{Tr}}

\newcommand{\ex}[1]{\langle #1 \rangle}
\newcommand{\bra}[1]{\langle #1 |}
\newcommand{\ket}[1]{| #1 \rangle}

\newcommand{\g}[1]{{\color{black} #1}}

\newcommand{\tot}{\mathrm{tot}}

\newcommand{\non}{\nonumber\\}

\renewcommand{\vec}{\bm}

\newcommand{\wh}{\widehat}

\newcommand{\SA}{{\cal A}}

\newcommand{\eq}[1]{\begin{align} #1 \end{align}}

\begin{document}
\title{Universal limitation of quantum information recovery: symmetry versus coherence}

\author	{
			Hiroyasu Tajima
		}
		\email{hiroyasu.tajima@uec.ac.jp}
\affiliation{
				1. Department of Communication Engineering and Informatics, University of Electro-Communications, 1-5-1 Chofugaoka, Chofu, Tokyo, 182-8585, Japan
			}
\affiliation{
2. JST, PRESTO, 4-1-8 Honcho, Kawaguchi, Saitama, 332-0012, Japan
			}
\author{Keiji Saito}
\affiliation{
				Department of Physics, Keio University, 3-14-1 Hiyoshi, Yokohama, 223-8522, Japan
			}



\begin{abstract} 
Quantum  information is scrambled via chaotic time evolution in many-body systems. The recovery of initial information embedded locally in the system from the scrambled quantum state is a fundamental concern in many contexts. From a dynamical perspective, information recovery can measure dynamical instability in quantum chaos, fault-tolerant quantum computing, and the black hole information paradox. This article considers general aspects of quantum information recovery when the scrambling dynamics have conservation laws due to Lie group symmetries. Here, we establish fundamental limitations on the information recovery from scrambling dynamics with arbitrary Lie group symmetries. We show universal relations between information recovery, symmetry, and quantum coherence, which apply to many physical situations. 
The relations predict that the behavior of the Hayden-Preskill black hole model changes qualitatively under the assumption of the energy conservation law.
Consequently, we can rigorously prove that under the energy conservation law, the error of the information recovery from a small black hole remains unignorably large until it completely evaporates. 
Moreover, even when the black hole is very large, the recovery of information thrown into the black hole is not completed until most of the black hole evaporates.
The relations also provide a unified view of the symmetry restrictions on quantum information processing, such as the approximate Eastin-Knill theorem and the Wigner-Araki-Yanase theorem for unitary gates.
\end{abstract}


\maketitle
\section{Introduction}
Information locally embedded in many-body isolated systems generically diffuses over the entire system due to its chaotic dynamical motion. 
Recently, a significant amount of research has been conducted on the question of how to accurately recover the initial information from the scrambled state using a fixed protocol without any partial knowledge of the input information.
In classical systems, a chaotic motion with a positive Lyapunov exponent scrambles the phase space without limitations. Hence, the information recovery is practically hopeless due to the sensitivity against a small perturbation to the recovery protocol \cite{gaspard,casati}. However, in quantum systems, the scrambling nature becomes much more moderate than that in the classical case due to the existence of the Planck cell arising from Heisenberg's uncertainty relation. Hence, information recovery in the quantum regime seems feasible. 
Fundamental questions here are how and how accurately one can recover the information. The quantum information recovery problems have provided a lot of surprises at a fundamental level. Additionally, they have brought practical importance in several fields, such as fault-tolerant quantum computation.

The quantum information recovery problem historically dates back to the information paradox of the black hole \cite{Hawking1,Hawking2}. In the 1970s, Hawking raises a question on the information loss of the thrown information into a black hole. In the classical picture, information leakage from a classical black hole is unlikely due to the no-hair theorem \cite{nohair1}. However, quantum black holes can release the quantum information via the Hawking radiation \cite{Hawking1,Hawking2,Hayden-Preskill,Sekino-Susskind,Lashkari,Dupis,Yoshida-Kitaev}.
\begin{figure}[tb]
		\centering
		\includegraphics[width=.495\textwidth]{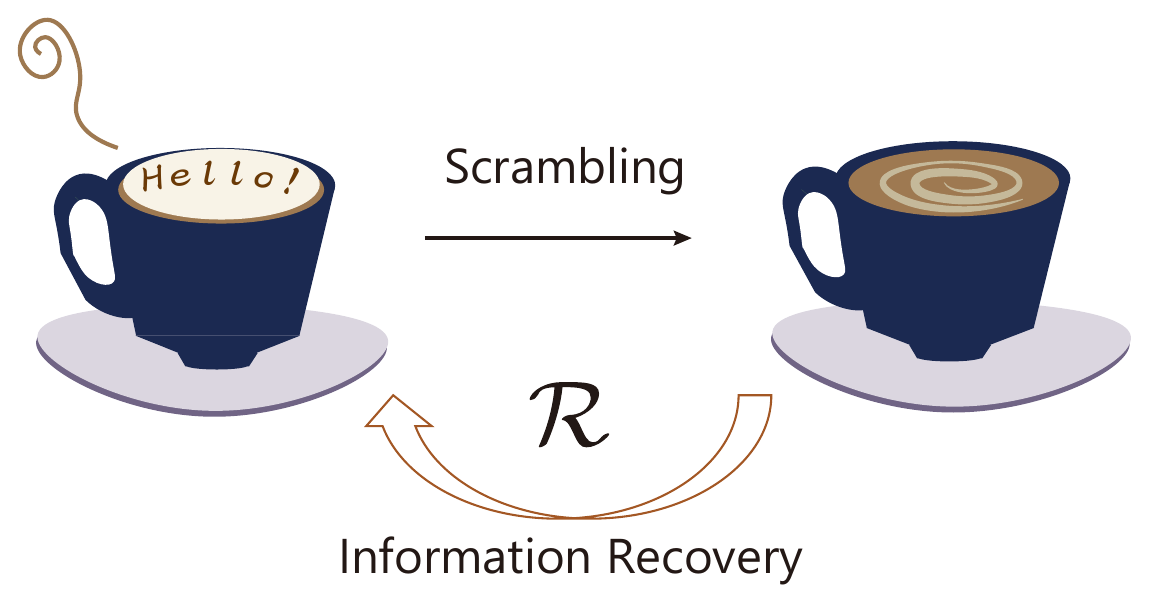}
		\caption{Schematic of the information recovery. We recover the initial information with a fixed protocol ${\mathcal R}$ after scrambling. }
		\label{recovery_fig}
	\end{figure}
Hayden and Preskill developed a model based on the quantum information theory and have found a remarkable result, i.e., arbitrary k-qubit quantum data thrown into the black hole can be almost perfectly recovered by collecting only a few more than $k$-qubit information from the Hawking radiation \cite{Hayden-Preskill}. In other words, quantum black holes act as informative mirrors. This finding is highly counterintuitive compared to the classical dynamical nature, and hence it has triggered a lot of studies \cite{Sekino-Susskind,Lashkari,Dupis,Yoshida-Kitaev}. 
Furthermore, today the technique and concept of information recovery go beyond the black hole physics and commonly appear in various problems considering the recovery of quantum correlations, such as quantum chaos \cite{yoshida-channel}, fault tolerance of quantum information \cite{Yan-Sinitsyn} and measurement-induced phase transition \cite{Choi}.
With recent development in experimental techniques related to quantum information, predictions of quantum information recovery are becoming verifiable. 
Various experimental setups on information recovery have been proposed, and several have been actually realized in the laboratory \cite{p1,p2,p3,siddiqi,lab1}.

Since the seminal work of Hayden and Preskill, the typical theoretical framework for information recovery adopted thus far is to use the Haar random unitary without any conservation laws to represent scrambling dynamics, exception of several works in specific situations \cite{Yoshida-soft,JLiu,Nakata}.
However, we should note that in real many-body systems, information scrambling can occur in dynamics with conservation laws, such as energy and momentum conservations.
Even for a black hole, which is considered the most chaotic system in the universe, the energy conservation law needs to be satisfied. 
Using the out-of-time-order correlation, an estimator of the degree of scrambling \cite{otoc1,otoc2}, one can show that isolated quantum many-body systems with nonlocal interactions demonstrate even fast scrambling phenomenon \cite{zhou1,zhou2,kuwa-sai}.
Conservation laws also can play a critical role in the dynamics in typical isolated many-body quantum systems such as isolated cold atomic systems.
With this backgrounds, it is vital to determine the universal effects of symmetries on the information recovery for the in-depth understanding of its quantum nature and also further applications.

In this work, we address this question by developing the techniques in resource theory of asymmetry \cite{Bartlett,Gour,Marvian,Marvian-thesis,Marvian2018,Lostaglio2018,TSS,TSS2,TN-WAY,Takagi2018,Marvian distillation}. 
Consequently, we  present several general limitations  on  information  recovery  when  the  scrambling  dynamics possess Lie group  symmetries.  
The limitations quantitatively state that when the scrambling has symmetries, significant inevitable errors occur in the recovery, and
only large amounts of quantum fluctuation of the conserved quantities can mitigate these errors to a certain extent.

Since our technique does not require assumptions other than unitarity and symmetry of the scrambling dynamics, the established limitations can be applied to many important situations. An interesting application exists within  black hole physics. In the information recovery from black holes, our results indicate that quantum coherence is required in the initial state of the black hole for accurate information recovery. 
Consequently, we demonstrate that when considering Hayden-Preskill's black hole model with the energy conservation law,  that law drastically limits the success rate of information recovery. 
Our results are summarized in two messages. First, when the size of both black hole and thrown object are comparable, the error of information recovery remains large until the black hole completely evaporates. Second, even when the black hole is much larger than the thrown object, 
information escapes very slowly, and a significant error in recovery remains until the black hole has almost evaporated.
Furthermore, our theorems can explain several limitations on quantum information processing with symmetry \cite{ozawa1,ozawa2,Karasawa2007,Karasawa2009,TSS,TSS2,Eastin-Knill,e-EKFaist,e-EKKubica,e-EKZhou,e-EKYang}.
Examples include the approximate Eastin-Knill theorem on covariant codes \cite{Eastin-Knill,e-EKFaist,e-EKKubica,e-EKZhou,e-EKYang}.
Our study demonstrates that limitations discussed independently so far can be derived from a single general theorem in a unified way.

This paper is organized as follows.
In the section \ref{sec2}, we formulate a general setup of quantum information recovery from unitary dynamics with symmetry.
In the section \ref{secmain}, we present the main results, i.e., the fundamental limitations on quantum information recovery.
In the section \ref{HPwithX}, we apply the main results to the Hayden-Preskill black hole model.
In the section \ref{APPnonsc}, we apply the results to the quantum information processing with symmetry.
In the section \ref{numerical check}, we give a numerical check of the main results.
Finally, in the Appendix, we provide basic tips of resource theory of asymmetry, and present the proof of the main results.

\section{Setup}\label{sec2}
A setup on the information recovery is introduced in a general form.
As discussed later, the setup described here is directly applicable to various situations including black hole scrambling \cite{Hayden-Preskill,Sekino-Susskind,Lashkari,Page,Dupis,Yoshida-Kitaev}, error correcting codes \cite{Eastin-Knill,e-EKFaist,e-EKKubica,e-EKZhou,e-EKYang} and the implementation of quantum computation gates \cite{ozawa1,ozawa2,Karasawa2007,Karasawa2009,TSS,TSS2}.

We consider four finite-level quantum systems $A$, $B$, $R_A$ and $R_B$, represented schematically in Fig.~\ref{setupSR}. 
The part $A$ is the system of interest with a mixed state $\rho_A$ as an initial state. Then, we make a purification between the system $A$ and $R_A$, the wave function of which is described as $\ket{\psi_{AR_A}}$. We assume that the initial state of the composite system $B R_B$ is pure state $\ket{\phi_{BR_B}}$, which is an entangled state. Through entanglement, the systems $R_A$ and $R_B$ have partial quantum information of the system $A$ and $B$, respectively. 
For this initial state, the unitary operation $U$ is applied on the systems $A$ and $B$, which scrambles the quantum information of the initial state. 
A main task in the information recovery problem is to recover the initial state $\ket{\psi_{AR_A}}$ with aid of partial information of the scrambled state. To this end, we suppose that the composite system $AB$ is either naturally or artificially divided into an accessible part $A'$ and the other part $B'$ after the unitary operation, where the Hilbert space of $AB$ and $A'B'$ are the same (see Fig.~\ref{setupSR} again). We then apply a recovery operation $\calR$ which is a completely positive and trace preserving (CPTP) map acting from $A'R_B$ to $A$ without touching $R_A$. Through this recovery operation, we try to recover the initial state $\ket{\psi_{AR_A}}$ as accurate as possible using the quantum information contained in the subsystems $A'$ and $R_B$.
Following the standard argument of information recovery including the black hole information paradox \cite{Hayden-Preskill,Sekino-Susskind,Lashkari,Page,Dupis,Yoshida-Kitaev} and the quantum error correction \cite{Eastin-Knill,e-EKFaist,e-EKKubica,e-EKZhou,e-EKYang}, we define the recovery error $\delta$ as the distance between the initial wave function $\ket{\psi_{AR_A}}$ and the output state on $AR_A$ with the best choice of the recovery operation:
\begin{align}
  \delta &:=\!\! \min_{\mbox{$\calR $ } \atop  (A'R_B \to A )     }  \!\!\!\!\!\!
           D_F\!\left( \rho_{AR_A},\mathrm{id}_{R_A}\otimes\calR[\Tr_{B'}(U\rho_{AR_A}\otimes\rho_{BR_B}U^\dagger) ] \right) \, , 
\end{align}
where $\rho_{A R_A} :=\ket{\psi_{AR_A}}\bra{\psi_{AR_A}}$ and $\rho_{B R_B} :=\ket{\phi_{B R_B}}\bra{\phi_{BR_B}}$. The symbol $\mathrm{id}_{R_A}$ represents the identity operation for the system $R_A$.
The function $D_F$ is the purified distance defined as $D_F(\rho,\sigma):=\sqrt{1-F(\rho,\sigma)^2}$ with the Uhlmann's fidelity $F(\rho , \sigma):=\Tr[\sqrt{\sqrt{\sigma}\rho\sqrt{\sigma}}]$ for arbitrary density operators $\rho$ and $\sigma$ \cite{Tomamichel}. The recovery error $\delta$ is a function of the initial states and the unitary operator.

We remark on another setup. Namely, one may want to input pure initial state for the system $A$ without using the reference state $R_A$, and ask the information recovery\cite{3Horodecki}. We can show that the above setup using the reference state is sufficient to evaluate the recovery accuracy even for this setup. See the argument in the Appendix \ref{AppB}.

When we look at the systems $A$ and $A'$, the unitary operation realizes a CPTP map $\calE$. Namely, the state on $A'$ after the unitary operation is simply described as $\calE(\rho_A)$.
From this picture, one may interpret the recovery error as an indicator of the {\it irreversibility} of the quantum operation $\calE$.

\begin{figure}[tb]
		\centering
		\includegraphics[width=.45\textwidth]{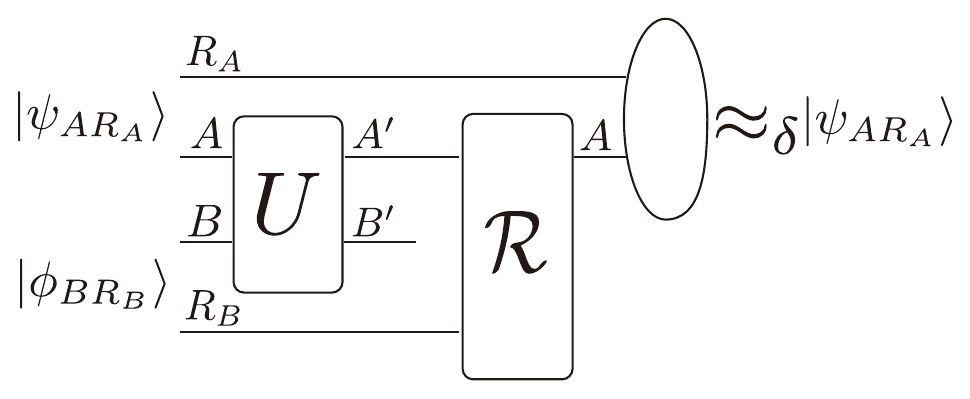}
		\caption{Schematic diagram of the general information recovery.}
		\label{setupSR}
	\end{figure}

        \g{The primary objective of this study} is to show that there is a fundamental limitation on the recovery error when the unitary operation has a Lie group symmetry. The symmetry generically generates conserved quantities such as energy and spin etc.
\g{For simplicity, we} consider a single conserved quantity $X$ under the unitary operation, i.e., 
\begin{align}
U(X_A+X_B)U^{\dagger}=(X_{A'}+X_{B'}) \, , \label{c-law}
\end{align}
where $X_{\alpha}$ is the operator of the local conserved quantity of the system $\alpha$ $(\alpha=A,B,A'\mbox{ or }B')$. We \g{note that} the case with many conserved quantities \g{can also be addressed} 
(see Supplementary Material Supp.$\rm V\hspace{-.1em}I$). 

We \g{now} introduce two key quantities to describe the limitation of information recovery.
While the conservation law for the total system is assumed, local conserved quantities can \g{fluctuate.} The first key quantity we focus on is the dynamical fluctuation associated with the quantum operation $\calE$, \g{i.e.,} a fluctuation of the change between the initial value of $X_A$ and the value of $X_{A'}$ after the quantum operation. The change of the values of \g{the} local conserved quantity depends on the initial state $\rho_A$. We characterize such fluctuation arising from the choice of the initial state,
\g{considering that} the initial reduced density operator for the system $A$ \g{can be decomposed} as $\rho_A =\sum_j p_j \rho_j$ \g{with} weight $p_j$ satisfying $\sum_j p_j =1$. Such a decomposition is not unique. While the linearity on the CPTP map guarantees that the decomposition reproduces the same output state on $A'$, i.e., $\calE ( \rho_A ) = \sum_j p_j \calE(\rho_j)$, each path from the density operator $\rho_j$ shows a variation on the change of local conserved quantities in general. Taking account of this property, we define the following quantity $\SA$ to quantify the dynamical fluctuation on the local conserved quantity for a given initial density operator:
\begin{align}
  \begin{split}
  \SA&:=\max_{\{p_j,\rho_j\}}\sum_jp_j|\Delta_j|,\\
  \Delta_j&:=\left(\ex{X_A}_{\rho_j}-\ex{X_{A'}}_{\calE(\rho_j)}\right)-\left(\ex{X_A}_{\rho_A}-\ex{X_{A'}}_{\calE(\rho_A)}\right) \, ,
  \end{split}
\end{align}
where $\ex{...}_{\rho}:=\Tr(...\rho )$, and the set $\{p_j,\rho_j\}$ covers all decompositions $\rho_A=\sum_jp_j\rho_j$. Note that the quantity $\SA $ is a function of the state $\rho_A$ and the CPTP map. 
When the systems $A$ and $B$ are \g{identical} to $A'$ and $B'$, \g{respectively,} and the unitary operator is decoupled between the systems as $U=U_A \otimes U_B$, the dynamical fluctuation is trivially zero. 
A finite value of the dynamical fluctuation is generated for a finite interaction between the systems\g{. This is} reflected from the fact that the global symmetry does not completely restrict the behaviour of the subsystem.

Another key quantity is quantum coherence. Following the standard argument in  
the resource theory of asymmetry, we employ  the SLD-quantum Fisher information \cite{Helstrom,Holevo} for the state family $\{e^{-iXt}\rho e^{iXt}\}_{t\in\mathbb{R}}$ to quantify the quantum coherence on $\rho$ \cite{Takagi2018,Marvian distillation}:
\begin{align}
\calF_{\rho}(X):=4\lim_{\epsilon\rightarrow0}\frac{D_F(e^{-iX\epsilon}\rho e^{iX\epsilon},\rho)^2}{\epsilon^2}.
\end{align}
The quantum Fisher information is a good indicator of the amount of quantum coherence in $\rho$ with the basis of the eigenvectors of $X$. It is known that this quantity is directly connected to the amount of \textit{quantum fluctuation} (see Appendix \ref{AppA}) \cite{Q-Fisher=Q-fluctuation1, Q-Fisher=Q-fluctuation2}. We consider the quantum coherence contained inside the system $B$ as discussed below. 

\section{Main results}\label{secmain}
\subsection{Fundamental limitation on information recovery}\label{mainA}
With the two key quantities introduced above, we establish two fundamental relations on the limitations of the information recovery. 
We note that the relations are obtained for general unitary operations with conservation laws, \g{without assumptions} such as the Haar random unitary.

The first relation on the limitation of the information recovery is described \g{as} follows \cite{footnoteA2}:
\begin{align}
\frac{\SA}{2(\sqrt{\calF}+4\Delta_{+})}\le\delta \, , \label{SIQ1}
\end{align}
where $\calF:=\calF_{\rho_{ BR_B}}(X_B\otimes1_{R_B})$ is the quantum coherence in the initial state of the system $BR_B$. The quantity $\Delta_+$ is a measure of possible change on the local conserved quantities, i.e., $\Delta_+:=(\calD_{X_A}+\calD_{X_{A'}})/2$ where $\calD_{X_A}$ and $\calD_{X_{A'}}$ are the differences between the maximum and minimum eigenvalues of the operators $X_A$ and $X_{A'}$, respectively.

The inequality (\ref{SIQ1}) shows a close relation between the recovery error (irreversibility), the dynamical fluctuation, and the quantum coherence. 
It shows that when the dynamical fluctuation is finite, \g{perfect recovery} is impossible.
Moreover, high performance recovery is possible only when the quantum coherence \g{sufficiently fills} the initial state of $BR_B$.
Note that the dynamical fluctuation is generically finite, since the systems $A$ and $B$ interact with each other via the unitary operation. 
We show a specific example in Supplementary Materials Supp.$\rm I\hspace{-.1em}I$, 
where filling vast quantum coherence in $BR_B$ actually makes the error $\delta$ smaller than $\SA/8\Delta_{+}$ and negligibly small.

The above inequality uses the quantum coherence $\calF$ of the initial state of $BR_B$. We can also \g{establish} another inequality with the quantum coherence of the final state, which is the second main relation \cite{footnoteA2}:
\begin{align}
  \frac{\SA}{2(\sqrt{\calF_f}+\Delta_{\max})}\le \delta
    \, ,\label{SIQ2}
\end{align}
where $\Delta_{\max}:=\max_{\{p_j,\rho_j\}}\max_{j}|\Delta_{j}|$, \g{and} the set $\{p_j,\rho_j\}$ covers all decompositions satisfying $\rho=\sum_{j}p_j\rho_j$. The quantum coherence here is measured for the final state as $\calF_f:=\calF_{\sigma_{B'R'_B}}(X_{B'}\otimes1_{R'_B})$, where the state $\sigma_{B'R'_{B'}}$ is a purification of the final state of $B'$ using the reference $R'_{B'}$.

It is critical to comment on what happens if the symmetry is violated. One can discuss the degree of violation of the symmetry, by defining the operator $Z:=(X_A+X_B)-U^{\dagger}(X_{A'}+X_{B'})U$ and its variance $V_Z:=V_{\rho_A\otimes\rho_B}(Z)$. 
  Then, the dynamical fluctuation term in the relations \eqref{SIQ1} and \eqref{SIQ2} is replaced by a modified function which becomes small when the degree of violation is large (see Supplementary Material Supp.$\rm V\hspace{-.1em}II$). For instance, the relation (\ref{SIQ1}) is modified as the inequality $(\SA-V_Z)/[2(\sqrt{\calF}+4\Delta_{+}+3V_Z)]\le\delta $.
When the violation of \g{the} symmetry is large, the numerator becomes negative, which implies that the inequalities reduce to trivial bounds. Hence, the meaningful limitations provided above exist due to the existence of symmetry.

\subsection{Limitation on the information recovery without using $R_B$}\label{S-woRBmain}
Here we discuss the case without using the information of $R_B$. The recovery operation $\calR$ in this case maps the state on the system $A'$ to $A$.  
We then define the recovery error as
\begin{align}
  \tilde{\delta}
  :=\min_{\mbox{$\calR$} \atop {A'\to A }} \!\!
  D_F\!\left( \rho_{AR_A},\mathrm{id}_{R_A}\otimes\calR\circ\calE(\rho_{AR_A})] \right) \, .\label{dwithoutRBmain}
\end{align}
Since $\tilde{\delta}\ge\delta$, \g{we} can substitute $\tilde{\delta}$ for $\delta$ in \eqref{SIQ1} and \eqref{SIQ2} to get a limitation of recovery in the present setup. Moreover, in Supplementary Material Supp.$\rm I\hspace{-.1em}V$, we can derive a tighter relation than this simple substitution as follows \cite{footnoteA2}:
\begin{align}
\frac{\SA}{2(\sqrt{\calF_B}+4\Delta_{+})}\le\tilde{\delta} \, , \label{SIQ1'-Smain}
\end{align}
where $\calF_{B}:=\calF_{\rho_B}(X_B)$. Note that $\calF_{B}\le\calF$ holds in general.

\subsection{Mechanism of how conservation laws hinder quantum information recovery}

Let us explain in an intuitive manner why conservation laws prevent quantum information recovery. 
We focus here on the case of
the perfect recovery, i.e., $\delta =0$.
In this case, after applying the recovery map $\calR$, the state of the system $AR_A$ is equal to the initial state $\ket{\psi_{AR_A}}$. (Fig.~\ref{delta_0_case}) Given that the state $\ket{\psi_{AR_A}}$ is pure, the final state of the system $B’$ is uncorrelated with $AR_A$: 
\eq{
&\psi_{AR_A}\otimes\rho_{B'}\non
&=(\calR\otimes\id_{B'R_A})\circ(\calU\otimes\id_{R_AR_B})(\psi_{AR_A}\otimes\phi_{BR_B}),\label{9main}
}
where $\calU(...):=U...U^\dagger$, $\psi_{AR_A}:=\ket{\psi_{AR_A}}\bra{\psi_{AR_A}}$ and $\phi_{BR_B}:=\ket{\phi_{BR_B}}\bra{\phi_{BR_B}}$.
Therefore, no matter what measurement is made on the final state of system $R_A$, no change will occur in the system $B’$ according to the result of that measurement. Because $R_A$ does not interact with any other system during the whole process, performing a measurement on the \textit{final} state of $R_A$ and performing the same measurement on the \textit{initial} state of $R_A$ will have exactly the same result.
This is confirmed by performing a measurement $\{M_{j,R_A}\}$ on $R_A$ on both hand sides of \eqref{9main}:
\eq{
&\rho_{j,AR_A}\otimes\rho_{B'}\non
&=(\calR\otimes\id_{B'R_A})\circ(\calU\otimes\id_{R_AR_B})\left(\rho_{j,AR_A}\otimes\phi_{BR_B}\right).\label{10main}
}
where $\rho_{j,AR_A}:=M_{j,R_A}\psi_{AR_A}M^\dagger_{j,R_A}/q_j$ and $q_j:=\Tr[M^\dagger_jM_j\psi_{AR_A}]$.
Because the recovery map $\calR$ is a CPTP map from $A'R_B$ to $A$, we can remove it from \eqref{10main} through the partial tracing of  $AR_A$ on both sides of the equation.
We then obtain
\eq{
\rho_{B'}=\Tr_{A'}[U\left(\rho_{j,A}\otimes\rho_{B}\right)U^\dagger].\label{11main}
}
Here $\rho_{j,A}:=\Tr_{R_A}[\rho_{j,AR_A}]$ and $\rho_B:=\Tr_{R_B}[\phi_{BR_B}]$.
Noting that $\rho_{j,A}$ is the resultant state on $A$ obtained when the measurement outcome of $\{M_j\}$ on $R_A$ is $j$, we can see that no matter what measurement is applied to the initial state of $R_A$, the final state of $B'$ will be independent of the results of the measurement.

\begin{figure}[tb]
		\centering
		\includegraphics[width=.4\textwidth]{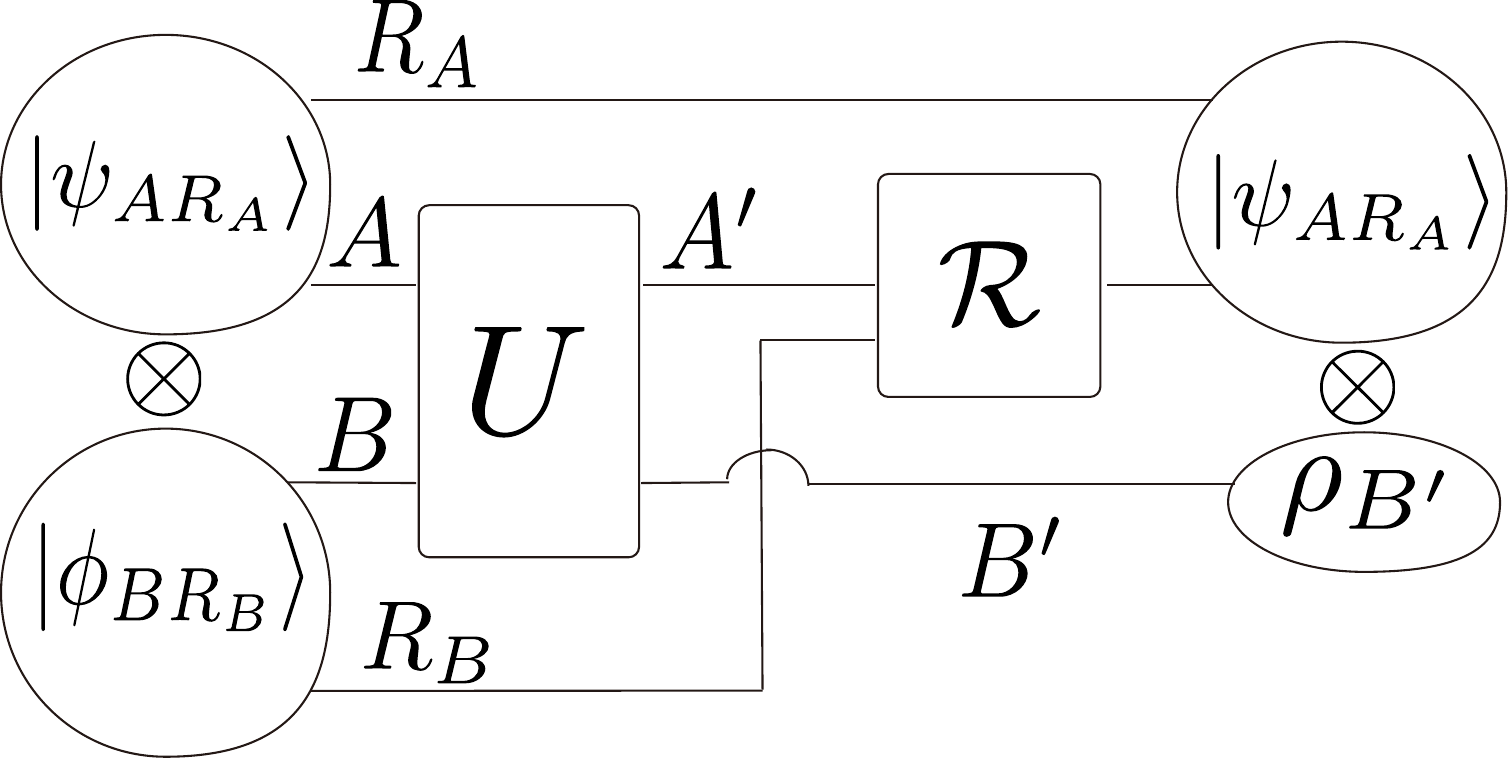}
		\caption{Schematic of a fully successful quantum information recovery. In this case, the error $\delta$ is equal to zero, and the final state of $AR_A$ is equal to the initial state $\ket{\psi_{AR_A}}$. Moreover, the final state of $AR_A$ is completely uncorrelated with $B'$, because pure states do not have any correlation with other systems. Therefore, when $\delta=0$, any measurement on $R_A$ cannot affect the state of $B'$. However, when a conservation law exists, and when $\calA>0$ holds, the unitary $U$ correlates $AR_A$ and $B'$, and hence measurements on $R_A$ affect the final state of $B'$. Therefore, under the conservation law, the perfect information recovery from the scrambling processes satisfying $\calA>0$ cannot be attained.}
		\label{delta_0_case}
\end{figure}

However, when there is a conservation law and a local conserved quantity is changed by the local dynamics (i.e., when $\calA>0$ is valid), the above never holds. This is because by performing a measurement on the initial state of $R_A$, the state of $A$ changes according to the results of the measurement, and the expectation value of the conserved quantity in $B'$ also changes. 
A simple example to understand this argument is when the conserved quantity $X$ is energy and the initial state $\ket{\psi_{AR_A}}$ of $AR_A$ is $(\ket{00}_{AR_A}+\ket{11}_{AR_A})/\sqrt{2}$ such that the strict inequalities $\ex{X_A}_{\ket {0}\bra{0}_A}-\ex{X_{A'}}_{\calE(\ket{0}\bra{0}_A)}>0$ and $\ex{X_A}_{\ket{1}\bra{1}_A}-\ex{X_{A'}}_{\calE(\ket{1}\bra{1}_A)}<0$ hold.
This example indicates the change in expectation value of energy from $A$ to $A'$ is positive if the initial state of $A$ is $\ket{0}_A$ and negative if $\ket{1}_A$.
Because energy conservation holds, the change in energy from $B$ to $B'$ is negative (positive) if the state of $A$ is $\ket{0}_A$ ($\ket{1}_A$). 
In other words, the expectation value of $X_{B'}$ is different depending on whether the state of $A$ is $\ket{0}_A$ or $\ket{1}_A$.
Therefore, in this example, the measurement $\{\ket{0}\bra{0}_{R_A},\ket{1}\bra{1}_{R_A}\}$ on $R_A$ affects the final state of $B’$ because, depending on whether the measurement result of $R_A$ is 0 or 1, the initial state for $A$ of $\ket{0}_A$ or $\ket{1}_A$ changes, and therefore the expectation value of $X_{B'}$ also changes.
However, when $\delta=0$, any measurement on $R_A$ cannot affect the final state of $B’$.
This is a contradiction, and therefore, $\delta=0$ and the conservation law are incompatible, at least when $\calA>0$.

The above argument intuitively shows how symmetry hinders quantum information recovery. For finite values of $\delta$, this argument can be refined to derive quantitative trade-off relations. The derivation of \eqref{SIQ1} and \eqref{SIQ2} follows similarly (see Appendix \ref{derivation_main} for details).

Also, note that the above argument holds only when $U$ satisfies the conservation law. When there is no conservation law, the change in energy for $B$ is not necessarily linked to that for $A$. Therefore, in this case, the perfect recovery is possible even if $\ex{X_A}_{\ket {0}\bra{0}_A}-\ex{X_{A'}}_{\calE(\ket{0}\bra{0}_A)}>0$ and $\ex{X_A}_{\ket{1}\bra{1}_A}-\ex{X_{A'}}_{\calE(\ket{1}\bra{1}_A)}<0$ hold.  Hence our inequalities \eqref{SIQ1} and \eqref{SIQ2} are weakened when $U$ violates the conservation law (see Supplementary Material Supp.$\rm V\hspace{-.1em}II$). An example is known in the context of black hole information recovery when $U$ is a general Haar random unitary with no conservation laws, allowing almost perfect recovery independent of the state of $A$. This example is presented in the next section.

\section{Application to the Hayden-Preskill model with a conservation law}\label{HPwithX}

\begin{figure}[b]
		\centering
		\includegraphics[width=.45\textwidth]{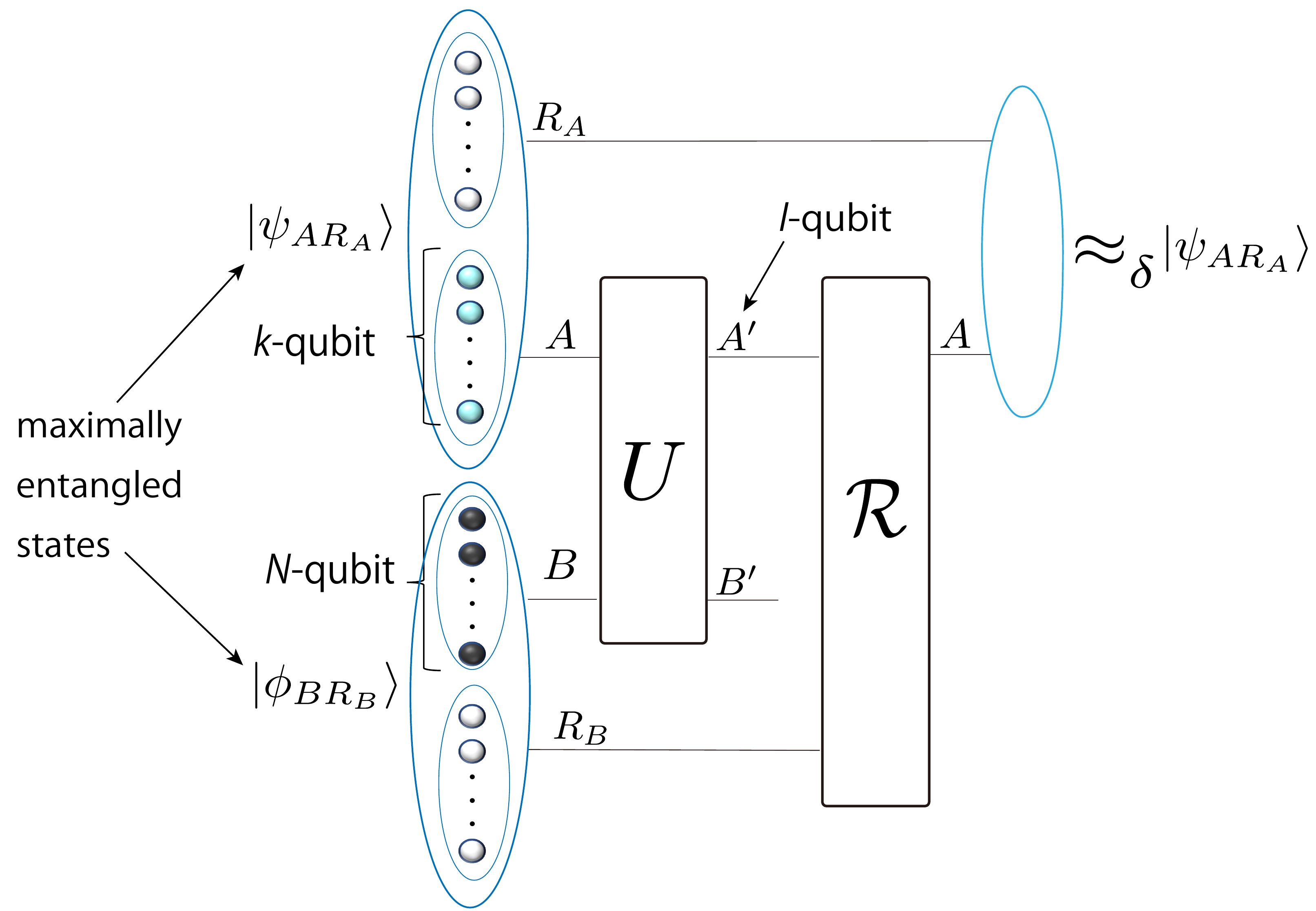}
		\caption{Schematic diagram of the HP black hole model, which is almost a special case of our setup illustrated in Fig.~\ref{setupSR}.}
		\label{HPmodel}
	\end{figure}

Our results are applicable to the black hole information recovery problems with a conservation law.

\g{Here, we briefly review} the Hayden-Preskill (HP) model \cite{Hayden-Preskill} (Fig.~\ref{HPmodel}). The HP model is a quantum mechanical model where Alice trashes her diary $A$ into a black hole $B$, and Bob tries to recover the contents of the diary through \g{Hawking radiation}, assuming that the dynamics of the black hole is unitary. 
The diary $A$ contains $k$-qubit \g{quantum information, and} is initially maximally entangled with another system $R_A$. The black hole is assumed to contain $N$-qubit quantum information, where \g{ $N:=S_{BH}$ is interpreted} as the Bekenstein-Hawking entropy. 
After throwing the diary into the black hole, the HP model assumes \g{a Haar random unitary operation that} scrambles the quantum information  \cite{Hayden-Preskill,Lashkari,kuwa-sai}. 
Another assumption is that the black hole $B$ is sufficiently old\g{, and} is maximally entangled with another system $R_B$\g{,} which is the Hawking radiation emitted from $B$ before the diary $A$ is trashed. 
Bob can use the information in $R_B$\g{, and} can capture and use the Hawking radiation emitted after $A$ is trashed, \g{denoted by} $A'$. The quantum information of $A'$ is assumed to be \g{of} $l$-qubits. Then, we perform a quantum operation $\calR$ from $A'R_B$ to $A$, and \g{recover} the initial maximally entangled state of $AR_A$.
We remark that recently realization of this recovery setup through laboratory experiment is proposed \cite{p1,p2,p3,siddiqi,lab1}.

Under this setup, Hayden and Preskill \g{established} the following \textit{upper bound} of the recovery error \cite{Hayden-Preskill}:
\begin{align}
\delta\le \mathrm{const.}\times2^{-(l-k)/2}.\label{mirrorineq}
\end{align}
\g{A remarkable aspect of this result} is that the recovery error decreases exponentially \g{with increasing $l$, and that} only a few more qubits than $k$ \g{are required to recover} the initial state with good accuracy.

	\begin{figure}[tb]
		\centering
		\includegraphics[width=.5\textwidth]{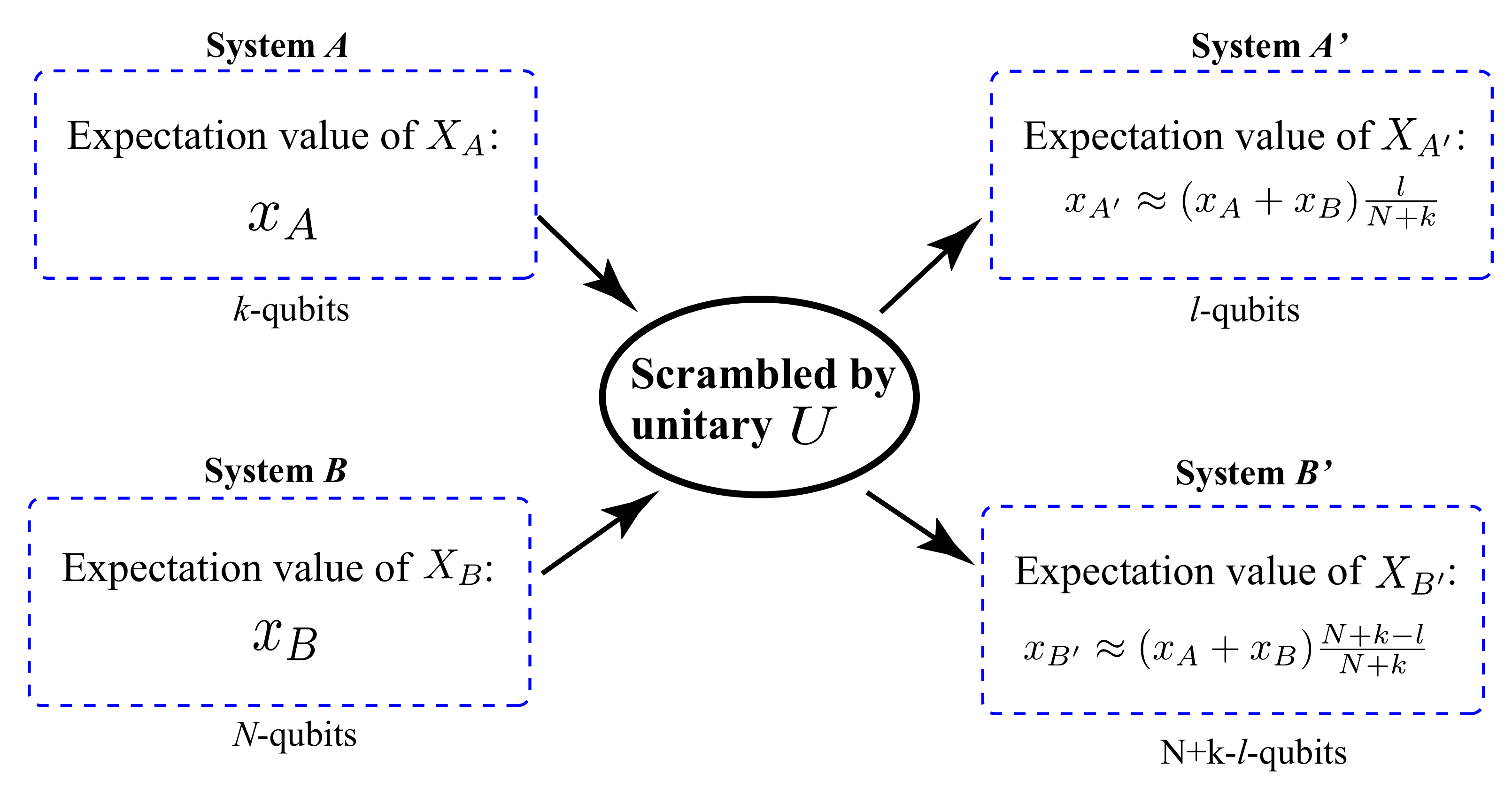}
		\caption{Schematic diagram of the assumption of how the expectation value of \g{the} conserved quantity $X$ is distributed.  In this diagram, we refer to the expectation values of $X$ in $\alpha$ as $x_\alpha$ ($\alpha=A,B,A',\mbox{ and }B'$). 
We assume that the expectation value is given through the equidistribution. 
Precisely, we assume that after the unitary time evolution $U$, the expectation values of the conserved quantity $X$ are divided \g{among} $A'$ and $B'$ in proportion to \g{the corresponding number of qubits}.}
		\label{ex-scrambling}
	\end{figure}

Note that the setup for the HP model is similar to \g{the setup described in Section} \ref{sec2}. 
The important difference is that the unitary operation of the HP model is described by the Haar random unitary without any conservation law \eqref{c-law}, whereas the dynamics of our setup has symmetry. 
We discuss the effect of this symmetry that generates a conserved quantity $X$, e.g., energy.
For simplicity, we also set the difference between minimum and the maximum eigenvalue of $X_{\alpha}$ (= $\calD_{X_{\alpha}}$ for $\alpha=A,B,A',B'$) to be equal to the number of particles of the system $\alpha$ (= $k$, $N$, $l$, $N+k-l$ for $\alpha=A,B,A',B'$). 
We do not use the Haar random unitary, but impose two weaker assumptions that are provided by typicality \cite{cano_typi}.
First, the expectation values of $X_{A'}$ and $X_{B'}$ are equi-distributed (see Fig.~\ref{ex-scrambling}).
Second, the black hole dynamics $U$ satisfies the following inequality for arbitrary eigenstates $\ket{i,a}$ and $\ket{j,b}$ of $X_A$ and $X_B$, unless the sum of the eigenvalues of $\ket{i,a}$ and $\ket{j,b}$ is too close to the maximum or minimum eigenvalues of $X_A+X_B$:
\eq{
V_{\rho'_{\alpha'|i,a,j,b,U}}(X_{\alpha'})\le\frac{1+\epsilon}{4}\min\{l,\gamma(N+k)\},\label{Rstar}
}
where $\epsilon$ is a negligibly small number describing the error of the equidistribution on the expectation value and where $\rho'_{\alpha'|i,a,j,b,U}:=\Tr_{\neg\alpha'}[U(\ket{i,a}\bra{i,a}\otimes\ket{j,b}\bra{j,b})U^\dagger]$ where $(\alpha',\neg\alpha')=(A',B') $ or $(B',A')$.
These two assumptions hold whenever $U$ thermalizes the state of the subsystem sufficiently.
Indeed, \g{when} $U$ is a typical Haar random unitary satisfying \eqref{c-law}, these two assumptions are rigorously shown to be satisfied (see Supplementary Material Supp.$\rm I\hspace{-.1em}I\hspace{-.1em}I$).
Additionally, to increase the generality of the results, we do not restrict the initial states $\ket{\psi_{AR_A}}$ and $\ket{\phi_{BR_B}}$ to the maximally entangled states.
Instead of that, we only assume that the initial state $\rho_B:=\Tr_{R_B}[\phi_{BR_B}]$ of the black hole $B$ satisfies $V_{\rho_B}(X_B)\le  \calD_{X_B}/4(=N/4)$. 
We remark that this assumption is satisfied when $X$ is energy, $B$ is a natural thermodynamic system, and $\rho_B$ is a non-zero temperature Gibbs state including the maximally mixed state.

Under the above conditions, we now use the result \eqref{SIQ2}.
In particular, when $\rho_A$ commutes with $X_A$, we can evaluate $\SA$, $\calF_f$, and $\Delta_{\max}$ in \eqref{SIQ2} as follows (for details, see Supplementary Material Supp.$\rm I\hspace{-.1em}I\hspace{-.1em}I$):
\begin{align}
\SA&\ge \gamma M(1-\epsilon) \, , \label{A-bound}\\
\sqrt{\calF_{f}}&\le2(1+\epsilon)\sqrt{\gamma(N+k^2)} \, \label{Rstar2}, \\
\Delta_{\max}&\le\gamma k(1+\epsilon) \, , \label{D-upp}
\end{align}
where $\gamma:=\left(1-l/(N+k)\right)$, and $M:=\ex{|X_A-\ex{X_A}|}_{\rho_A}$ is the mean deviation of $X_A$ in $\rho_A$.
Due to \eqref{A-bound}--\eqref{D-upp}, when $N+k>l$, we can convert \eqref{SIQ2} into the following form:
\eq{
\frac{1-\epsilon}{1+\epsilon}\frac{\gamma M}{2(\gamma k+2\sqrt{\gamma(N+k^2)})}\le\delta
}
To \g{interpret} this inequality, we set $k=\sqrt{N}$ and $M=k/2$ (\g{we can take such an} $M$ by considering a relevant $\rho_A$ and its decomposition, e.g., $\rho_A=(\rho^{\max}_{k}+\rho^{\max}_{0})/2$, where $\rho^{\max}_{x}$  is the maximally mixed state of the eigenspace of $X_A$ corresponding to eigenvalue $x$).
We then obtain the following \textit{lower bound} of the recovery error:
\eq{
\frac{\mathrm{const.}}{1+\frac{2\sqrt{2}}{\sqrt{\gamma}}}\le\delta,\label{doda}
}
where $\mathrm{const.}$ is a real number larger than $0.24$.
This inequality rigorously restricts recovery of quantum information from the black hole. 
Since this inequality depends only on the ratio between the remaining part of the black hole $B'$ and the total amount of qubits, i.e. $\gamma=1-l/(N+k)$, we see that the recovery error remains non-negligible, even after a considerable evaporation of the black hole.

\begin{figure}[tb]
		\centering
		\includegraphics[width=.45\textwidth]{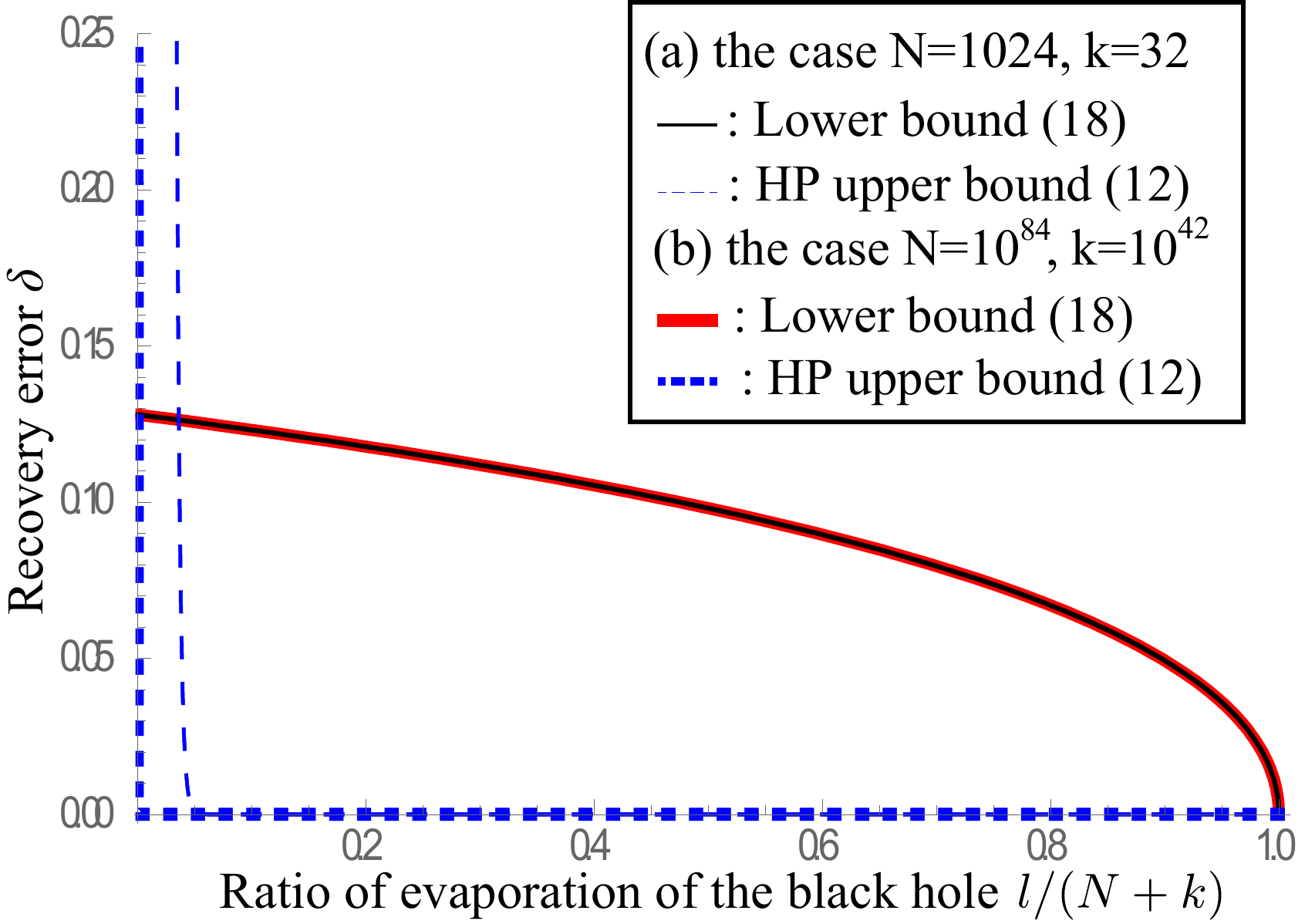}
		\caption{
Comparisons between the lower bound \eqref{doda} and the original Hayden-Preskill (HP) prediction \eqref{mirrorineq}, the former being a lower bound of the error obtained imposing the energy conservation law and the latter being an upper bound of the error without imposing any conservation law. 
Plots (a) and (b) are obtained under settings $(N=1024, k=32)$, and $(N=10^{84}, k=10^{42})$, respectively.
The horizontal axis corresponds to the ratio $l/(N+k)$ quantifying how much the black hole has evaporated after the diary $A$ is thrown in. 
No evaporation has yet occurred when the horizontal axis is 0; when the ratio is 1, the black hole has evaporated entirely.
As seen from the graphs, imposition of the energy conservation law changes the behavior of the recovery error considerably, even when $N\gg k$.
When no conservation laws apply, an almost-complete information recovery is possible when the amount of Hawking radiation emitted is equal to that of the object thrown in. 
In contrast, under the energy conservation law, there is an inevitable error that depends only on the fraction of the black hole evaporated.
The error remains even when the evaporation of the black hole is quite advanced. Here, we use a tighter version of the bound \eqref{SIQ2} for which the right-hand side is twice that of the original \eqref{SIQ2}. (see the footnote \cite{footnoteA2}).}
		\label{dodagraph} 
	\end{figure}

The restriction \eqref{doda} is important in two respects. First, this result differs at a qualitative level from the original prediction \eqref{mirrorineq}, which is given in the absence of energy conservation (see Fig.~\ref{dodagraph}).
The graphs (a) in Fig.~\ref{dodagraph} plot the lower bound \eqref{doda} of the recovery error when the energy conservation law holds, and the upper bound \eqref{mirrorineq} of the error, which holds in the absence of any conservation law.
The graphs deal with the case $k=32$ and $N=1024$.
As seen from the graphs, when there is no conservation law, the recovery of information is instantaneous, whereas a non-ignorable error remains when the energy conservation applies, even when 90 percent of the black hole evaporates.

Second, this result works even when ``Alice's diary'' $A$ is negligibly small compared with the black hole $B$.
Note that the qubit number $N$ of the black hole is considered to correspond to the Bekenstein-Hawking entropy of the black hole, and thus it is a very large number. For this reason, $k=\sqrt{N}$ is usually compatible with $k$ being much smaller than $N$.
For example, the Bekenstein-Hawking entropy of Sagittarius A (the black hole at the center of the Milky Way) is approximately equal to $10^{85}$.
In this case, $\sqrt{N}=10^{42.5}$, and thus \eqref{doda} is valid when $k/N=10^{-42.5}\ll1$.
Even in this scenario, the recovery error remains non-negligible until the 90 percent of the black hole evaporates (see the graph (b) in Fig.~\ref{dodagraph}).
Therefore, the energy conservation law provides a clear limit on the escape of information from a black hole, even if the object thrown into the black hole is much smaller than the black hole itself.

\begin{figure}[tb]
		\centering
		\includegraphics[width=.5\textwidth]{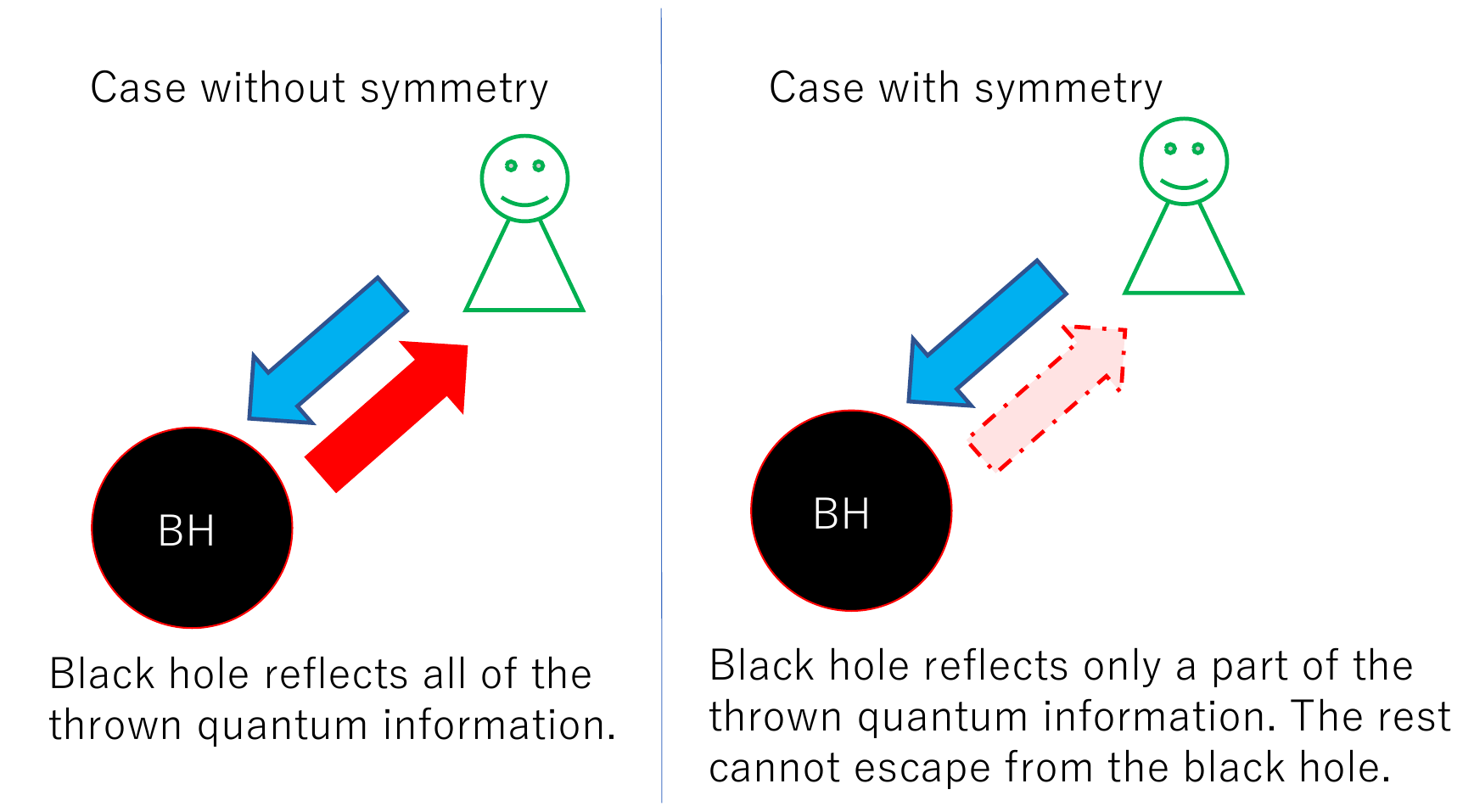}
		\caption{Schematic highlighting the difference between the original HP prediction \eqref{mirrorineq} and our results \eqref{doda} and \eqref{foggyineq}.
The original prediction treats the case of no symmetry and predicts that  to recover the original information within $\delta$, we only have to collect $k+O(\log \delta)$ Hawking radiation particles. Therefore, we can interpret black holes as information mirrors.
According to the bounds \eqref{doda} and \eqref{foggyineq}, when there is a conservation law, the situation \g{changes radically} \cite{footnote-contribution}. 
Both bounds predict that the error $\delta$ remains large, even if one collects much more information than $k$-qubits from Hawking radiation.
The first bound \eqref{doda} guarantees that this effect works even when the black hole is much larger than Alice's diary.
The second bound \eqref{foggyineq} also predicts that when the size of the black hole is comparable to that of the diary, the error cannot be smaller than $const/(1+N/k)$ until the black hole completely evaporates.
In other words, a part of the quantum information is not reflected, and it cannot escape from the black hole.}
		\label{foggyBH} 
	\end{figure}

The above bound restricts general scenarios, including those for which the black hole is much larger than the diary.
When the size of the black hole is comparable to the size of the diary, \eqref{SIQ2} provides another restriction.
Since $\sqrt{F_f}$ is always smaller than the square of the qubit number of $B'$, we obtain
\eq{
\sqrt{\calF_f}\le\gamma (N+k).\label{Ffbound}
}
Combining \eqref{A-bound}, \eqref{D-upp}, and \eqref{Ffbound}, we obtain
\begin{align}
\frac{1-\epsilon}{1+\epsilon}\times\frac{M}{2(N+2k)}\le \delta  \, .
\end{align}
Similar to the derivation of \eqref{doda}, we set $M= k/2$ and obtain the following lower bound of the recovery error:
\begin{align}
\frac{\rm const.}{1+N/2k}\le\delta,\label{foggyineq}
\end{align}
where the constant on the left-hand side is larger than $1/9$.
Unlike the inequality \eqref{doda}, the inequality \eqref{foggyineq} becomes trivial when $N\gg k$.
However, when $N$ is comparable to $k$, the inequality \eqref{foggyineq} gives a non-negligible lower bound for the error $\delta$ that is independent of $l$. The lower bound is valid whenever $l<N+k$ holds. Therefore, when the ratio $N/k$ is not so large, the recovery error cannot be small until $l=N+k$ holds. In other words, the recovery of the quantum information associated with Alice's diary does not finish until the black hole completely evaporates.

Our results \eqref{doda} and \eqref{foggyineq} show that the behavior of the quantum information recovery under the energy conservation law is qualitatively different from the original analysis in the HP model without energy conservation \eqref{mirrorineq} as evident in Fig.~\ref{foggyBH}.

\section{Applications to quantum information processing with symmetry}\label{APPnonsc}

Our \g{formulae} \eqref{SIQ1} and \eqref{SIQ2} are applicable to various phenomena other than scrambling. Below, we apply our bounds to quantum error correction (QEC) and implementation of unitary gates as examples of application.

\subsection{Application 1: quantum error correcting codes with symmetry}

In QEC, we encode quantum information in a {\it logical} system $A$ into a \textit{physical} system $A'$ which is a composite system of $N$ subsystems $\{A'_j\}^{N}_{j=1}$ by an \textit{encoding channel} $\calC$, which is a CPTP map. After the encoding, noise occurs on the physical system $A'$, which is described by a CPTP-map $\calN$.
Finally, we \g{recover} the initial state by performing a recovery CPTP map $\calR$ from $A'$ to $A$.
Then, the recovery error is defined as 
\begin{align}
\delta_C &:=\!\! \min_{\mbox{$\calR $ } \atop (A' \to A ) } \!\!\!
\max_{\rho_{AR_A}}D_F(\rho_{AR_A},\calR\circ\calN\circ\calC(\rho_{AR_A})) \, .
\end{align}

Here we focus on the case where the channel $\calC$ \textit{transversal} with respect to a unitary representation $\{U_{A,t}\}_{t\in \mathbb{R}}$, i.e.
\begin{align}
\calC\circ\calU^{A}_t(...)=\calU^{A'}_t\circ\calC(...), \enskip\forall t\in \mathbb{R},
\end{align}
where $\calU^{\alpha}_t(...)=e^{iX_{\alpha}t} (...) e^{-iX_{\alpha}t}$ ($\alpha=A,A'$) and $X_{A'}$ is described as $X_{A'}:=\sum_{j}X_{A'_j}$ with operators $\{X_{A'_{j}}\}^{N}_{j=1}$ on $A'_{j}$.

\g{The limitations} of the transversal codes is a critical issue \cite{Eastin-Knill,e-EKFaist,e-EKKubica,e-EKZhou,e-EKYang}.
It is shown that the code $\calC$ cannot make $\delta_C=0$ for local noise by the Eastin-Knill theorem \cite{Eastin-Knill}.
Recently, the Eastin-Knill theorem were extended to the cases where $\delta_C$ is finite \cite{e-EKFaist,e-EKKubica,e-EKZhou,e-EKYang}.
These approximate Eastin-Knill theorems show that the size $N$ of the physical system must be inversely proportional to $\delta_C$.

From \eqref{SIQ2}, we can derive a variant of the approximate Eastin-Knill theorem as a corollary (see Supplementary Material Supp.$\rm V$):
\begin{align}
\frac{\calD_{X_A}}{4 \calD_{\max} ( N+\calD_{X_A}/(4\calD_{\max}) )}\le \delta_C \, . 
\label{almostFaist-T}
\end{align}
Here $\calD_{\max}:=\max_{i}\calD_{X_{A'_i}}$. 
Our bounds \eqref{SIQ1} and \eqref{SIQ2} are also applicable to cases where $\calN$ is non-local, and more general covariant codes with general Lie group \g{symmetries} (see Supplementary Materials Supp.$\rm V\hspace{-.1em}I$).

\subsection{Application 2: Implementation of unitary dynamics}\label{APPA}

The last application is on the implementation of the unitary dynamics on the subsystem $A$ through the unitary time-evolution of the isolated total system \cite{TSS,TSS2}. This subject has a long history in the context of the limitation on the quantum computation imposed by conservation laws \cite{ozawa1,ozawa2,Karasawa2007,Karasawa2009,TSS,TSS2}, which is considered as an extension of the Wigner-Araki-Yanase theorem \cite{Wigner1952,Araki-Yanase1960,Yanase1961,OzawaWAY} to quantum computation. Suppose that we try to approximately realize a desired unitary dynamics $U_A$ on a system $A$ as a result of the interaction between another system $B$. We assume that the interaction satisfies a conservation law: $[U,X_A+X_B]=0$. We then define the implementation error $\delta_U$ as:
\begin{align}
\delta_U & :=\max_{\rho_{AR_A}:\mathrm{pure}}D_F(\rho_{AR_A},\mathrm{id}_{R_A}\otimes\calU^{\dagger}_{A}\circ\calE(\rho_{AR_A})).
\end{align}
Here $\calU^{\dagger}_{A}(...):=U^{\dagger}_A(...)U_A$. The quantum operation $\calE$ is the CPTP-map where $A'$ is equal to $A$.
Then, by definition, the inequality $\delta_U\ge\max_{\rho_{AR_A}}\tilde{\delta}\ge\max_{\rho_{AR_A}}\delta$ holds.
Therefore, we can directly apply \eqref{SIQ1'-Smain} and \eqref{SIQ2} to this problem.
In particular, we obtain the following inequality from \eqref{SIQ1'-Smain}:
\begin{align}
  \frac{\SA}{2(\sqrt{\calF_{B}}+4\Delta_{+})}\le \delta_U 
    \label{C-WAY}
\end{align}
This inequality gives a trade-off between the implementation error and the coherence cost of implementation of unitary gates.
The physical message is that the implementation of the desired 
  unitary operator requires the quantum coherence inversely proportional to the square of the implementation error.
We remark that several similar bounds for the coherence cost were already given in Refs. \cite{TSS,TSS2}. However, we stress that \eqref{C-WAY} is given as a corollary of a more general relation \eqref{SIQ1}.
Moreover, as we pointed out several times, our results can be extended to the cases of general Lie group symmetries. In Supplementary Materials \ref{g-symmetry}, we show a generalized version of \eqref{C-WAY} for such cases.

\section{Numerical check of the main inequality \eqref{SIQ2}}\label{numerical check}

So far, we have applied our main result \eqref{SIQ2} to information scrambling and quantum information processing. Our bound works regardless of the size of systems $A$ and $B$ and is especially tight when system B is large. For example, as shown in the previous section, the bounds \eqref{almostFaist-T} and \eqref{C-WAY} obtained from \eqref{SIQ2} become optimal when $\calF$ is very large. 

We next give a numerical check of \eqref{SIQ2} for situations in which the system $B$ is small to show that the bound \eqref{SIQ2} also works rigorously.
For this purpose, we prepare a concrete model.
Let us consider four qubits, $A$, $R_A$, $B_1$, and $R_{B_1}$.
We also take a natural number $b$ and a $b+2$-dimensional system $B_2$.
For $A$, $B_1$, and $B_2$, we define Hermitian operators $X_A$, $X_{B_1}$ and $X_{B_2}$ as follows:
\eq{
X_A&:=\ket{1}\bra{1}_{A},\\
X_{B_1}&:=\ket{1}\bra{1}_{B_1},\\
X_{B_2}&:=\sum^{b+2}_{x=1}2x\ket{x}\bra{x}_{B_2},
}
where $\{\ket{x}_{A}\}_{x=0,1}$, $\{\ket{x}_{B_1}\}_{x=0,1}$, $\{\ket{x}_{B_2}\}_{x=0,...,d+2}$ are orthogonal basis on $A$, $B_1$, and $B_2$, respectively.

\begin{figure}[tb]
		\centering
		\includegraphics[width=.45\textwidth]{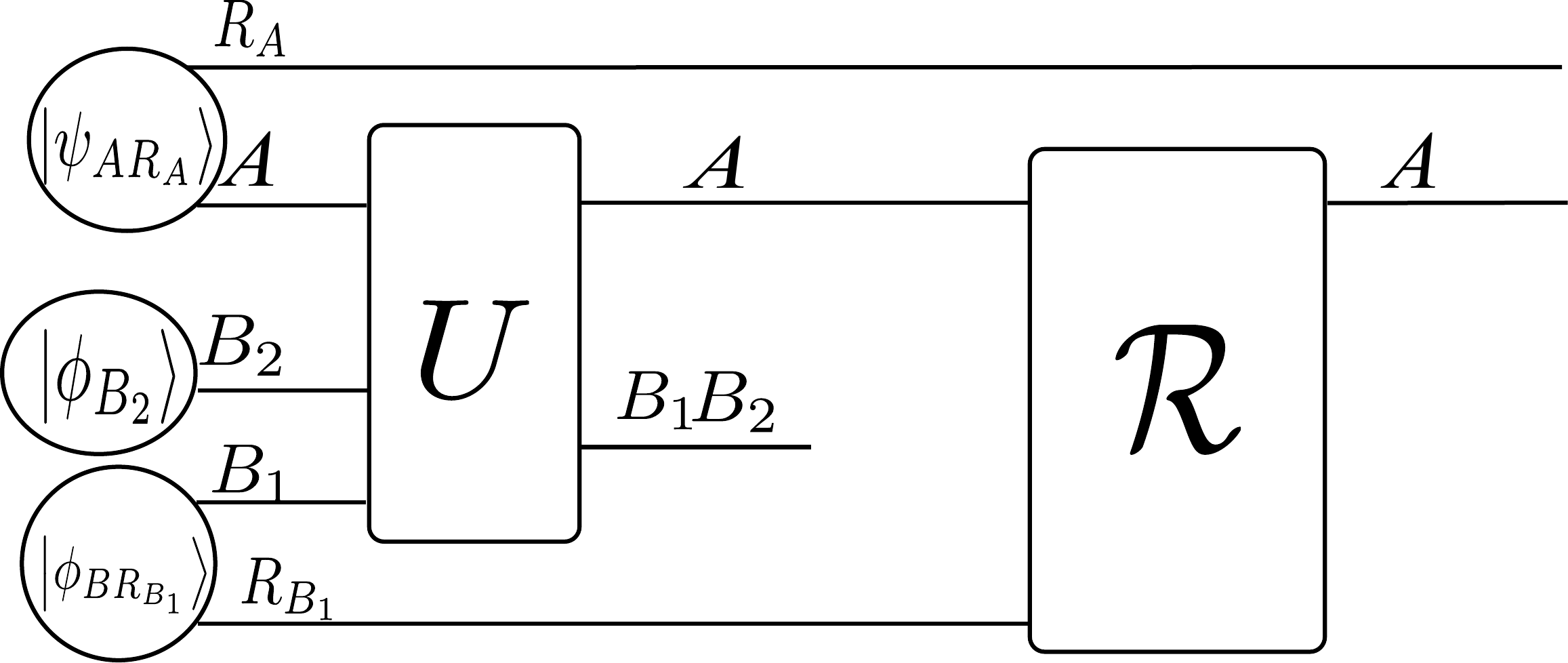}
		\caption{Schematic diagram of the concrete model for numerical check.}
		\label{check_fig}
	\end{figure}

For the above system, we prepare the following initial states:
\eq{
\ket{\psi_{AR_A}}&:=\frac{\ket{00}_{AR_A}+\ket{11}_{AR_A}}{\sqrt{2}},\\
\ket{\phi_{B_1R_{B_1}}}&:=\frac{\ket{00}_{B_1R_{B_1}}+\ket{11}_{B_1R_{B_1}}}{\sqrt{2}},\\
\ket{\phi_{B_2}}&:=\frac{1}{\sqrt{b}}\sum^{b}_{x=1}\ket{x}_{B_2}.
} 
We prepare the following unitary $U_{AB_1B_2}$ on $AB_1B_2$:
\eq{
&U_{AB_1B_2}:=\sum_{1\le k\le b+1}
 \big(\ket{11}\bra{00}_{AB_1}\otimes\ket{k-1}\bra{k}\non
&+\ket{10}\bra{01}_{AB_1}\otimes\ket{k}\bra{k}
+\ket{01}\bra{10}_{AB_1}\otimes\ket{k}\bra{k}\non
&+\ket{00}\bra{11}_{AB_1}\otimes\ket{k}\bra{k-1}\big)\nonumber\\
&+\ket{00}\bra{00}_{AB_1}\otimes\ket{0}\bra{0}
+\ket{01}\bra{01}_{AB_1}\otimes\ket{0}\bra{0}\non
&+\ket{10}\bra{10}_{AB_1}\otimes\ket{0}\bra{0}
+\ket{11}\bra{11}_{AB_1}\otimes\ket{b+1}\bra{b+1}.
}
After a unitary operation $U$, we perform a CPTP map on $AR_{B_1}$ to recover  the initial state $\ket{\psi_{AR_A}}$ on $AR_A$.
Following our framework, the minimum recovery error is defined as follows:
\eq{
\delta=\min_{\calR}D_F(\psi_{AR_A},id_{R_A}\otimes\calR\circ\calE_{A\rightarrow AR_{B_1}}(\psi_{AR_A})).
}
Here $\calE_{A\rightarrow AR_{B_1}}(...):=\Tr_{B_1B_2}[U_{AB_1B_2}(...\otimes\phi_{B_2}\otimes\phi_{BR_{B_{1}}}) U^\dagger_{AB_1B_2}]$.

Since $U_{AB_1B_2}$ conserves $X_{A}+X_{B_1}+X_{B_2}$ and $\rho_A:=\Tr_{R_A}[\psi_{AR_A}]=\frac{\ket{0}\bra{0}+\ket{1}\bra{1}}{2}$, our bound \eqref{SIQ2} is applicable to the above model.
In this case, each term on the right-hand side of \eqref{SIQ2} is evaluated as follows:
\eq{
\calA&\ge1,\label{check_ex1}\\
\Delta_{\max}&\le1,\label{check_ex2}\\
\calF_f&=\frac{b^2+14}{3}.\label{check_ex3}
}
Thus, for each $b$ $(=2,....,7)$, the inequality \eqref{SIQ2} predicts that there is no recovery that can make the recovery error smaller than the black dot in Fig.~\ref{figXXX}.
In Fig.~\ref{figXXX}, we confirm the prediction by numerical calculation.
We randomly generate a CPTP-map from $AR_{B_1}$ to $A$ as a recovery map $10^8$ times for each $b$, and plot the smallest recovery error obtained by the random recoveries (blue point in Fig.~\ref{figXXX}).
Fig.~\ref{figXXX} clearly shows that all of the random recoveries satisfy the bound \eqref{SIQ2}.

\begin{figure}[tb]
		\centering
		\includegraphics[width=0.5\textwidth]{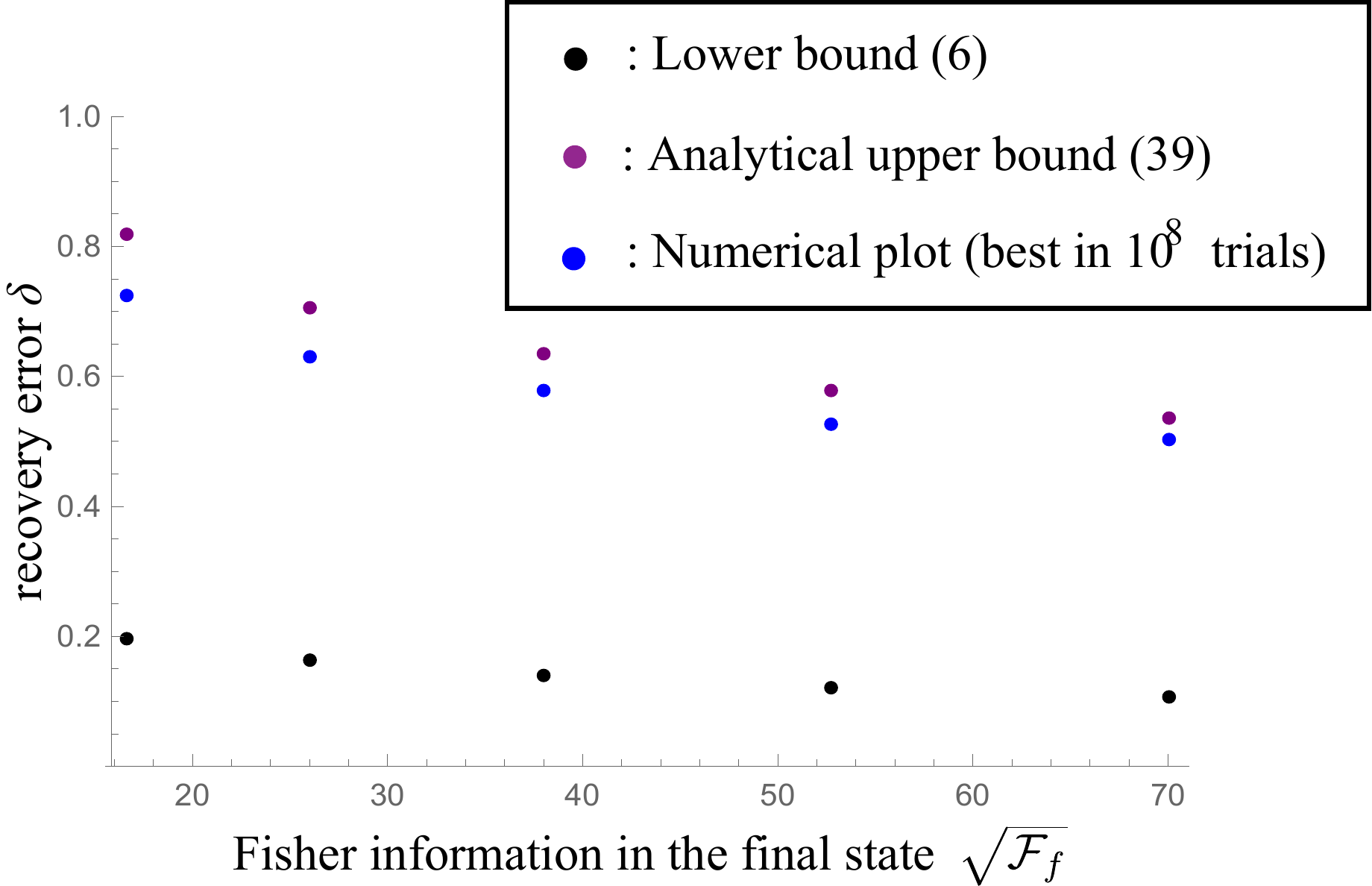}
		\caption{Comparisons between lower bound, analytical upper bound, and numerical upper bound from plots of recovery error versus $\calF_f$. The data points correspond to $b=3,4,5,6,7$. Here, we use a tighter version of the bound \eqref{SIQ2} for which the right-hand side is twice  that of the original \eqref{SIQ2} (see the footnote \cite{footnoteA2}).}
		\label{figXXX}
	\end{figure}

Furthermore, to ensure that a sufficiently large number of trials is obtained, we give the following specific upper bound of the minimum recovery error by extending the \r{A}berg protocol \cite{catalyst} to the model:
\eq{
\delta\le\sqrt{\frac{2}{b}}.\label{upper_ex}
}
Combined with \eqref{check_ex3}, \eqref{upper_ex} gives the following upper bound for each $b$ (purple points in Fig.~\ref{figXXX}):
\eq{
\delta\le\sqrt{\frac{2}{\sqrt{3\calF_f-14}}}.
}
Clearly, the best of the randomly generated recoveries performs better than this upper bound. Therefore, our number of trials is considered sufficient.
In addition, the lower bound \eqref{SIQ2} and the upper bound \eqref{upper_ex} are off by a factor of at most 2 or 3 (Fig. \ref{figXXX}). This suggests that even for a small system, the evaluation of the recovery error by the bound \eqref{SIQ2} is acceptable.

\section{Summary}
In summary, we have clarified fundamental limitations to information recovery from dynamics with general Lie group symmetry. 
As demonstrated in Appendix \ref{AppC}, all results in this paper are given as corollaries of \eqref{SIQ2}. 
It is remarkable that one single inequality \eqref{SIQ2} provides a unifying limit for black holes, quantum error correcting codes and unitary gates.
In particular, the HP model with the energy conservation, some of the information thrown into the black hole cannot escape to the end. 
We also remark that our prediction may be \g{validated in} laboratory experiments that mimic the HP model with symmetry \cite{p1,p2,p3,siddiqi,lab1}.
Moreover, an intriguing topic to consider is the relationship between our relations and a recent argument on the weak violation of global symmetries in quantum gravity \cite{Banks-Seiberg,Harlow-Ooguri1,Harlow-Ooguri2}.
The effect of symmetry on the OTOC decay is another intresting future direction from our results.


\begin{acknowledgments}
	The present work was supported by JSPS Grants-in-Aid for Scientific Research No. JP19K14610 (HT), No. JP22H05250 (HT), No. JP25103003 (KS), and No. JP16H02211 (KS), and JST PRESTO No. JPMJPR2014 (HT), and JST MOONSHOT (H. T. Grant No. JPMJMS2061). 
We thank Richard Haase, PhD, from Edanz (https://jp.edanz.com/ac) for editing a draft of this manuscript.
\end{acknowledgments}

\appendix

\section{Tips for resource theory of asymmetry and quantum Fisher information}\label{AppA}
For \g{convenience, we discuss} the resource theory of asymmetry and the quantum Fisher information briefly.
The resource theory of asymmetry is a resource theory \cite{Bartlett,Gour,Marvian,Marvian-thesis,Marvian2018,Lostaglio2018,TSS,TSS2,TN-WAY,Takagi2018,Marvian distillation} \g{that handles the symmetries of the} dynamics.
In the main text, we consider the simplest case for which the symmetry is $R$ or $U(1)$ and the dynamics obeys a conservation law.
\g{More} general cases are introduced in Supplementary Material Supp.I.

We firstly introduce covariant operations, which are free operations of the resource theory of asymmetry.
If a CPTP map $\calC$ from $S$ to $S'$ and Hermite operators $X_S$ and $X_{S'}$ on $S$ and $S'$ satisfy the following relation, we call $\calC$ a \textit{covariant} operation with respect to $X_S$ and $X_{S'}$:
\begin{align}
\calC(e^{iX_St}...e^{-iX_St})=e^{iX_{S'}t}\calC(...)e^{-iX_{S'}t},\enskip\forall t.\label{covariantcond}
\end{align}
A very important property of covariant operations is that we can implement any covariant operation using a unitary operation satisfying a conservation law and a quantum state which commutes with the conserved quantity.
To be specific, let us \g{consider} a covariant operation $\calC$ with respect to $X_S$ and $X_{S'}$.
Then, there exist quantum systems $E$ and $E'$ satisfying $SE=S'E'$,  Hermite operators $X_{E}$ and $X_{E'}$ on $E$ and $E'$, a unitary operation $U$ on $SE$ satisfying $U(X_S+X_{E})U^{\dagger}=X_{S'}+X_{E'}$, and a \textit{symmetric} state $\mu_E$ on $E$ satisfying $[\mu_{E},X_E]=0$ such that \cite{Marvian distillation}
\begin{align}
\calC(...)=\Tr_{E'}[U(...\otimes\mu_E) U^{\dagger}].
\end{align}

The $SLD$-quantum Fisher information for the family $\{e^{-iXt}\rho e^{iXt}\}_{t\in\mathbb{R}}$, described as $\calF_{\rho_S}(X_S)$, is a standard resource measure in the resource theory of asymmetry \cite{Takagi2018, Marvian distillation}.
It is also known as a standard measure of \g{quantum fluctuation}, since it is related to the variance $V_{\rho_S}(X_S):=\ex{X^2_S}_{\rho_S}-\ex{X_S}^2_{\rho_S}$ as follows \cite{Q-Fisher=Q-fluctuation1, Q-Fisher=Q-fluctuation2, Marvian distillation}:
\begin{align}
\calF_{\rho_S}(X_S)&=4\min_{\{q_i,\phi_i\}}\sum_{i}q_iV_{\phi_i}(X_S)\label{F=4V}\\
&=4\min_{\ket{\Psi_{SR}},X_R}V_{\Psi_{SR}}(X_S+X_R)\label{F=4V2}
\end{align}
where $\{q_i,\phi_i\}$ runs over the ensembles satisfying $\rho=\sum_{i}q_i\phi_i$ and each $\phi_i$ is pure, and $\{\ket{\Psi_{SR}},X_R\}$ runs over purifications of $\rho_S$ and Hermitian operators on $R$.
The equality of \eqref{F=4V} shows that $\calF_{\rho}(X)$ is the minimum average of the fluctuation arising from quantum superposition.
Note that it also means that if $\rho$ is pure, $\calF_{\rho}(X)=4V_{\rho}(X)$ holds.
The $\ket{\Psi_{SR}}$ and $X_{R}$ achieving the minimum of $V_{\Psi_{SR}}(X_S+X_R)$ in \eqref{F=4V2} are $\ket{\Psi_{SR}}:=\sum_{l}\sqrt{r_l}\ket{l_S}\ket{l_R}$ and 
\begin{align}
X_{R}:=\sum_{ll'}\frac{2\sqrt{r_{l}r_{l'}}}{r_{l}+r_{l'}}\bra{l_{S}}X_{S}\ket{l_{S}}\ket{l'_{R}}\bra{l_{R}},\label{bestXR}
\end{align}
where $\{r_l\}$ and $\{\ket{l_{S}}\}$ denote the eigenvalues and eigenvectors of $\rho_S$ \cite{Marvian distillation}.

\section{Note on entanglement fidelity and average gate fidelity}\label{AppB}

In this subsection, we show that the recovery error $\delta$ can approximate the average of the recovery error which is averaged thorough pure states on the entire Hilbert space of $A$ or on its subspace using special initial states as $\ket{\psi_{AR_A}}$ \cite{3Horodecki}.

For explanation, let us introduce the average fidelity and the entanglement fidelity.
For a CPTP map $\calC$ from a quantum state $Q$ to $Q$,  these two quantities are defined as follows:
\begin{align}
F^{(2)}_{\mathrm{avg}}(\calC)&:=\int d\psi_Q F(\ket{\psi_Q},\calC(\psi_Q))^2,\\
F^{(2)}_{\mathrm{ent}}(\calC)&:=F(\ket{\psi_{QR_Q}},1_{R_Q}\otimes\calE(\psi_{QR_Q}))^2,
\end{align}
where $\ket{\psi_{QR_Q}}$ is a maximally entangled state  between $Q$ and $R_Q$, and the integral is taken with the uniform (Haar) measure on the state space of $Q$.
For these two quantities, the following relation is known \cite{3Horodecki}:
\begin{align}
F^{(2)}_{\mathrm{avg}}(\calC)=\frac{d_QF^{(2)}_{\mathrm{ent}}(\calC)+1}{d_Q+1}.\label{avg-ent}
\end{align}
Let us take a subspace $\calS$ of the state space of $A$, and define the following average recovery error:
\begin{align}
&\delta^{(2)}_{\mathrm{avg},\calS}
:=\nonumber\\
&\min_{\mbox{$\calR$ on $A'R_B$}}\int_{\calS} d\psi_AD_F(\ket{\psi_{A}},\calR(\Tr_{B'}U(\psi_{A}\otimes\phi_{BR_B})U^\dagger))^2.
\end{align}
Then, due to \eqref{avg-ent}, when we set $\ket{\psi_{AR_A,\calS}}=\frac{\sum_{i}\ket{i}_A\ket{i}_{R_A}}{\sqrt{d_{\calS}}}$ where $\{\ket{i}_A\}$ is an arbitrary orthonormal basis of $\calS$ and $d_{\calS}$ is the dimension of $\calS$, the recovery error $\delta_{\calS}:=\delta(\ket{\psi_{AR_A,\calS}},\ket{\phi_{BR_B}},U)$ satisfies the following relation:
\begin{align}
\delta^{(2)}_{\mathrm{avg},\calS}=\frac{d_{\calS}}{d_\calS+1}\delta^2_{\calS}.
\end{align}
Therefore, when we use a maximally entangled state between a subspace of $A$ and $R_B$ as $\ket{\psi_{AR_A}}$, the recovery error $\delta$ for the $\ket{\psi_{AR_A}}$ approximates the average of recovery error which is averaged through all pure states of the subspace of $A$.

\begin{figure*}[tb]
		\centering
		\includegraphics[width=.8\textwidth]{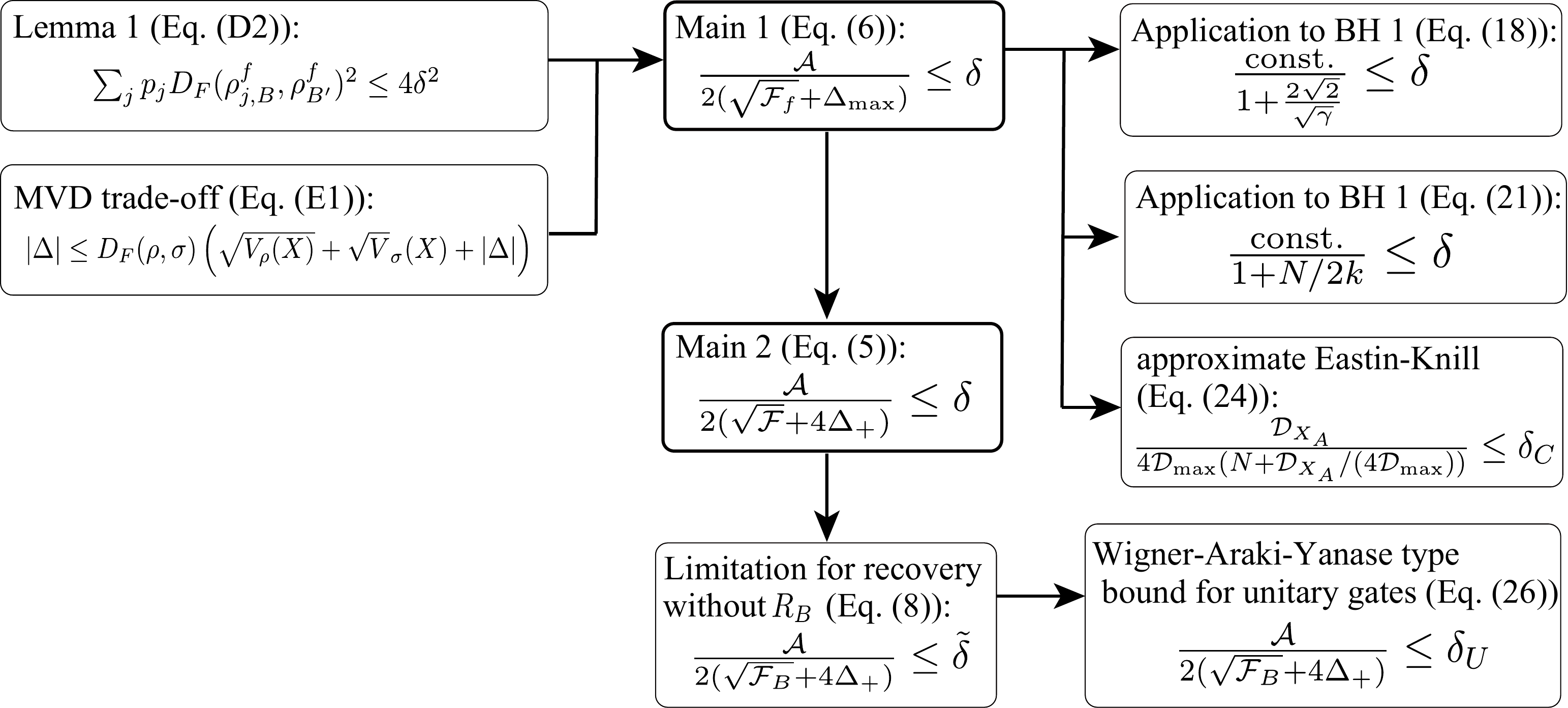}
		\caption{Schematic depicting the relationship between the main results and applications.}
		\label{methods_flow} 
	\end{figure*}

\section{Relations between main results and applications in this paper}\label{AppC}
Next, we show the relation between the main results and applications in this paper (Fig.~\ref{methods_flow}).
We derive \eqref{SIQ2} from two lemmas which we give in the next two subsections.
All of the physical results in this paper including \eqref{SIQ1} and \eqref{foggyineq} are given as corollaries of \eqref{SIQ2}.
In that sense, \eqref{SIQ2} is a universal restriction on information recovery from dynamics with Lie group symmetry.
In addition to what is described in the main text, various results can be given in a similar way. 
For instance, we can derive the Wigner-Araki-Yanase theorem for unitary gates from \eqref{SIQ1'-Smain}. 
We also derive another restriction on HP model with symmetry from \eqref{SIQ1}.

We remark that there exist several variations and generalizations of the results in Fig.~\ref{methods_flow}.
For instance, in Appendix \ref{tighterSIQ}, we derive tighter variations of \eqref{SIQ1} and \eqref{SIQ2}.
We also extend \eqref{SIQ1} and \eqref{SIQ2} to general Lie group symmetries in Supplementary Material Supp.$\rm V\hspace{-.1em}I$.

\section{Trade-off relation between irreversibility and back-reaction}

In the derivation of \eqref{SIQ1} and \eqref{SIQ2}, we use the following lemma:
\begin{lemma}\label{SCL}
In the setup of \g{Section} 2, let us \g{consider} an arbitrary decomposition of the initial state of $A$ as $\rho_A=\sum_{j}p_j\rho_j$.
We also refer to the final states of $B'$ for the cases where the initial states of $A$ are $\rho_j$ and $\rho_A$ as $\rho^{f}_{j,B'}$ and $\rho^{f}_{B'}$, respectively.
Namely, $\rho^{f}_{j,B'}:=\Tr_{A'}[U(\rho_j\otimes\rho_B)U^\dagger]$ and $\rho^{f}_{B'}:=\Tr_{A'}[U(\rho_A\otimes\rho_B)U^\dagger]$ where $\rho_B:=\Tr_{R_B}[\rho_{BR_B}]$.
Then, there exists a state $\sigma_{B'}$ such that
\begin{align}
\sum_{j}p_jD_{F}(\rho^{f}_{j,B'},\sigma_{B'})^2\le\delta^2.\label{SCL-S1}
\end{align}
Moreover, the following inequality holds:
\begin{align}
\sum_{j}p_jD_{F}(\rho^{f}_{j,B'},\rho^{f}_{B'})^2\le4\delta^2.\label{SCL-S2}
\end{align}
\end{lemma}
Lemma \ref{SCL} holds even when $U(X_A+X_B)U^\dagger\ne X_{A'}+X_{B'}$.
The proof of this lemma is given in Appendix \ref{S-SCL}.
Roughly speaking, this lemma means that when the recovery error $\delta$ is small (i.e. the realized CPTP map $\calE$ is approximately reversible), then the final state of $B'$ becomes almost independent of the initial state of $A$.

This lemma is a generalized version of (16) in Ref. \cite{TSS} and Lemma 3 in Ref. \cite{TSS2}. The original lemmas are given for the implementation error of unitary gates, and used for lower bounds of resource costs to implement desired unitary gates in the resource theory of asymmetry \cite{TSS,TSS2} and in the general resource theory \cite{TT}.

\section{mean-variance-distance trade-off relation}

For an arbitrary Hermite operator $X$ and arbitrary states $\rho$ and $\sigma$, there is a trade-off relation between the difference of expectation values $\Delta:=\ex{X}_{\rho}-\ex{X}_{\sigma}$, the variances $V_{\rho}(X)$ and $V_{\sigma}(X)$, and the distance between $\rho$ and $\sigma$ \cite{Katsube}:
\begin{align}
|\Delta|\le D_F(\rho,\sigma)(\sqrt{V_{\rho}(X)}+\sqrt{V_{\sigma}(X)}+|\Delta|),\label{RCR2}
\end{align}
This is an improved version of the original inequality (15) in Ref. \cite{TSS}.
In the original inequality, the purified distance $D_F(\rho,\sigma)$ is replaced by the Bures distance $L(\rho,\sigma):=\sqrt{2(1-F(\rho,\sigma))}$.
\g{These} inequalities mean that if two states have different expectation values and are close to each other, then at least one of the two states \g{exhibits} large fluctuation.

\section{Properties of variance and expectation value of the conserved quantity $X$}

We use several properties of variance and expectation value of the conserved quantity $X$.
In our setup described in Section \ref{sec2}, we have assumed that the unitary dynamics $U$ satisfies the conservation law of $X$: $U(X_A+X_B)U^\dagger=X_{A'}+X_{B'}$.
Under this assumption, for arbitrary states $\xi_A$ and $\xi_B$ on $A$ and $B$, the following relations hold:
\begin{align}
\ex{X_A}_{\xi_A}-\ex{X_{A'}}_{\xi^{f}_{A'}}=&\ex{X_{B'}}_{\xi^{f}_{B'}}-\ex{X_{B}}_{\xi_{B}}.\label{Eeq}\\
\sqrt{V_{\xi^{f}_{B'}}(X_{B'})}\le&\sqrt{V_{\xi^f_{A'}}(X_{A'})}\non
&+\sqrt{V_{\xi_A}(X_A)}+\sqrt{V_{\xi_B}(X_B)}\nonumber\\
\le&\sqrt{V_{\xi_{B}}(X_{B})}+\Delta_{+},\label{Vineq}\\
\sqrt{V_{\xi_{B}}(X_{B})}\le&\sqrt{V_{\xi^f_{A'}}(X_{A'})}\non
&+\sqrt{V_{\xi_A}(X_A)}+\sqrt{V_{\xi^f_{B'}}(X_{B'})}\nonumber\\
\le&\sqrt{V_{\xi^{f}_{B'}}(X_{B'})}+\Delta_{+},\label{Vineq2-S}
\end{align}
where $\xi^{f}_{A'}:=\calE(\xi_A)=\Tr_{B'}[U(\xi_A\otimes\xi_B)U^{\dagger}]$ and $\xi^{f}_{B'}:=\Tr_{A'}[U(\xi_A\otimes\xi_B)U^{\dagger}]$.
We show these two relations in Appendix \ref{S-VE}.

\section{Derivation of the limitations of information recovery error (case of single conserved quantity)}\label{derivation_main}

Combining the above three methods, we can derive our main results \eqref{SIQ1} and \eqref{SIQ2}.
We first decompose $\rho_A=\sum_{j}p_{j}\rho_{j}$ such that $\SA=\sum_{j}p_j|\Delta_j|$. 
Then, due to \eqref{Eeq}, we obtain
\begin{align}
|\Delta_j|=|\ex{X_{B'}}_{\rho^{f}_{j,B'}}-\ex{X_{B'}}_{\rho^{f}_{B'}}|.\label{Eeq2}
\end{align}
Now, we derive \eqref{SIQ2} as follows:
\begin{align}
\SA&\stackrel{(a)}{=}\sum_{j}p_j|\ex{X_{B'}}_{\rho^{f}_{j,B'}}-\ex{X_{B'}}_{\rho^{f}_{B'}}|\nonumber\\
&\stackrel{(b)}{\le}\sum_{j}p_jD_{F}(\rho^{f}_{j,B'},\rho^{f}_{B'})\non
&\times\left(\sqrt{V_{\rho^{f}_{j,B'}}(X_{B'})}+\sqrt{V_{\rho^{f}_{B'}}(X_{B'})}+|\Delta_{j}|\right)\nonumber\\
&\stackrel{(c)}{\le}\sqrt{\sum_{j}p_jD_{F}(\rho^{f}_{j,B'},\rho^{f}_{B'})^2}\sqrt{\sum_jp_jV_{\rho^{f}_{j,B'}}(X_{B})}\nonumber\\
&+2\delta\left(\sqrt{V_{\rho^{f}_{B'}}(X_{B'})}+\Delta_{\max}\right)\nonumber\\
&\stackrel{(d)}{\le}2\delta\left(2\sqrt{V_{\rho^{f}_{B'}}(X_{B'})}+\Delta_{\max}\right)\nonumber\\
&\stackrel{(e)}{=}2\delta\left(\sqrt{\calF_f}+\Delta_{\max}\right).
\end{align}
Here we use \eqref{Eeq2} in (a), \eqref{RCR2} in (b), the Cauchy-Schwartz inequality, Lemma \ref{SCL} and $|\Delta_j|\le\Delta_{\max}$ in (c), Lemma \ref{SCL} and the concavity of the variance in (d), and $\calF_f=4V_{\rho^{f}_{B'}}(X_{B'})$ in (e).

We also derive \eqref{SIQ1} from \eqref{SIQ2}:
\begin{align}
\SA&\le2\delta\left(\sqrt{\calF_f}+\Delta_{\max}\right)\nonumber\\
&\stackrel{(a)}{=}2\delta\left(2\sqrt{V_{\rho^{f}_{B'}}(X_{B'})}+\Delta_{\max}\right)\nonumber\\
&\stackrel{(b)}{\le}2\delta\left(2\sqrt{V_{\rho_{B}}(X_{B})}+4\Delta_{+}\right)\nonumber\\
&\stackrel{(c)}{=}2\delta\left(\sqrt{\calF}+4\Delta_{+}\right).
\end{align}
Here we use  $\calF_f=4V_{\rho^{f}_{B'}}(X_{B'})$ in (a), \eqref{Vineq} in (b), and $\calF=4V_{\phi_{BR_B}}(X_B\otimes1_{R_B})=4V_{\rho_{B}}(X_{B})$ in (c).

\subsection{The case where $\rho_A=(\rho_0+\rho_1)/2$ is possible}\label{tighterSIQ}
Let us consider the case that $\ket{\psi_{AR_A}}$ satisfies $\rho_A=(\rho_0+\rho_1)/2$ with some proper density matrices $\rho_0$ and $\rho_1$.
In this case, we can define a variation of $\calA$ as follows:
\begin{align}
\SA_2:=\max_{\rho_0,\rho_1}\sum^{1}_{j=0}\frac{1}{2}|\Delta_j|.
\end{align}
where $\{\rho_0,\rho_1\}$ runs over $\rho_A=\frac{\rho_0+\rho_1}{2}$.
For $\SA_2$, we can obtain the following relations:
\begin{align}
\frac{\SA_2}{\sqrt{\calF}+4\Delta_{+}}&\le\delta,\label{R-SIQ1}\\
\frac{\SA_2}{\sqrt{\calF_f}+\Delta_{\max}}&\le\delta\label{R-SIQ2}.
\end{align}

Let us prove \eqref{R-SIQ1} and \eqref{R-SIQ2}.
We can derive \eqref{R-SIQ1} from \eqref{R-SIQ2} in the same manner as the derivation \eqref{SIQ1} from \eqref{SIQ2}.
Therefore, we only have to prove \eqref{R-SIQ2}.
From \eqref{SCL-S1}, we obtain
\begin{align}
\sum^{1}_{j=0}\frac{1}{2}D_{F}(\rho^{f}_{j,B'},\sigma_{B'})^2\le\delta^2.
\end{align}
Therefore, using $(x+y)^2/4\le(x^2+y^2)/2$ and the triangle inequality for $D_F$, we obtain
\begin{align}
D_{F}(\rho^{f}_{0,B'},\rho^{f}_{1,B'})\le2\delta.\label{D2decom}
\end{align}
Let us take a decomposition $\rho_A=\frac{\sum^{1}_{j=0}\rho_j}{2}$ satisfying $\SA_2=\sum^{1}_{j=0}\frac{1}{2}|\Delta_{j}|$.
Then, due to \eqref{Eeq}, we obtain the following relation for both $j=0$ and $j=1$:
\begin{align}
|\Delta_j|&=|\ex{X_{B'}}_{\rho^{f}_{j,B'}}-\ex{X_{B'}}_{\rho^{f}_{B'}}|\nonumber\\
&=|\ex{X_{B'}}_{\rho^{f}_{j,B'}}-\ex{X_{B'}}_{\frac{\rho^{f}_{0,B'}+\rho^{f}_{1,B'}}{2}}|\nonumber\\
&=\frac{|\ex{X_{B'}}_{\rho^{f}_{0,B'}}-\ex{X_{B'}}_{\rho^{f}_{1,B'}}|}{2}\label{A2decom}
\end{align}

Then, we derive \eqref{R-SIQ1} as follows:
\begin{align}
\SA_2&\stackrel{(a)}{=}\sum^{1}_{j=0}\frac{|\ex{X_{B'}}_{\rho^{f}_{0,B'}}-\ex{X_{B'}}_{\rho^{f}_{1,B'}}|}{4}\nonumber\\
&=\frac{1}{2}|\ex{X_{B'}}_{\rho^{f}_{0,B'}}-\ex{X_{B'}}_{\rho^{f}_{1,B'}}|\nonumber\\
&\stackrel{(b)}{\le} \frac{1}{2}D_{F}(\rho^{f}_{0,B'},\rho^{f}_{1,B'})\non
&\times\left(\sqrt{V_{\rho^{f}_{0,B'}}(X_{B'})}+\sqrt{V_{\rho^{f}_{1,B'}}(X_{B'})}+\Delta_{0,1}\right)\nonumber\\
&\stackrel{(c)}{\le} \delta \left(\sqrt{V_{\rho^{f}_{0,B'}}(X_{B'})}+\sqrt{V_{\rho^{f}_{1,B'}}(X_{B'})}+\Delta_{\max}\right)\nonumber\\
&\stackrel{(d)}{\le}\delta \left(2\sqrt{\frac{V_{\rho^{f}_{0,B'}}(X_{B'})+V_{\rho^{f}_{1,B'}}(X_{B'})}{2}}+\Delta_{\max}\right)\nonumber\\
&\stackrel{(e)}{\le}\delta \left(2\sqrt{V_{\rho^{f}_{B'}}(X_{B'})}+\Delta_{\max}\right)\nonumber\\
&\stackrel{(f)}{\le}\delta \left(\sqrt{\calF_f}+\Delta_{\max}\right),
\end{align}
Here, we use \eqref{A2decom} in (a), \eqref{RCR2} in (b), \eqref{D2decom} and $\Delta_{0,1}:=|\ex{X_{B'}}_{\rho^{f}_{0,B'}}-\ex{X_{B'}}_{\rho^{f}_{1,B'}}|\le\Delta_{\max}$ in (c), $\sqrt{x}+\sqrt{y}\le2\sqrt{(x+y)/2}$ in (d), the concavity of the variance in (e), and $\calF_f=4V_{\rho^{f}_{B'}}(X_{B'})$ in (f).

\begin{figure*}[tb]
		\centering
		\includegraphics[width=0.85\textwidth]{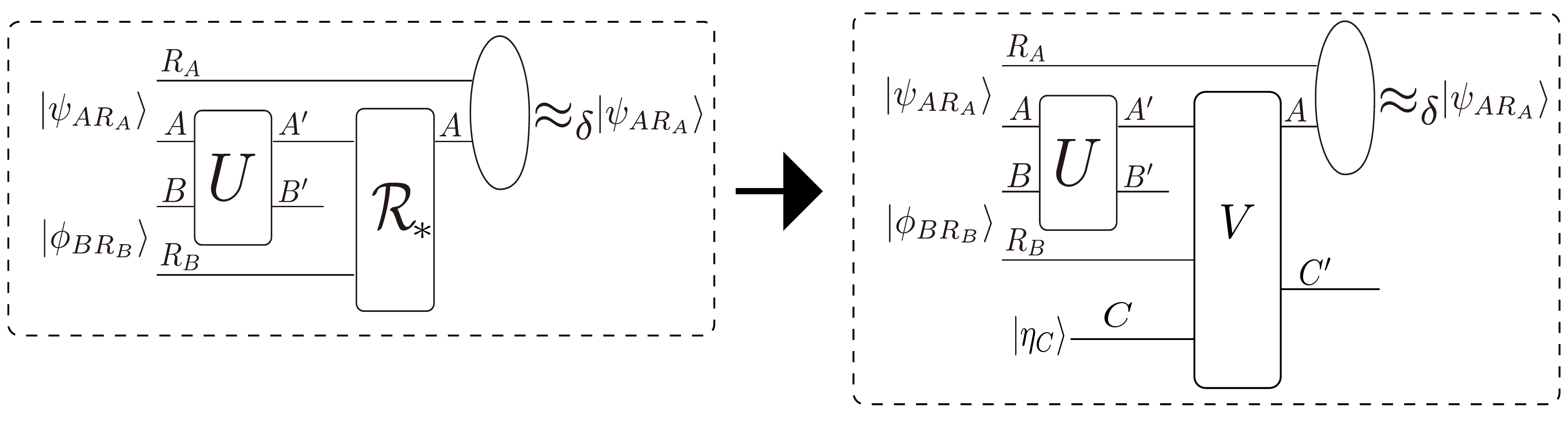}
		\caption{}
		\label{SCLproof}
	\end{figure*}

\section{Derivation of Lemma \ref{SCL}}\label{S-SCL}

In this section, we prove Lemma \ref{SCL}

\begin{proofof}{Lemma \ref{SCL}}
We denote by $\calR_{*}$ the best recovery operation which achieves $\delta$ and take its Steinspring representation $(V,\ket{\eta_C})$ (Fig.~\ref{SCLproof}).
Here, $V$ is a unitary operation on $A'R_BC$, and $\ket{\eta_C}$ is a pure state on $C$.
Since $\calR_*$ is a CPTP-map from $A'R_B$ to $A$, we can take another system $C'$ satisfying $A'R_BC=AC'$.
We refer to the initial and final state of the total system as $\ket{\psi_{\tot}}$ and $\ket{\psi^{f}_{\tot}}$. Then, these two states are expressed as follows:
\begin{align}
\ket{\psi_{\tot}}&:=\ket{\psi_{AR_A}}\otimes\ket{\phi_{BR_B}}\otimes\ket{\eta_C},\\
\ket{\psi^{f}_{\tot}}&:=(1_{R_A}\otimes V\otimes 1_{B'})(1_{R_A}\otimes U\otimes 1_{R_BC})\ket{\psi_{\tot}}
\end{align}
Due to the definitions of $\delta$ and $\calR_*$, for $\psi^f_{AR_A}:=\Tr_{B'C'}[\psi^f_{\tot}]$,
\begin{align}
D_F(\psi^f_{AR_A},\ket{\psi_{AR_A}})=\delta.
\end{align}
Therefore, due to the Uhlmann theorem and the fact that $\ket{\psi_{AR_A}}$ is pure, there exists a pure state $\ket{\phi^f_{B'C'}}$ such that
\begin{align}
D_F(\ket{\psi^f_{\tot}},\ket{\psi_{AR_A}}\otimes\ket{\phi^f_{B'C'}})=\delta.
\end{align}
Since the purified distance $D_F$ is not increased by the partial trace, we obtain
\begin{align}
D_F(\psi^f_{B'C'},\ket{\phi^f_{B'C'}})\le\delta.\label{star-1}
\end{align}
where $\psi^f_{B'C'}:=\Tr_{AR_A}[\psi^f_{\tot}]$.
Let us define $\sigma_{B'}$ as $\sigma_{B'}:=\Tr_{C'}[\phi^f_{B'C'}]$.
Then, due to $\Tr_{C'}[\psi^{f}_{B'C'}]=\rho^{f}_{B'}$ and \eqref{star-1},
\begin{align}
D_F(\rho^{f}_{B'},\sigma_{B'})\le\delta.\label{star0}
\end{align}

Here, we assume that there are states $\{\tilde{\psi}^f_{j,B'C'}\}$ on $B'C'$ such that
\begin{align}
\psi^{f}_{B'C'}&=\sum_{j}p_j\tilde{\psi}^f_{j,B'C'},\label{star1}\\
\Tr_{C'}[\tilde{\psi}^f_{j,B'C'}]&=\rho^{f}_{j,B'}.\label{star2}
\end{align}
Below, we firstly prove \eqref{SCL-S1} and \eqref{SCL-S2} under the assumption of the existence of $\{\tilde{\psi}^f_{j,B'C'}\}$.
We demonstrate the existence of $\{\tilde{\psi}^f_{j,B'C'}\}$ at the end of the proof.

Combining \eqref{star-1} and \eqref{star1}, we obtain 
\begin{align}
D_F(\sum_{j}p_j\tilde{\psi}^f_{j,B'C'},\ket{\phi^f_{B'C'}})\le\delta.
\end{align}
From $D_F(\rho,\sigma)=\sqrt{1-F(\rho,\sigma)^2}$ and $F(\rho,\ket{\phi})^2=\bra{\phi}\rho\ket{\phi}$, we obtain
\begin{align}
1-\delta^2&\le\sum_{j}p_j\bra{\phi^f_{B'C'}}\tilde{\psi}^f_{j,B'C'}\ket{\phi^f_{B'C'}}\nonumber\\
&=1-\sum_{j}p_jD_F(\tilde{\psi}^f_{j,B'C'},\ket{\phi^f_{B'C'}})^2.\label{star3pre}
\end{align}
Due to \eqref{star2}, \eqref{star3pre} and the monotonicity of $D_F$, we obtain the \eqref{SCL-S1}:
\begin{align}
\sum_jp_jD_F(\rho^f_{j,B'},\sigma_{B'})^2\le\delta^2.\label{star3}
\end{align}
Since the root mean square is greater than the average, we also obtain
\begin{align}
\sum_jp_jD_F(\rho^f_{j,B'},\sigma_{B'})\le\delta.\label{star4}
\end{align}
Since the purified distance satisfies the triangle inequality \cite{Tomamichel}, we obtain \eqref{SCL-S2} as follows:
\begin{align}
\sum_jp_jD_F(\rho^f_{j,B'},\rho^f_{B'})^2&\le\sum_jp_j(D_F(\rho^f_{j,B'},\sigma_{B'})\nonumber\\
&\enskip\enskip\enskip\enskip
\enskip\enskip\enskip\enskip\enskip+D_F(\sigma_{B'},\rho^f_{B'}))^2\nonumber\\
&\stackrel{(a)}{\le}\sum_{j}p_j(D_F(\rho^f_{j,B'},\sigma_{B'})+\delta)^2\nonumber\\
&\stackrel{(b)}{\le}4\delta^2.
\end{align}
Here we use \eqref{star0} in (a) and \eqref{star3} and \eqref{star4} in (b).

Finally, we show the existence of $\{\tilde{\psi}^f_{j,B'C'}\}$ satisfying \eqref{star1} and \eqref{star2}.
We firstly take a partial isometry $W_{R_A}$ from $R_A$ to $R'_{A1}R'_{A2}$ such that
\begin{align}
1_A\otimes W_{R_A}\ket{\psi_{AR_A}}&=\sum_{j}\sqrt{p_j}\ket{\psi_{j,AR'_{A1}}}\otimes\ket{j_{R'_{A2}}},\\
1_A\otimes W^{\dagger}_{R_A}W_{R_A}\ket{\psi_{AR_A}}&=\ket{\psi_{AR_A}}.\label{WRA}
\end{align}
Here $\{\ket{j_{R'_{A2}}}\}$ are orthonormal and $\ket{\psi_{j,AR'_{A1}}}$ is a purification of $\rho_j$. We abbreviates $R'_{A1}R'_{A2}$ as $R'_{A}$.
The existence of $W_{R_A}$ is guaranteed as follows.
We firstly note that there exists a ``minimal'' purification $\ket{\psi_{AR^{*}_A}}$ of $\rho_A$, for which we can take isometries $W^{(1)}$ from $R^{*}_A$ to $R_A$ and $W^{(2)}$ from $R^{*}_A$ to $R'_{A}$ such that \cite{Paulsen}
\begin{align}
(1_{A}\otimes W^{(1)})\ket{\psi_{AR^{*}_A}}&=\ket{\psi_{AR_A}},\\
(1_{A}\otimes W^{(2)})\ket{\psi_{AR^{*}_A}}&=\sum_{j}\sqrt{p_j}\ket{\psi_{j,AR'_{A1}}}\otimes\ket{j_{R'_{A2}}}.
\end{align}
The desired $W_{R_A}$ is defined as $W_{R_A}:=W^{(2)}W^{(1)\dagger}$.
Since $W^{(2)}$ and $W^{(1)}$ are isometry, $W_{R_A}$ is a partial isometry.
Using $W^{(2)\dagger}W^{(2)}=W^{(1)\dagger}W^{(1)}=1_{R^{*}_{A}}$, we  obtain \eqref{WRA} as follows:
\begin{align}
&1_A\otimes W^{\dagger}_{R_A}W_{R_A}\ket{\psi_{AR_A}}\non
&=1_A\otimes W^{(1)}W^{(2)\dagger}W^{(2)}W^{(1)\dagger}\ket{\psi_{AR_A}}\nonumber\\
&=1_A\otimes W^{(1)}W^{(2)\dagger}W^{(2)}W^{(1)\dagger}W^{(1)}\ket{\psi_{AR^{*}_A}}\nonumber\\
&=\ket{\psi_{AR_A}}.
\end{align}

Since the partial isometry $W_{R_A}$ works only on $R_A$, we obtain 
\begin{align}
&(W_{R_A}\otimes1_{AB'C'})(1_{R_A}\otimes V\otimes 1_{B'})(1_{R_A}\otimes U\otimes 1_{R_BC})\non
&=
(1_{R'_A}\otimes V\otimes 1_{B'})(1_{R'_A}\otimes U\otimes 1_{R_BC})(W_{R_A}\otimes1_{ABR_BC})
\end{align}
Therefore, for $\ket{\tilde{\psi}^f_{\tot'}}:=(W_{R_A}\otimes1_{AB'C'})\ket{\psi^f_{\tot}}$,
\begin{align}
\ket{\tilde{\psi}^f_{\tot'}}=&(1_{R'_A}\otimes V\otimes 1_{B'})(1_{R'_A}\otimes U\otimes 1_{R_BC})\non
&\sum_{j}\sqrt{p_j}\ket{\psi_{j,AR'_{A1}}}\otimes\ket{j_{R'_{A2}}}\otimes\ket{\phi_{BR_B}}\otimes\ket{\eta_C}\nonumber\\
=&\sum_j\sqrt{p_j}\ket{\tilde{\psi}^f_{j,AR'_{A1}B'C'}}\otimes\ket{j_{R'_{A2}}},
\end{align}
where $\ket{\tilde{\psi}^f_{j,AR'_{A1}B'C'}}:=(1_{R'_{A1}}\otimes V\otimes 1_{B'})(1_{R'_{A1}}\otimes U\otimes 1_{R_BC})\ket{\psi_{j,AR'_{A1}}}\otimes\ket{\phi_{BR_B}}\otimes\ket{\eta_C}$.

Now, we define the desired $\tilde{\psi}^{f}_{j,B'C'}$ as $\tilde{\psi}^{f}_{j,B'C'}:=\Tr_{AR'_{A1}}[\tilde{\psi}^f_{j,AR'_{A1}B'C'}]$. 
Then, since $\{\ket{j_{R'_{A2}}}\}$ are orthonormal, for $\tilde{\psi}^{f}_{B'C'}:=\Tr_{AR'_A}[\tilde{\psi}^f_{\tot'}]$,
\begin{align}
\tilde{\psi}^{f}_{B'C'}=\sum_jp_j\tilde{\psi}^{f}_{j,B'C'}\label{star6}
\end{align}
We establish $\tilde{\psi}^{f}_{B'C'}=\psi^f_{B'C'}$ as follows:
\begin{align}
\tilde{\psi}^{f}_{B'C'}=&\Tr_{AR'_A}[\tilde{\psi}^f_{\tot'}]\nonumber\\
=&\Tr_{AR'_A}[W_{R_A}\otimes1_{AB'C'}\psi^f_{\tot}W^\dagger_{R_A}\otimes1_{AB'C'}]\nonumber\\
=&\Tr_{AR_A}[(W^\dagger_{R_A}W_{R_A}\otimes1_{AB'C'})(1_{R_A}\otimes V\otimes 1_{B'})\non
&(1_{R_A}\otimes U\otimes 1_{R_BC})\non
&\ket{\psi_{\tot}}\bra{\psi_{\tot}}(1_{R_A}\otimes U^\dagger\otimes 1_{R_BC})(1_{R_A}\otimes V^\dagger\otimes 1_{B'})]
\nonumber\\
=&\Tr_{AR_A}[(1_{R_A}\otimes V\otimes 1_{B'})(1_{R_A}\otimes U\otimes 1_{R_BC})\non
&(W^\dagger_{R_A}W_{R_A}\otimes1_{AB'C'})\non
&\ket{\psi_{\tot}}\bra{\psi_{\tot}}(1_{R_A}\otimes U^\dagger\otimes 1_{R_BC})(1_{R_A}\otimes V^\dagger\otimes 1_{B'})]
\nonumber\\
\stackrel{(a)}{=}&\Tr_{AR_A}[(1_{R_A}\otimes V\otimes 1_{B'})(1_{R_A}\otimes U\otimes 1_{R_BC})
\non
&\ket{\psi_{\tot}}\bra{\psi_{\tot}}(1_{R_A}\otimes U^\dagger\otimes 1_{R_BC})(1_{R_A}\otimes V^\dagger\otimes 1_{B'})]
\nonumber\\
=&\Tr_{AR_A}[\psi^f_{\tot}]\nonumber\\
=&\psi^{f}_{B'C'}.\label{star7}
\end{align}
Here we use \eqref{WRA} in (a).
Combining \eqref{star6} and \eqref{star7}, we obtain \eqref{star1}.

Similarly, we can obtain \eqref{star2} as follows:
\begin{align}
\Tr_{C'}[\tilde{\psi}^f_{j,B'C'}]
=&\Tr_{AR'_{A1}C'}[\tilde{\psi}^f_{j,AR'_{A1}B'C'}]\nonumber\\
=&\Tr_{AR'_{A1}C'}[(1_{R'_{A1}}\otimes V\otimes 1_{B'})\non
&(1_{R'_{A1}}\otimes U\otimes 1_{R_BC})\psi_{j,AR'_{A1}}\otimes\phi_{BR_B}\otimes\eta_C\non
&(1_{R'_{A1}}\otimes U^\dagger\otimes 1_{R_BC})(1_{R'_{A1}}\otimes V^\dagger\otimes 1_{B'})]\nonumber\\
=&\Tr_{AC'}[(V\otimes 1_{B'})(U\otimes 1_{R_BC})\non
&\rho_{j}\otimes\phi_{BR_B}\otimes\eta_C( U^\dagger\otimes 1_{R_BC})(V^\dagger\otimes 1_{B'})]\nonumber\\
=&\rho^f_{j,B'}.
\end{align}
Therefore, $\{\tilde{\psi}^{f}_{j,B'C'}\}$ actually satisfy \eqref{star1} and \eqref{star2}.
\end{proofof}

\section{Derivation of the properties of the variance and expectation values of the conserved quantity $X$}\label{S-VE}

In this section, we prove \eqref{Eeq}--\eqref{Vineq2-S}.
We first derive \eqref{Eeq}. We evaluate the difference between the left-hand side and the right-hand side of \eqref{Eeq} as follows:
\begin{align}
&\left(\ex{X_A}_{\xi_A}-\ex{X_{A'}}_{\xi^{f}_{A'}}\right)-\left(\ex{X_{B'}}_{\xi^{f}_{B'}}-\ex{X_{B}}_{\xi_{B}}\right)\non
&=\left(\ex{X_A}_{\xi_A}+\ex{X_{B}}_{\xi_{B}}\right)-\left(\ex{X_{A'}}_{\xi^{f}_{A'}}+\ex{X_{B'}}_{\xi^{f}_{B'}}\right)\nonumber\\
&=\Tr[(X_{A}+X_{B})\xi_{A}\otimes\xi_{B}]-\Tr[(X_{A'}+X_{B'})U\xi_{A}\otimes\xi_{B}U^\dagger]\nonumber\\
&\stackrel{(a)}{=}0
\end{align}
Here we use $U(X_A+X_B)U^\dagger=X_{A'}+X_{B'}$ in (a).

We next show \eqref{Vineq}.
Note that
\begin{align}
\ex{(X_{A'}+X_{B'})^2}_{U\xi_{A}\otimes\xi_{B}U^\dagger}&=\Tr[(X_{A'}+X_{B'})^2U\xi_{A}\otimes\xi_{B}U^{\dagger}]\nonumber\\
&=\Tr[U^\dagger(X_{A'}+X_{B'})^2U\xi_{A}\otimes\xi_{B}]\nonumber\\
&=\Tr[(X_{A}+X_{B})^2\xi_{A}\otimes\xi_{B}]\nonumber\\
&=\ex{(X_A+X_B)^2}_{\xi_{A}\otimes\xi_{B}}.
\end{align}
Combining this and $\ex{X_A+X_B}_{\xi_{A}\otimes\xi_{B}}=\ex{X_{A'}+X_{B'}}_{U\xi_{A'}\otimes\xi_{B'}U^\dagger}$ which is easily obtained from \eqref{Eeq}, we obtain 
\begin{align}
V_{\xi_A\otimes\xi_B}(X_A+X_B)=V_{U(\xi_A\otimes\xi_B)U^\dagger}(X_{A'}+X_{B'}).\label{VP1}
\end{align}
From \eqref{VP1}, the lower bound for $V_{\xi_A}(X_A)+V_{\xi_B}(X_B)$ is
\begin{align}
&V_{\xi_A}(X_A)+V_{\xi_B}(X_B)\non
&=V_{\xi_A\otimes\xi_B}(X_A+X_B)\nonumber\\
&=V_{U(\xi_A\otimes\xi_B)U^\dagger}(X_{A'}+X_{B'})\nonumber\\
&=V_{\xi^f_{A'}}(X_{A'})+V_{\xi^f_{B'}}(X_{B'})+2Cov_{U(\xi_A\otimes\xi_B)U^\dagger}(X_{A'}:X_{B'})\nonumber\\
&\ge V_{\xi^f_{A'}}(X_{A'})+V_{\xi^f_{B'}}(X_{B'})-2\sqrt{V_{\xi^f_{A'}}(X_{A'})V_{\xi^f_{B'}}(X_{B'})}\nonumber\\
&=\left(\sqrt{V_{\xi^f_{A'}}(X_{A'})}-\sqrt{V_{\xi^f_{B'}}(X_{B'})}\right)^2,
\end{align}
where $Cov_{\xi}(X:Y):=\ex{\{X-\ex{X}_{\xi},Y-\ex{Y}_{\xi}\}}_{\xi}/2$ and $\{X,Y\}:=XY+YX$.
Taking the square root of both sides and applying $\sqrt{x}+\sqrt{y}\ge\sqrt{x+y}$ to the left-hand side, we obtain
\begin{align}
\sqrt{V_{\xi^f_{B'}}(X_{B'})}&\le\sqrt{V_{\xi^f_{A'}}(X_{A'})}+\sqrt{V_{\xi_A}(X_A)}+\sqrt{V_{\xi_B}(X_B)}\nonumber\\
&\le\sqrt{V_{\xi_B}(X_B)}+\Delta_{+}.
\end{align}
We can derive \eqref{Vineq2-S} in the same way as \eqref{Vineq}.

\clearpage

\begin{widetext}
\begin{center}
{\large \bf Supplementary information for \protect \\ 
``Universal limitation of quantum information recovery: symmetry versus coherence''}\\
\vspace*{0.3cm}
Hiroyasu Tajima$^{1,2}$ and Keiji Saito$^{3}$ \\
\vspace*{0.1cm}
$^{1}${\small \em Graduate School of Informatics and Engineering,
The University of Electro-Communications,
1-5-1 Chofugaoka, Chofu, Tokyo 182-8585, Japan}
\\
$^{2}${\small \em JST, PRESTO, 4-1-8 Honcho, Kawaguchi, Saitama, 332-0012, Japan}
\\
$^{3}${\small \em Department of Physics, Keio University, 3-14-1 Hiyoshi, Yokohama, 223-8522, Japan}
\end{center}

\setcounter{equation}{0}
\setcounter{lemma}{0}
\setcounter{page}{1}
\setcounter{section}{0}
\setcounter{figure}{0}
\renewcommand{\theequation}{\Roman{section}.\arabic{equation}}
\renewcommand{\thesection}{Supp.\Roman{section}}
\renewcommand{\thefigure}{S.\arabic{figure}}

The supplementary material is organized as follows.
In Sec.~\ref{TipsRToA}, we introduce several useful tips about the resource theory of asymmetry. The tips is a generalized version of tips in Appendix \ref{AppA}. 
In Sec. \ref{alleviation}, we give a concrete example that quantum coherence alleviates the recovery error.
In Sec.~\ref{HPTips}, we introduce several tips about the Hayden-Preskill model with the conservation law of $X$.
In Sec.~\ref{S-woRB}, we show the universal limitation of information recovery without using $R_B$.
In Sec.~\ref{covariant-codes}, we show that the approximate Eastin-Knill theorem is given as corollary of \eqref{SIQ2}.
In Sec.~\ref{g-symmetry}, we generalize the results in the main text to the case of general Lie group symmetries.
In Sec.~\ref{v-symmetry}, we generalize the results in the main text to the case of weakly violated symmetry.
Finally, in Sec.~\ref{Stightness}, we derive \eqref{check_ex1}--\eqref{check_ex3} and \eqref{upper_ex} in the main text.

For the readers' convenience, here we present our basic setup which we use in this paper.
Our setup is shown in Fig.~\ref{Setup1}.
We prepare four systems $A$, $B$, $R_A$ and $R_B$ and two pure states $\ket{\psi_{AR_A}}$ and $\ket{\phi_{BR_B}}$ on $AR_A$ and $BR_B$. 
After preparation, we perform a unitary operation $U$ on $AB$ and divide $AB$ into $A'$ and $B'$.
Then, we try to recover the initial state $\ket{\psi_{AR_A}}$ on $AR_A$ by performing a recovery operation $\calR$ which is a CPTP map from $A'R_B$ to $A$ while keeping $R_A$ as is.
And we define the minimum recovery error of the above process as $\delta$:
\begin{align}
 \delta(\psi_{AR_A},\calI) &:=\!\! \min_{\mbox{$\calR $ } \atop  (A'R_B \to A )     }  \!\!\!\!\!\!
           D_F\!\left( \psi_{AR_A},\mathrm{id}_{R_A}\otimes\calR[\Tr_{B'}(U\psi_{AR_A}\otimes\phi_{BR_B}U^\dagger) ] \right) \,.
\end{align}
Here we use the purified distance $D_F(\rho,\sigma):=\sqrt{1-F^2(\rho,\sigma)}=\sqrt{1-\Tr[\sqrt{\sqrt{\sigma}\rho\sqrt{\sigma}}]^2}$ \cite{Tomamichel-S} and abbreviations  $\psi_{AR_A}:=\ket{\psi_{AR_A}}\bra{\psi_{AR_A}}$, $\phi_{BR_B}:=\ket{\phi_{BR_B}}\bra{\phi_{BR_B}}$ and $\calI:=(\phi_{BR_B},U)$.
Without special notice, we abbreviates $\delta(\psi_{AR_A},\calI)$ as $\delta$ as the main text.
We also use abbreviations for density operators of pure states like $\eta=\ket{\eta}\bra{\eta}$.
Hereafter, we refer to this setup as ``Setup 1.''
In each section of this Supplementary Material, we use several different additional assumptions with Setup 1.
When we use such additional assumptions, we mention them.
Note that Setup 1 does not contain the conservation law of $X$.
When we assume the conservation law of $X$, i.e. $U(X_A+X_B)U^\dagger=X_{A'}+X_{B'}$ for Hermite operators $X_\alpha$ on $\alpha$ ($\alpha=A,B,A',B'$), we say ``Setup 1 with the conservation law of $X$.''

\begin{figure}[h]
		\centering
		\includegraphics[width=0.5\textwidth]{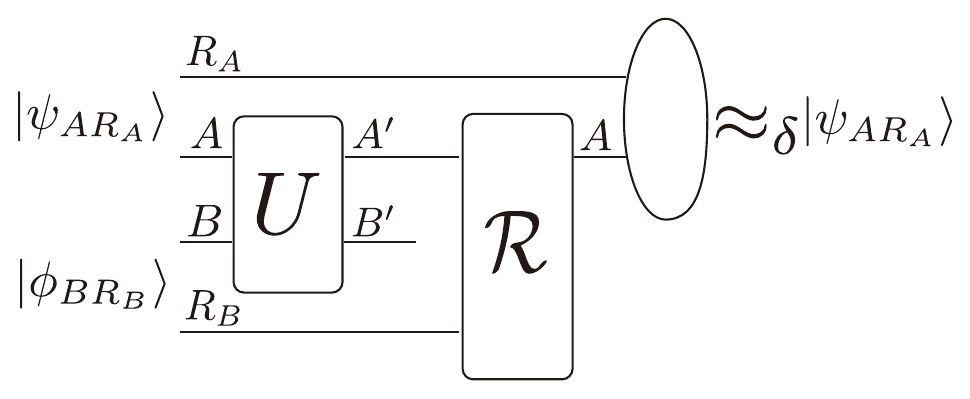}
		\caption{Schematic diagram of Setup 1.}
		\label{Setup1}
	\end{figure}

\section{Tips for resource theory of asymmetry for the case of general symmetry}\label{TipsRToA}

In this section, we give a very basic information about the resource theory of asymmetry (RToA) \cite{Gour-S,Marvian-S,Marvian-thesis-S,Keyl} for the case of general symmetry.

We firstly introduce covariant operations that are free operations in RToA.
Let us consider a CPTP map $\calE$ from a system $A$ to another system $A'$ and unitary representations $\{U_{g,A}\}_{g\in G}$ on $A$ and  $\{U_{g,A'}\}_{g\in G}$ on $A'$ of a group $G$.
The CPTP $\calE$ is said to be covariant with respect to $\{U_{g,A}\}_{g\in G}$ and  $\{V_{g,A'}\}_{g\in G}$, when the following relation holds:
\begin{align}
\calV_{g,A'}\circ\calE(...)=\calE\circ\calU_{g,A}(...),\enskip\forall g\in G,
\end{align}
where $\calU_{g,A}(...):=U_{g,A}(...)U^{\dagger}_{g,A}$ and $\calV_{g,A'}(...):=V_{g,A'}(...)V^\dagger_{g,A'}$.
Similarly, a unitary operation $U_A$ on $A$ is said to be invariant with respect to $\{U_{g,A}\}_{g\in G}$ and $\{V_{g,A}\}_{g\in G}$, when the following relation holds:
\begin{align}
\calV_{g,A}\circ\calU(...)=\calU\circ\calU_{g,A}(...),\enskip\forall g\in G,
\end{align}
where $\calU(...):=U(...)U^{\dagger}$.

Next, we introduce symmetric states that are free states of resource theory of asymmetry. 
A state $\rho$ on $A$ is said to be a symmetric state when it satisfies the following relation:
\begin{align}
\rho=\calU_{g,A}(\rho),\enskip\forall g\in G.
\end{align}

When a CPTP-map $\calE$ is covariant, it can be realized by invariant unitary and symmetric state \cite{Marvian-thesis-S,Keyl}. 
To be concrete, when a CPTP map $\calE$: $A\rightarrow A'$ is covariant with respect to $\{U_{g,A}\}_{g\in G}$ and  $\{U_{g,A'}\}_{g\in G}$, there exist another system $B$,  unitary representations $\{U_{g,B}\}_{g\in G}$ and $\{V_{g,B'}\}_{g\in G}$ on $B$ and $B'$ ($AB=A'B'$), a unitary $U_{AB}$ which is invariant with respect to $\{U_{g,A}\otimes U_{g,B}\}_{g\in G}$ and $\{V_{g,A'}\otimes V_{g,B'}\}_{g\in G}$ , and a symmetric state $\mu_B$ with respect to $\{U_{g,B}\}_{g\in G}$ such that
\begin{align}
\calE(...)=\Tr_{B'}[U_{AB}(...\otimes\mu_B)U^{\dagger}_{AB}].
\end{align}

\section{An example of the error suppression by quantum coherence in information recovery}\label{alleviation}
In this section, we give a concrete example that large $\calF$ actually enables the recovery error $\delta$ to be smaller than $\SA/8\Delta_{+}$.
We consider Setup 1 with the conservation law of $X$, i.e., $U(X_A+X_B)U^\dagger=X_{A'}+X_{B'}$.
We set $A$ to be a qubit system and $B$ to be a $6M+1$-level system, where $M$ is a natural number that we can choose freely.
We also set $R$ and $R_B$ as copies of $A$ and $B$, respectively.
We take $X_A$ and $X_B$ as follows:
\begin{align}
X_{A}&:=\ket{1}_{A}\bra{1}_{A},\label{XAeg}\\
X_{B}&:=\sum^{3M}_{k=-3M}k\ket{k}_{B}\bra{k}_{B}.
\end{align}
where $\{\ket{k}_A\}^{1}_{k=0}$ and $\{\ket{k}_B\}^{3M}_{k=-3M}$ are orthonormal basis of $A$ and $B$.

Under this setup, we consider the case where $A=A'$, $B=B'$, $X_{A}=X_{A'}$ and $X_{B}=X_{B'}$.
In this case, due to \eqref{XAeg} and $X_A=X_{A'}$, the equality $\Delta_{+}=1$ holds.
Therefore, \eqref{SIQ1} becomes the following inequality:
\begin{align}
\frac{\SA}{2(\sqrt{\calF}+4)}\le\delta.
\end{align} 
Therefore, when $\calF=0$, the error $\delta$ can not be smaller than $\SA/8$.
Here, we show that when $\calF$ is large enough, the error $\delta$ actually becomes smaller than $\SA/8$.
Let us take $\ket{\psi_{AR_A}}$, $\ket{\phi_{BR_B}}$ and $U$ as 
\begin{align}
\ket{\psi_{AR_A}}&=\frac{\ket{0}_{A}\ket{0}_{R_A}+\ket{1}_{A}\ket{1}_{R_A}}{\sqrt{2}}\\
\ket{\phi_{BR_B}}&=\frac{\sum^{M}_{k=-M}\ket{k}_{B}\ket{k}_{R_B}}{\sqrt{2M+1}},\\
U&=\sum_{-2M\le k\le 2M}\ket{1}_{A}\bra{0}_{A}\otimes\ket{k-1}_{B}\bra{k}_{B}
+\sum_{-2M-1\le k\le 2M-1}\ket{0}_{A}\bra{1}_{A}\otimes\ket{k+1}_{B}\bra{k}_{B}\nonumber\\
&+\sum_{k<-2M, 2M<k}\ket{0}_{A}\bra{0}_{A}\otimes\ket{k}_{B}\bra{k}_{B}
+\sum_{k<-2M-1,2M-1<k}\ket{1}_{A}\bra{1}_{A}\otimes\ket{k}_{B}\bra{k}_{B}.
\end{align}
Then, $U$ is a unitary satisfying $U(X_{A}+X_{B})=X_{A}+X_{B}$, and the CPTP-map $\calE$ implemented by $(U,\ket{\phi_{BR_B}})$ is expressed as 
\begin{align}
\calE(...)=\ket{1}_A\bra{0}_A(...)\ket{0}_A\bra{1}_A+\ket{0}_A\bra{1}_A(...)\ket{1}_A\bra{0}_A.\label{Eexam}
\end{align}
Due to \eqref{Eexam} and $\rho_A:=\Tr_{R_A}[\psi_{AR_A}]=\frac{\ket{0}\bra{0}_A+\ket{1}\bra{1}_A}{2}$, the quantity $\SA$ is equal to 1/2.
Here, let us define a recovery CPTP-map $\calR_V$ as
\begin{align}
\calR_V(...):=\Tr_{R_BB}[V_{AR_B}(...)V^\dagger_{AR_B}]
\end{align}
where $V_{AR_B}$ is a unitary operator on $AR_B$ defined as
\begin{align}
V_{AR_B}&:=\sum_{-3M+1\le k\le 3M}\ket{0}\bra{1}_{A}\otimes\ket{k-1}\bra{k}_{R_B}
+\sum_{-3M\le k\le3M-1}\ket{1}\bra{0}_{A}\otimes\ket{k+1}\bra{k}_{R_B}\nonumber\\
&+\ket{0}\bra{1}_{A}\otimes\ket{3M}\bra{-3M}_{R_B}
+\ket{1}\bra{0}_{A}\otimes\ket{-3M}\bra{3M}_{R_B}.
\end{align}
(Note that the recovery $V_{AR_B}$ is not required to satisfy the conservation law).
Then, after $V_{AR_B}$, the total system is in 
\begin{align}
&(V_{AR_B}\otimes1_{BR_A})(U_{AB}\otimes1_{R_AR_B})(\ket{\psi_{AR_A}}\otimes\ket{\phi_{BR_B}})\nonumber\\
&=\frac{1}{\sqrt{2(2M+1)}}\sum^{M}_{k=-M}(\ket{0}_A\ket{0}_{R_A}\ket{k-1}_{B}\ket{k-1}_{R_B}+\ket{1}_A\ket{1}_{R_A}\ket{k+1}_{B}\ket{k+1}_{R_B})\nonumber\\
&=\frac{\sqrt{2M-1}}{\sqrt{2M+1}}\ket{\psi_{AR_A}}\otimes\ket{\tilde{\phi}_{BR_B}}+\frac{1}{\sqrt{2M+1}}\ket{00}_{AR_A}\frac{\ket{-M,-M}_{BR_B}+\ket{-M-1,-M-1}_{BR_B}}{\sqrt{2}}\nonumber\\
&+\frac{1}{\sqrt{2M+1}}\ket{11}_{AR_A}\frac{\ket{M,M}_{BR_B}+\ket{M+1,M+1}_{BR_B}}{\sqrt{2}},
\end{align}
where $\ket{\tilde{\phi}_{BR_B}}:=\frac{1}{\sqrt{2M-1}}\sum^{M-1}_{k=-M+1}\ket{k,k}_{BR_B}$.
By partial trace of $BR_B$, we obtain the final state of $AR_A$ as follows:
\begin{align}
\psi^{f}_{AR_A}=\frac{2M-1}{2M+1}\psi_{AR_A}+\frac{1}{2M+1}\ket{00}\bra{00}_{AR_A}+\frac{1}{2M+1}\ket{11}\bra{11}_{AR_A}.
\end{align}
Therefore, 
\begin{align}
D_F(\psi^f_{AR_A},\psi_{AR_A})^2=1-\bra{\psi_{AR_A}}\psi^{f}_{AR_A}\ket{\psi_{AR_A}}=\frac{2}{2M+1}.
\end{align}
Thus,  we obtain
\begin{align}
\delta\le\sqrt{\frac{2}{2M+1}}.
\end{align}
Hence, when $M$ is large enough, we can make $\delta$ strictly smaller than $\SA/8=1/16$.  
Since $\calF=4V_{\rho_B}(X_B)$ ($\rho_B:=\Tr_{R_B}[\phi_{BR_B}]$), large $M$ means large $\calF$.
Therefore, when $\calF$ is large enough, we can make $\delta$ smaller than $1/16$.

\section{Tips for the application to Hayden-Preskill model with a conservation law}\label{HPTips}

\begin{figure}[tb]
		\centering
		\includegraphics[width=.45\textwidth]{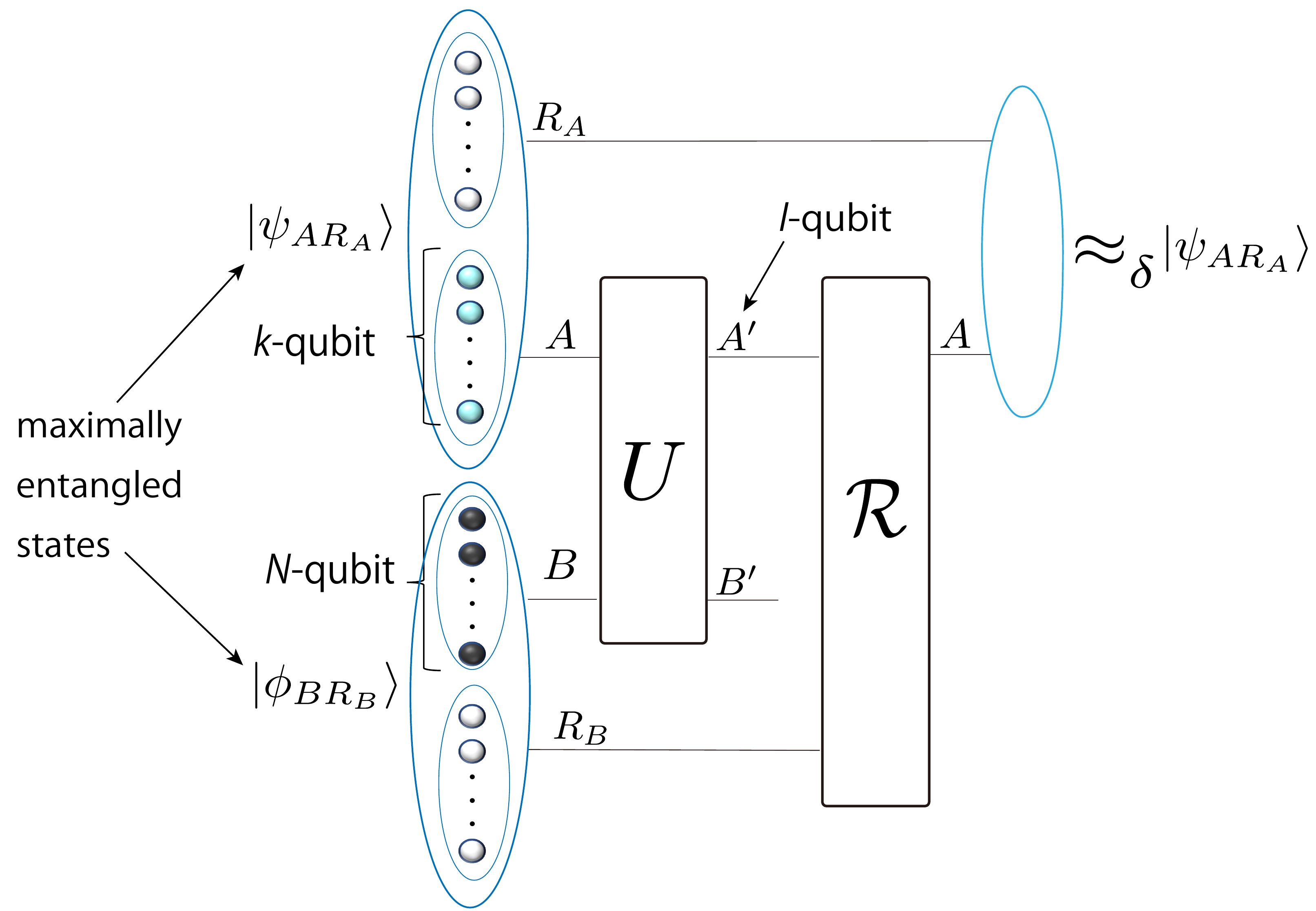}
		\caption{Schematic diagram of the Hayden-Preskill black hole model. It is almost a special case of our setup illustrated in Fig.~\ref{setupSR}.}
		\label{HPmodel-S}
	\end{figure}

\subsection{Derivation of \eqref{A-bound}, \eqref{D-upp}, and \eqref{Ffbound} in the main text}
In this subsection, we give the detailed description of the scrambling of the expectation values and derivation of  \eqref{A-bound}, \eqref{D-upp}, and \eqref{Ffbound} in the main text.

For the readers' convenience, we firstly review the Hayden-Preskill model with the conservation law of $X$ which is introduced in the section \ref{HPwithX} in the main text. (Fig.~\ref{HPmodel-S})
The model is a specialized version of Setup 1 with the conservation law of $X$.
The specialized points are as follows:
1. $A$, $B$, $A'$ and $B'$ are $k$-, $N$-, $l$- and $k+N-l$-qubit systems, respectively.
2. We assume that the difference between minimum and the maximum eigenvalue of $X_{\alpha}$ (= $\calD_{X_{\alpha}}$ for $\alpha=A,B,A',B'$) to be equal to the number of particles $\alpha$ (= $k$, $N$, $l$, $N+k-l$ for $\alpha=A,B,A',B'$).

Under the above setup, when the conserved quantities are scrambled in the sense of the expectation values, we can derive  \eqref{A-bound}, \eqref{D-upp}, and \eqref{Ffbound}. Below, we show the derivation.
For simplicity of the description, we will use the following expression for real numbers $x$ and $y$:
\begin{align}
x\approx_{\epsilon}y\Leftrightarrow_{\mbox{def}}|x-y|\le\epsilon.
\end{align}
We also express the expectation values of $X_{\alpha}$ ($\alpha=A,B\mbox{ and }A'$) as follows:
\begin{align}
x_{A}(\rho_A)&:=\ex{X_A}_{\rho_A},\\
x_{B}(\rho_B)&:=\ex{X_{B}}_{\rho_B},\\
x_{A'}(\rho_A,\rho_B,U)&:=\ex{X_{A'}}_{\rho^{f}_{A'}}.
\end{align}
We show  \eqref{A-bound}, \eqref{D-upp}, and \eqref{Ffbound} as the following theorem:
\begin{theorem}\label{ex-sc-t1}
Let us take a real positive number $\epsilon$ which is lower than $1/(N+k)^2$, and the set of $(\ket{\psi_{AR_A}}, \ket{\phi_{BR_B}}, U)$. We refer to the initial state of $A$ as $\rho_A:=\Tr_{R_A}[\psi_{AR_A}]$, and assume that $[\rho_A,X_A]=0$ and $M:=M_{\rho_A}(X_A)>0$ holds.
We also assume that $(\ket{\phi_{BR_B}},U)$ satisfies the following relation for an arbitrary state $\rho$ on the support of $\rho_A$:
\begin{align}
x_{A'}(\rho,\rho_B,U)\approx_{\frac{1}{2}\epsilon M\gamma}(x_{A}(\rho)+x_{B}(\rho_B))\times\frac{l}{N+k},\label{ex-sc-eq}
\end{align}
where $\gamma:=\left(1-\frac{l}{N+k}\right)$.
Then, the following three inequalities hold:
\begin{align}
\SA&\ge M\gamma (1-\epsilon)\label{Aineq}\\
\sqrt{\calF_{f}}&\le\gamma(N+k)\label{Fineq}\\
\Delta_{\max}&\le\gamma k(1+\epsilon)\label{Dineq}
\end{align}
And when $\rho_A$ can be decomposed as $\rho_A=\frac{\rho_{0,A}+\rho_{1,A}}{2}$ such that $\rho_{0,A}$ and $\rho_{1,A}$ are eigenstates of $X_A$, the following inequality also holds:
\eq{
\calA_2\ge M\gamma (1-\epsilon)\label{A2ineq}
}
\end{theorem}

\begin{proof}
We firstly point out \eqref{Fineq} is easily derived by noting that $\calF_f=4V_{\rho^{f}_{B'}}(X_{B'})$ and that the number of qubits in $B'$ is $N+k-l$, which is equal to $(N+k)\gamma$.

To show \eqref{Aineq}, \eqref{Dineq} and \eqref{A2ineq}, let us take an arbitrary decomposition $\rho_A=\sum_jp_j\rho_{j,A}$, and evaluate $|\Delta_j|$ for the decomposition as follows:
\begin{align}
|\Delta_j|&=\left|\left(x_A(\rho_{j,A})-x_{A'}(\rho_{j,A},\rho_B,U)\right)-\left(\sum_jp_jx_{A}(\rho_{j,A})-\sum_jp_jx_{A'}(\rho_{j,A},\rho_B,U)\right)\right|\nonumber\\
&\approx_{\frac{1}{2}\epsilon M\gamma}
\left|x_A(\rho_{j,A})-(x_{A}(\rho_{j,A})+x_{B}(\rho_B))\frac{l}{N+k}-\sum_{j}p_jx_{A}(\rho_{j,A})+\sum_{j}p_j(x_{A}(\rho_{j,A})+x_{B}(\rho_B))\frac{l}{N+k}\right|\nonumber\\
&=\left|x_A(\rho_{j,A})-\sum_{j}p_jx_{A}(\rho_{j,A})\right|\gamma.\label{ev-Deltaj}
\end{align}

To derive \eqref{Aineq} from the above evaluation, let us choose a decomposition $\rho_A=\sum_jp_j\rho_{j,A}$ where each $\rho_{j,A}$ is in eigenspace of $X_A$. 
We can choose such a decomposition due to $[\rho_A,X_A]=0$.
Then, $\sum_jp_j\left|x_A(\rho_{j,A})-\sum_{j}p_jx_{A}(\rho_{j,A})\right|=M$ holds. 
Applying \eqref{ev-Deltaj} to this decomposition, we obtain \eqref{Aineq}:
\begin{align}
\SA&\ge\sum_jp_j|\Delta_j|\ge M\gamma-M\gamma\epsilon.
\end{align}
And when $\rho_A$ can be decomposed as $\rho_A=\frac{\rho_0+\rho_1}{2}$ where $\rho_0$ and $\rho_1$ are eigenstates of $X_A$, applying \eqref{ev-Deltaj} to the decomposition $\rho_A=\frac{\rho_0+\rho_1}{2}$, we obtain \eqref{A2ineq} because of $\sum_{j=0,1}\frac{1}{2}\left|x_A(\rho_{j,A})-\sum_{j=0,1}\frac{1}{2}x_{A}(\rho_{j,A})\right|=M$:
\begin{align}
\SA_2&\ge\sum_{j=0,1}\frac{1}{2}|\Delta_j|\ge M\gamma-M\gamma\epsilon.
\end{align}

Similarly, we can derive \eqref{Dineq} as follows
\begin{align}
\Delta_{\max}&=\max_{\{p_j,\rho_{j,A}\}}|\Delta_j|\nonumber\\
&\le \max_{\{p_j,\rho_{j,A}\}}\left(\left|x_A(\rho_{j,A})-\sum_{j}p_jx_{A}(\rho_{j,A})\right|\right)\gamma+M\gamma\epsilon,\nonumber\\
&\le\max_{\{p_j,\rho_{j,A}\}}\left(\left|x_A(\rho_{j,A})-\sum_{j}p_jx_{A}(\rho_{j,A})\right|\right)\gamma(1+\epsilon),\nonumber\\
&\le \calD_{X_A}\gamma(1+\epsilon),\nonumber\\
&\le k\gamma(1+\epsilon).
\end{align}
where $\{p_j,\rho_{j,A}\}$ runs over all possible decompositions of $\rho_A$.

\end{proof}

\subsection{Derivation of (\ref{Rstar2}) in the main text}
In this subsection, we derive \eqref{Rstar2} based on the assumptions introduced in the main text. We show that the assumptions actually hold in the Haar random unitary with  the conservation law of $X$ in the next subsection.
We show \eqref{Rstar2} as the following theorem:
\begin{theorem}\label{va-sc-t2}
Let us take a real positive number $\epsilon$ which is lower than $1/(N+k)^2$, and a real non-negative number $x$. 
We define projections $P^{AB}_x$, $P^{A}_x$ and $P^{B}_{x}$ on $AB$ and $B$ as follows:
\eq{
P^{AB}_x&:=\sum_{\lambda_{\min}(X_A+X_B)+x\le x_A(i)+x_B(j)\le\lambda_{\min}(X_A+X_B)-x}\ket{i,a}\bra{i,a}\otimes\ket{j,b}\bra{j,b},\\
P^{A}_x&:=\sum_{\lambda_{\min}(X_A)+x\le x_A(i)\le\lambda_{\min}(X_A)-x}\ket{i,a}\bra{i,a},\\
P^{B}_x&:=\sum_{\lambda_{\min}(X_B)+x\le x_B(j)\le\lambda_{\min}(X_B)-x}\ket{j,b}\bra{j,b},
}
where $\lambda_{\min}(X_A+X_B)$ and $\lambda_{\max}(X_A+X_B)$ are the minimum and maximum eigenvalues of $X_A+X_B$, and $\ket{i,a}$ and $\ket{j,b}$ are eigenvectors of $X_A$ and $X_B$ whose eigenvalues are $x_{A}(i)$ and $x_{B}(j)$.
We also take pure states $\ket{\psi_{AR_A}}$ and $\ket{\phi_{BR_B}}$ on $AR_A$ and $BR_B$. 
We assume that $\rho_A:=\Tr_{R_A}[\ket{\psi_{AR_A}}\bra{\psi_{AR_A}}]$  satisfies $[X_A,\rho_A]=0$ and $\rho_B:=\Tr_{R_B}[\ket{\phi_{BR_B}}\bra{\phi_{BR_B}}]$ satisfies $[X_B,\rho_B]=0$ and 
\eq{
V_{\rho_B}(X_B)\le N/4.\label{assumVB}
}
We also assume that $\rho_A$ and $\rho_B$ satisfy at least one of the following two inequalities:
\eq{
\left\|\frac{P^{A}_x\rho_AP^A_x}{\Tr[(\rho_A)P^{A}_x]}-\rho_A\right\|_1&\le\frac{\epsilon^2}{3(N+k)^2},\label{Rstar1.75SA}\\
\left\|\frac{P^{B}_x\rho_BP^B_x}{\Tr[(\rho_B)P^{B}_x]}-\rho_B\right\|_1&\le\frac{\epsilon^2}{3(N+k)^2}.\label{Rstar1.75S}
}
Let us take a unitary operation $U$ on $AB$ satisfying the conservation law of $X$.
We assume that $U$ satisfies the following two relations for arbitrary $\ket{i,a}\otimes\ket{j,b}$ satisfying $P^{AB}_x\ket{i,a}\otimes\ket{j,b}=\ket{i,a}\otimes\ket{j,b}$:
\eq{
V_{\rho'_{\alpha'|i,a,j,b,U}}(X_{\alpha'})&\le\frac{1+\epsilon}{4}\min\{l,\gamma(N+k)\},\label{RstarS}\\
x_{A'}(\ket{i,1},\ket{j,b},U)&\approx_{\epsilon}\frac{l}{N+k}(x_A(\ket{i,a})+x_B(\ket{j,b})).\label{Rstar1.5S}
}
Then, the following inequality hold:
\eq{
\sqrt{\calF_f}&\le2(1+\epsilon)\sqrt{\gamma(N+k^2)}.\label{Rstar2S}
}
\end{theorem}

\begin{proofof}{\eqref{Rstar2S}(=\eqref{Rstar2} in the main text)}
We use the following equation which is valid for arbitrary probability $\{q_k\}$ and operator $W$:
\begin{align}
V_{\sum_kq_k\sigma_k}(W)&=\sum_kq_kV_{\sigma_k}(W)+V_{\{q_k\}}(\{\ex{W}_{\sigma_k}\}).\label{keyB}
\end{align}
where $V_{\{q_k\}}(\{\ex{W}_{\sigma_k}\}):=\sum_{k}q_k\ex{W}^2_{\sigma_k}-(\sum_kq_k\ex{W}_{\sigma_k})^2$ is the variance of the expectation values $\{\ex{W}_{\sigma_k}\}$ with the probability $\{q_k\}$.
We prove \eqref{keyB} as follows:
\begin{align}
V_{\sum_kq_k\sigma_k}(W)&=\sum_{k}q_k\ex{W^2}_{\sigma_k}-(\ex{W}_{\sum_kq_k\sigma_k})^2
\nonumber\\
&=\sum_{k}q_k\ex{W^2}_{\sigma_k}-\sum_{k}q_k\ex{W}^2_{\sigma_k}+\sum_{k}q_k\ex{W}^2_{\sigma_k}-(\ex{W}_{\sum_kq_k\sigma_k})^2\nonumber\\
&=\sum_{k}q_k(\ex{W^2}_{\sigma_k}-\ex{W}^2_{\sigma_k})+\sum_{k}q_k\ex{W}^2_{\sigma_k}-(\sum_kq_k\ex{W}_{\sigma_k})^2\nonumber\\
&=\sum_kq_kV_{\sigma_k}(W)+V_{\{q_k\}}(\{\ex{W}_{\sigma_k}\}).
\end{align}

We also use the following relation for arbitrary states $\rho$ and $\tilde{\rho}$ and an arbitrary operator $Y$:
\eq{
V_{\rho}(Y)\approx_{3\|Y\|^2_\infty\|\rho-\tilde{\rho}\|_1}V_{\tilde{\rho}}(Y).\label{III23}
}
We prove \eqref{III23} as follows:
\eq{
|V_{\rho}(Y)-V_{\tilde{\rho}}(Y)|&\le|\ex{Y^2}_{\rho}-\ex{Y^2}_{\tilde{\rho}}|+|\ex{Y}^2_{\rho}-\ex{Y}^2_{\tilde{\rho}}|\non
&\le\|Y\|^2_\infty\|\rho-\tilde{\rho}\|_1
+2\max\{|\ex{Y}_\rho|,|\ex{Y}_{\tilde{\rho}}|\}\|Y\|_\infty\|\rho-\tilde{\rho}\|_1\non
&\le3\|Y\|^2_\infty\|\rho-\tilde{\rho}\|_1.
}

Let us derive \eqref{Rstar2S}.
First we consider the case when \eqref{Rstar1.75S} holds.
Hereafter we use the abbreviations $\tilde{\rho}_B:=\frac{P^{B}_x\rho_BP^B_x}{\Tr[(\rho_B)P^{B}_x]}$ and $\tilde{\rho}'_{B'}:=\Tr_{A'}[U(\rho_A\otimes\tilde{\rho}_B)U^\dagger]$.
Then, we obtain 
\eq{
\|\rho'_{B'}-\tilde{\rho}'_{B'}\|_1\le\|\rho_A\otimes\rho_B-\rho_A\otimes\tilde{\rho}_B\|_1\le\frac{\epsilon^2}{3(N+k)^2}.\label{III25.5}
}
Combining \eqref{III23}, \eqref{III25.5} and $\|X_{B'}\|_{\infty}=\gamma(N+k)$, we obtain
\eq{
V_{\rho'_{B'}}(X_{B'})\approx_{\gamma^2\epsilon^2}V_{\tilde{\rho}'_{B'}}(X_{B'})\label{III26}
}

Note that the state $\rho_{A}\otimes\tilde{\rho}_{B}$ on $AB$ can be written in the following form:
\begin{align}
\rho_{A}\otimes\tilde{\rho}_{B}=\sum_{i,a,j,b}r_{i,a,A}r_{j,b,B}\ket{i,a}\bra{i,a}\otimes\ket{j,b}\bra{j,b}.
\end{align}
Therefore, we can divide $V_{\tilde{\rho}'_{B'}}(X_{B'})$ into two parts:
\eq{
V_{\tilde{\rho}'_{B'}}(X_{B'})=\sum_{i,a,j,b}r_{i,a,A}r_{j,b,B}V_{\rho'_{B'|i,a,j,b}}(X_{B'})+V_{\{r_{i,a,A}r_{j,b,B}\}}(\{\ex{X_{B'}}_{\rho'_{B'|i,a,j,b}}\}).\label{divide}
}
Due to \eqref{Rstar1.5S} and
$P^{AB}_x\rho_{A}\otimes\tilde{\rho}_{B}P^{AB}_x=\rho_{A}\otimes\tilde{\rho}_{B}$, $\ex{X_{B'}}_{\rho'_{B'|i,a,j,b}}$ satisfies
\eq{
\ex{X_{B'}}_{\rho'_{B'|i,a,j,b}}=\gamma(x_{A}(i)+x_{B}(j))+\epsilon_{i,a,j,b}
}
where $|\epsilon_{i,a,j,b}|\le\epsilon$.
Therefore, due to $\sqrt{V_{\{q_k\}}(z_k+z'_{k})}\le\sqrt{V_{\{q_k\}}(\{z_k\})}+\sqrt{V_{\{q_k\}}(\{z'_{k}\})}$, we obtain
\eq{
\sqrt{V_{\{r_{i,a,A}r_{j,b,B}\}}(\{\ex{X_{B'}}_{\rho'_{B'|i,a,j,b}}\})}&\le\sqrt{V_{\{r_{i,a,A}r_{j,b,B}\}}(\{\gamma(x_{A}(i)+x_{B}(j))\})}+\epsilon\non
&\le\gamma\sqrt{V_{\{r_{i,a,A}\}}(\{x_{A}(i)\})+V_{\{r_{j,b,B}\}}(\{x_{B}(j)\})}+\epsilon \non
&\le\gamma\frac{\sqrt{N+k^2}}{2}+\gamma^2\epsilon+\epsilon.
\label{keyC}
}
Here we used $V_{\{r_{i,a,A}\}}(\{x_{A}(i)\})\le k^2/4$ and $V_{\{r_{j,b,B}\}}(\{x_{B}(j)\})=V_{\tilde{\rho}_B}(X_B)\le N/4+\gamma^2\epsilon^2$, which is given by \eqref{assumVB} and \eqref{III26}, in the last line.

Combining \eqref{RstarS}, \eqref{III26}, \eqref{divide} and \eqref{keyC}, we obtain \eqref{Rstar2S} as follows
\eq{
\sqrt{\calF_f}&=2\sqrt{V_{\rho'_{B'}}(X_{B'})}\non
&\le2\sqrt{V_{\tilde{\rho}'_{B'}}(X_{B'})}+2\gamma\epsilon\non
&\le2\left(\sqrt{\sum_{i,a,j,b}r_{i,a,A}r_{j,b,B}V_{\rho'_{B'|i,a,j,b}}(X_{B'})}+\sqrt{V_{\{r_{i,a,A}r_{j,b,B}\}}(\{\ex{X_{B'}}_{\rho'_{B'|i,a,j,b}}\})}\right)+2\gamma\epsilon
\non
&\le2\sqrt{\frac{1+\epsilon}{4}\min\{l,\gamma(N+k)\}}+\gamma\sqrt{N+k^2}+2\gamma^2\epsilon+2\epsilon+2\gamma\epsilon
\non
&\le2(1+\epsilon)\sqrt{\gamma(N+k^2)}.
}
Here we use $\epsilon\le1/(N+k)^2$, $M\le k$ and $\frac{1}{N+k}\le\gamma\le 1$ (since now we treat the case of $l<N+k$) in the last line.

Next, we consider the case when \eqref{Rstar1.75SA} holds.
In this case we use the abbreviations $\hat{\rho}_A:=\frac{P^{A}_x\rho_AP^A_x}{\Tr[(\rho_A)P^{A}_x]}$ and $\hat{\rho}'_{B'}:=\Tr_{A'}[U(\hat{\rho}_A\otimes\rho_B)U^\dagger]$.
Then, we obtain 
\eq{
\|\rho'_{B'}-\hat{\rho}'_{B'}\|_1\le\|\rho_A\otimes\rho_B-\hat{\rho}_A\otimes\rho_B\|_1\le\frac{\epsilon}{3(N+k)^2}.
}
Therefore, from \eqref{III23} and $\|X_{B'}\|_{\infty}=\gamma(N+k)$, we obtain
\eq{
V_{\rho'_{B'}}(X_{B'})\approx_{\gamma^2\epsilon}V_{\hat{\rho}'_{B'}}(X_{B'})\label{III26'}
}

Note that the state $\hat{\rho}_{A}\otimes\rho_{B}$ on $AB$ can be written in the following form:
\begin{align}
\hat{\rho}_{A}\otimes\rho_{B}=\sum_{i,a,j,b}r'_{i,a,A}r'_{j,b,B}\ket{i,a}\bra{i,a}\otimes\ket{j,b}\bra{j,b}.
\end{align}
Therefore, we can divide $V_{\hat{\rho}'_{B'}}(X_{B'})$ into two parts:
\eq{
V_{\hat{\rho}'_{B'}}(X_{B'})=\sum_{i,a,j,b}r'_{i,a,A}r'_{j,b,B}V_{\rho'_{B'|i,a,j,b}}(X_{B'})+V_{\{r'_{i,a,A}r'_{j,b,B}\}}(\{\ex{X_{B'}}_{\rho'_{B'|i,a,j,b}}\}).\label{divide'}
}
Due to $P^{AB}_x\hat{\rho}_{A}\otimes\rho_{B}P^{AB}_x=\hat{\rho}_{A}\otimes\rho_{B}$, we can derive the following inequality in the same manner as the derivation of \eqref{keyC}:
\eq{
\sqrt{V_{\{r'_{i,a,A}r'_{j,b,B}\}}(\{\ex{X_{B'}}_{\rho'_{B'|i,a,j,b}}\})}\le\gamma\frac{\sqrt{N+k^2}}{2}+\gamma^2\epsilon+\epsilon.
\label{keyC'}
}
Combining \eqref{RstarS}, \eqref{III26'}, \eqref{divide'} and \eqref{keyC'}, we obtain \eqref{Rstar2S} again:
\eq{
\sqrt{\calF_f}\le2(1+\epsilon)\sqrt{\gamma(N+k^2)}.
}

\end{proofof}

\subsection{Derivation of a tighter variation of \eqref{doda}}
Here, let us derive a tighter variation of \eqref{doda} that we use in Figure \ref{dodagraph} in the main text.
Substituting \eqref{Dineq}, \eqref{A2ineq}, \eqref{Rstar2} into \eqref{R-SIQ2}, we obtain
\eq{
\frac{1-\epsilon}{1+\epsilon}\frac{\gamma M}{\gamma k+2\sqrt{\gamma(N+k^2)}}\le\delta
}
To \g{interpret} the meaning of \g{this inequality, we} consider the case of $k=\sqrt{N}$ and $M=k/2$ (\g{we can take such an} $M$ by considering a relevant $\rho_A$ and its decomposition, e.g., $\rho_A=(\rho^{\max}_{k}+\rho^{\max}_{0})/2$, where $\rho^{\max}_{x}$  is the maximally mixed state of the eigenspace of $X_A$ whose eigenvalues is $x$).
Then, we obtain
\eq{
\frac{\mathrm{const.}}{1+\frac{2\sqrt{2}}{\sqrt{\gamma}}}\le\delta,\label{R-doda}
}
where the $\mathrm{const.}$ is a real number larger than $0.49$.
This inequality is twice tigher than \eqref{doda} in the main text.

\subsection{Evaluation of expectation values and variance of conserved quantity in Haar random unitary with the conservation law}
In this subsection, we show that when $U$ is a typical Haar random unitary with the conservation law of $X$, the assumptions \eqref{ex-sc-eq}, \eqref{assumVB}, \eqref{RstarS} and \eqref{Rstar1.5S} which are the assumptions used in the previous subsection actually hold.
To prove them explicitly, we define the Haar random unitary with the conservation law of $X$ in the black hole model.
In this subsection, we assume that each operator $X_{i}$ on each $i$-th qubit \g{is} the same, and that $X_{\alpha}=\sum_{i\in\alpha}X_{i}$ ($\alpha=A,B,A'\mbox{ and }B'$.) We also assume that $\calD_{X_i}=1$ and the ground eigenvalue of each $X_i$ is 0.  
We remark that under these assumptions, \eqref{assumVB} clearly holds whenever $\rho_B$ is a non-zero temperature Gibbs state $e^{-\beta X_B}/\Tr[e^{-\beta X_B}]$ for $0\le\beta<\infty$.

Let us refer to the eigenspace of $X_A+X_B$ whose eigenvalue is $m$ as $\calH^{(m)}$. Then, the Hilbert space of $AB$ is written as
\begin{align}
\calH_{AB}=\oplus^{N+k}_{m=0}\calH^{(m)}.
\end{align}
In this model, $X_A+X_B=X_{A'}+X_{B'}=\sum_{h}X_h$ holds ($X_h$ is the operator of $X$ on the $h$-th qubit), and thus $U$ satisfying \eqref{c-law} is also written as
\begin{align}
U=\oplus^{N+k}_{m=0}U^{(m)},
\end{align}
where $U^{(m)}$ is a unitary operation on $\calH^{(m)}$.
We refer to the unitary group of all unitary operations on $\calH^{(m)}$ as $\calU^{(m)}$, and refer to the Haar measure on $\calU^{(m)}$ as $H^{(m)}$.
Then, we can define the product measure of the Haar measures $\{H^{(m)}\}^{N+k}_{m=0}$ as follows:
\begin{align}
H^{\calM_{all}}_{\times}:=\times^{N+k}_{m=0}H^{(m)},
\end{align}
where $\calM_{all}:=\{0,1,...,M+k\}$.
The measure $H^{\calM_{all}}_{\times}$ is a probabilistic measure on the following unitary group on $\calH_{AB}$:
\begin{align}
\calU^{\calM_{all}}_{\times}:=\times^{N+k}_{m=0}\calU^{(m)}.
\end{align}
Since every $U\in \calU^{\calM_{all}}_{\times}$ satisfies $U(X_A+X_B)U^\dagger=X_{A'}+X_{B'}$, we refer to $U$ chosen from $\calU^{\calM_{all}}_{\times}$ with the measure $H^{\calM_{all}}_{\times}$ as  ``the Haar random unitary with the conservation law of $X$.''

Additionally, for the later convenience, we also define the following subspace of $\calM_{all}$:
\begin{align}
\calM_{s}:=\{s,s+1,...,N+k-s\},
\end{align}
and the following products of Haar measures and unitary groups
\begin{align}
H^{\calM_{s}}_{\times}&:=\times_{m\in\calM_s}H^{(m)},\\
\calU^{\calM_{s}}_{\times}&:=\times_{m\in\calM_s}\calU^{(m)}.
\end{align}

In this subsection, hereafter we study the property of the Haar random unitaries with the conservation law of $X$.
In particular, we show that the assumptions used in Thereoms \ref{ex-sc-t1} and \ref{va-sc-t2} actually hold in the above Haar random unitary model.
For clarity, below we illustrate several propositions for specific parameter regions that are derived from the results obtained in this subsection. We stress that the results given in this subsection themselves are more general than the following propositions.
\begin{description}
    \item[Proposition 1] When the parameters $N$, $k$ and $\epsilon$ satisfy $N\ge10^3$, $k\le10N$ and $1/(N+k)^3\le\epsilon\le1/(N+k)^2$, and when the initial state $\rho_B$ of the black hole $B$ is the maximally mixed state, the following relation holds for an arbitrary state $\rho$ on $A$ satisfying $M\ge 1/(N+k)$:
\eq{
\mathrm{Prob}_{U\sim H^{\calM_{all}}_{\times}}\left[x_{A'}(\rho,\rho_B,U)\approx_{\frac{M\gamma\epsilon}{2}}(x_A(\rho)+x_B(\rho_B))\times(1-\gamma)\right]\ge 1-e^{-(N+k)^2}\label{item-1}
}
Therefore, the relation \eqref{ex-sc-eq}, which is the assumption used in Theorem \ref{ex-sc-t1}, holds with the probability larger than $1-e^{-(N+k)^2}$ in this case.
\item[Proposition 2] When the parameters $N$, $k$ and $\epsilon$ satisfy $N\ge10^3$, $47< k\le10N$, and $1/(N+k)^3\le\epsilon\le1/(N+k)^2$, and when $\rho_A$ satisfies $M\ge 1/(N+k)$ and its support is included by $\calH^{23\le x\le k-23}_A$ ($\calH^{23\le x\le k-23}_A$ is the subspace of $\calH_A$ spanned by eigenstates of $X_A$ whose eigenvalues are larger than 13 and lower than $k-13$), the following relation holds for an arbitrary state $\rho_B$ on $B$:
\eq{
\mathrm{Prob}_{U\sim H^{\calM_{all}}_{\times}}\left[x_{A'}(\rho_A,\rho_B,U)\approx_{\frac{M\gamma\epsilon}{4}}(x_A(\rho_A)+x_B(\rho_B))\times(1-\gamma)\right]\ge 1-e^{-(N+k)^2}\label{item-2}
}
Therefore, the relation \eqref{ex-sc-eq}, which is the assumption used in Theorem \ref{ex-sc-t1}, holds with the probability larger than $1-e^{-N^2}$ in this case.

\item[Remark on Proposition 2] Note that the support of $\rho_{A,\star}:=(\rho^{\max}_{k-24}+\rho^{\max}_{24})/2$ ($\rho^{\max}_{x}$ is the maximally mixed state of the eigenspace of $X_A$ corresponding to eigenvalue $x$) is included in $\calH^{23\le x\le k-23}_A$.
Therefore, the following inequality holds:
\eq{
\frac{1-\epsilon}{1+\epsilon}\frac{\frac{k}{2}-24}{2(N+2k)}\le\delta
}
Hence, when $N\ge10^3$, $10^3\le k\le10N$, regardless of the initial state of the black hole, the inequality \eqref{foggyineq} in the main text holds for a typical dynamics $U$ on $AB$:
\eq{
\frac{\rm const.}{1+N/2k}\le\delta
}
where ${\rm const.}$ is a positive number larger than $1/9$.

\item[Proposition 3] When the parameters $N$, $k$, and $\epsilon$ satisfy $N\ge 10^3$ and $1/(N+k)^3\le\epsilon\le1/(N+k)^2$, the following two relations for arbitrary  $\ket{i,a}$ and $\ket{j,b}$ of eigenstates of $X_A$ and $X_B$ whose eigenvalues are $i$ and $j$ satisfying $18\le i+j\le N+k-18$:
\eq{
\mathrm{Prob}_{U\sim H^{\calM_{all}}_{\times}}\left[V_{\rho'_{\alpha'|i,a,j,b,U}}(X_{\alpha'})\le(1+\epsilon)\frac{\min\{(1-\gamma)(N+k),\gamma(N+k)\}}{4} \right]
&\ge 1-4\exp\left(-10(N+k)\right),
\label{item-3A}\\
\mathrm{Prob}_{U\sim H^{\calM_{all}}_{\times}}\left[x_{A'}(\ket{i,a},\ket{j,b},U)\approx_{\epsilon}(x_A(i)+x_B(j))\times(1-\gamma)\right]&\ge 1-e^{-(N+k)^2}\label{item-3B}
}
where $\alpha'=A',\enskip B'$. Therefore, in the above parameter region of $N$, $k$ and $\epsilon$,  \eqref{RstarS} and \eqref{Rstar1.5S} which are the assumptions used in Theorem \ref{va-sc-t2} actually hold for $x=18$ with very high probability.

\item[Remark on Proposition 3] In the above parameter region, due to the inequalities \eqref{item-3A} and \eqref{item-3B}, whenever the initial states $\rho_A$ and $\rho_B$ satisfy \eqref{assumVB} and at least one of \eqref{Rstar1.75S} and \eqref{Rstar1.75SA} for $x=18$, the inequality \eqref{Rstar2S} (=\eqref{Rstar2} in the main text) holds.
We give two examples of such $\rho_A$ and $\rho_B$ and show that \eqref{doda} in the main text actually holds for the examples.
The first example is that $\rho_A$ is  $\rho_{A,\star2}:=(\rho^{\max}_{k-19}+\rho^{\max}_{19})/2$ and $\rho_B$ is an arbitrary non-zero temperature Gibbs state $e^{-\beta X_B}/\Tr[e^{-\beta X_B}]$ for $0\le\beta<\infty$.
Due to $X_{B}=\sum_{i\in B}X_{i}$, the Gibbs state $e^{-\beta X_B}/\Tr[e^{-\beta X_B}]$ for $0\le\beta<\infty$ clearly satisfies \eqref{assumVB}.
Also, $\rho_{A,\star2}$ satisfies \eqref{Rstar1.75SA} for $x=18$. 
Therefore, for these initial states the inequality \eqref{Rstar2S} holds, and thus we can derive the following inequality from \eqref{SIQ2}:
\eq{
\frac{1-\epsilon}{1+\epsilon}\frac{\gamma (\frac{k}{2}-19)}{2(\gamma k+2\sqrt{\gamma(N+k^2)})}\le\delta
}
By setting $k=\sqrt{N}$ and $N>10^{10}$, the inequality \eqref{doda} in the main text holds for typical dynamics $U$:
\eq{
\frac{\mathrm{const.}}{1+\frac{2\sqrt{2}}{\sqrt{\gamma}}}\le\delta,\label{doda_SS}
}
where $\mathrm{const.}$ is a positive constant larger than $0.24$. We can make this constant larger than $0.48$ by using \eqref{R-SIQ2} instead of \eqref{SIQ2}.

The second example is given for the case of $k\le 10N$ (Note that Proposition 3 does not require this restriction).
The example is that $\rho_A$ is  $\rho_{A,\star3}:=(\rho^{\max}_{k}+\rho^{\max}_{0})/2$ and $\rho_B$ is an enoughly-high temperature Gibbs state $e^{-\beta X_B}/\Tr[e^{-\beta X_B}]$ for $0\le\beta\le\log2$.
Again, due to $X_{B}=\sum_{i\in B}X_{i}$, the Gibbs state satisfies \eqref{assumVB}.
And $\rho_B$ satisifes \eqref{Rstar1.75S} for $x=18$.
(\textit{Proof:} Note the following inequalities:
\eq{
\left\|\frac{P^{B}_x\rho_BP^B_x}{\Tr[(\rho_B)P^{B}_x]}-\rho_B\right\|_1&
=\sum^{18}_{m=0}\frac{e^{-\beta m}+e^{-\beta(N-m)}}{(e^{-\beta}+1)^N}\times_{N}C_m,\\
\frac{\epsilon^2}{3(N+k)^2}&\ge\frac{1}{3\times11^8\times N^8}.
}
Here we used $k\le 10N$ and $1/(N+k)^3\le\epsilon$ in the second inequality.
Therefore, to show that $\rho_B$ satisifes \eqref{Rstar1.75S} for $x=18$, we only have to show
\eq{
3\times11^8\times N^8\sum^{18}_{m=0}\frac{e^{-\beta m}+e^{-\beta(N-m)}}{(e^{-\beta}+1)^N}\times_{N}C_m\le1.
}
Due to the left-hand side of the above inequality is monotonically decreasing with $N$ for $N\ge10^3$, and is monotonically increasing with $\beta$ for $0\le\beta\le\log2$. Therefore, because of 
\eq{
\left.\left(3\times11^8\times N^8\sum^{18}_{m=0}\frac{e^{-\beta m}+e^{-\beta(N-m)}}{(e^{-\beta}+1)^N}\times_{N}C_m\right)\right|_{\beta=\log2,N=10^3}\approx2\times10^{-111},
}
$\rho_B$ satisifes \eqref{Rstar1.75S} for $x=18$.
\textit{Proof end})
Hence, for these initial states the inequality \eqref{Rstar2S} holds, and thus we can derive the following inequality from \eqref{SIQ2}:
\eq{
\frac{1-\epsilon}{1+\epsilon}\frac{\gamma k}{4(\gamma k+2\sqrt{\gamma(N+k^2)})}\le\delta
}
By setting $k=\sqrt{N}$ and $N>10^{3}$, the inequality \eqref{doda} in the main text holds for typical dynamics $U$ again, where the positive constant ${\rm const.}$ is larger than $0.24$. Again, we can make this constant larger than $0.48$ by using \eqref{R-SIQ2} instead of \eqref{SIQ2}.
\end{description}

Now, we show our results giving the above propositions one by one. 
The first theorem guarantees that if we take average over all Haar random unitary, the expection values of  $X_{A'}$ and $X_{B'}$ are in proportion to the number of qubits in $A'$ and $B'$, respectively.
\begin{theorem}\label{ex-sc-t2}
For the quantity $x_{\alpha}$ $(\alpha=A,B,A')$ in Theorem \ref{ex-sc-t1} and arbitrary $\rho$ and $\rho_B$ on $A$ and $B$, the following equality holds:
\begin{align}
\overline{x_{A'}(\rho,\rho_B,U)}=(x_{A}(\rho)+x_{B}(\rho_B))\frac{l}{N+k},\label{alleq}
\end{align}
where $\overline{f(U)}$ is the average of the function $f$ with the product Haar measure $H^{\calM_{all}}_{\times}$.
Additionally, when the support of $\rho\otimes\rho_B$ is included in the subspace $\calH^{\calM_s}:=\otimes_{m\in\calM_s}\calH^{(m)}$, the following equality holds:
\begin{align}
\overline{x_{A'}(\rho,\rho_B,\tilde{U})}|_{H^{\calM_s}_{\times}}=(x_{A}(\rho)+x_{B}(\rho_B))\frac{l}{N+k},\label{seq}
\end{align}
where $\tilde{U}$ is a unitary which is described as $\tilde{U}=(\oplus_{m\in\calM_{s}}U^{(m)})\oplus(\oplus_{n\not\in\calM_s}I^{(m)})$ where $U^{(m)}\in\calU^{(m)}$, and $\overline{f(\tilde{U})}|_{H^{\calM_s}_{\times}}$ is the average of the function $f$ with the product Haar measure $H^{\calM_s}_{\times}$.
\end{theorem}

\begin{proof}
We refer to the state of the $h$-th qubit in $A'B'$ after $U$ as $\rho^{f}_{h}$.
The state $\rho^{f}_{h}$ satisfies
\begin{align}
\rho^{f}_{h}=\Tr_{\lnot h}[U(\rho\otimes\rho_B)U^\dagger],
\end{align}
where $\Tr_{\lnot h}$ is the partial trace of the qubits other than the $h$-th qubit.
Note that the following equality holds:
\begin{align}
\overline{x_{A'}(\rho,\rho_B,U)}=\sum_{h\in A'}\ex{X_h}_{\overline{\rho^{f}_{h}}},
\end{align}
where $X_h$ is the operator of $X$ on the $h$-th qubit.
Therefore, in order to show \eqref{alleq}, we only have to show
\begin{align}
\overline{\rho^{f}_{h}}=\overline{\rho^{f}_{h'}}\enskip\forall h,h'.\label{alleqpre}
\end{align}
To show \eqref{alleqpre}, we note that the swap gate $S_{h,h'}$ between the $h$-th and the $h'$-th qubits can be written in the following form:
\begin{align}
S_{h,h'}=\oplus_{m\in\calM_{all}}S^{(m)}_{h,h'},
\end{align}
where each $S^{(m)}_{h,h'}$ is a unitary on $\calH^{(m)}$.
Therefore, for any $U\in \calU^{\calM_{all}}_{\times}$, the unitary $S_{h,h'}U$ also satisfies $S_{h,h'}U\in \calU^{\calM_{all}}_{\times}$.
With using this fact and the invariance of the Haar measure, we can derive \eqref{alleqpre} as follows:
\begin{align}
\overline{\rho^{f}_{h}}&=\Tr_{\lnot h}[\int_{H^{\calM_{all}}_{\times}} d\mu U(\rho\otimes\rho_B)U^\dagger]
\nonumber\\
&=\Tr_{\lnot h'}[S_{h,h'}\int_{H^{\calM_{all}}_{\times}} d\mu U(\rho\otimes\rho_B)U^\dagger S^\dagger_{h,h'}]\nonumber\\
&=\Tr_{\lnot h'}[\int_{H^{\calM_{all}}_{\times}} d\mu S_{h,h'}U(\rho\otimes\rho_B)(S_{h,h'}U)^\dagger ]\nonumber\\
&=\Tr_{\lnot h'}[\int_{H^{\calM_{all}}_{\times}} d\mu U'(\rho\otimes\rho_B)U'^\dagger]\nonumber\\
&=\overline{\rho^{f}_{h'}}
\end{align}
Therefore, we have obtained \eqref{alleq}. 

We can also derive \eqref{seq} in a very similar way.
For an arbitrary unitary $V\in\calU^{\calM_{all}}_{\times}$, let us define $\tilde{V}\in\calU^{\calM_{s}}_{\times}$ as follows:
\begin{align}
\tilde{V}=\left(\oplus_{m\in\calM_s}V^{(m)}\right)\oplus\left(\oplus_{n\not\in\calM_s}I^{(m)}\right),\label{tilV}
\end{align}
where $\{V^{(m)}\}$ are defined as $V=\oplus_{m\in\calM_{all}}V^{(m)}$.
Note that when $\rho\otimes\rho_B$ is included in $\calH^{\calM_s}$, we can substitute $\tilde{S}_{h,h'}$ and $\tilde{U}$ for $S_{h,h'}$ and $U$ in the above derivation of \eqref{alleq}.
By performing this substitution, we obtain \eqref{seq}.
\end{proof}

In the next theorem, we show that under a natural assumption, the value of $x_{A'}(\rho,\rho_B,U)$ with a Haar random unitary $U$ is almost equal to its average with very high probability. 
\begin{theorem}\label{ev-prob-t}
For the quantity $x_{A'}$ in Theorem \ref{ex-sc-t1}, an arbitrary positive number $t$, an arbitrary non-negative integer $s$, and arbitrary states $\rho$ and $\rho_B$ on $A$ and $B$ which satisfy that the support of $\rho\otimes\rho_B$ is included in the subspace $\calH^{\calM_s}:=\otimes_{m\in\calM_s}\calH^{(m)}$ for $\calM_s:=\{s,s+1,...,N+k-s\}$, the following relation holds:
\begin{align}
\mathrm{Prob}_{U\sim H^{\calM_{all}}_{\times}}\left[|x_{A'}(\rho,\rho_B,U)-\overline{x_{A'}(\rho,\rho_B,U)}|>t\right]\le 2\exp\left(-\frac{(_{N+k}C_s-2)t^2}{48l^2}\right).\label{ev-prob}
\end{align} 
Here $\mathrm{Prob}_{U\sim H^{\calM_{all}}_{\times}}[...]$ is the probability that the event (...) happens when $U$ is chosen from $\calU^{\calM_{all}}_{\times}$ with the measure $H^{\calM_{all}}_{\times}$.
\end{theorem}
This theorem implies that when $U$ is a typical Haar random unitary with the conservation law of $X$, the assumption \eqref{ex-sc-eq} actually holds with very high probability.
To be concrete, we can derive the following corollary from Theorem \ref{ev-prob-t}:
\begin{corollary}\label{corollary:master_ex}
Let us take arbitrary states $\rho$ on $A$ and $\rho_B$ on $B$ and an arbitrary non-negative integer $s$.
We define the following $\rho^{\calM_s}_{B}$ from $\rho_B$ and refer to the distance between $\rho^{\calM_s}_{B}$ and $\rho_B$ as $\epsilon_{B,\calM_s}$:
\eq{
\rho^{\calM_s}_B&:=\frac{P^{\calM_s}\rho_BP^{\calM_s}}{\Tr[P^{\calM_s}\rho_B]},\\
\epsilon_{B,\calM_s}&:=\left\|\rho^{\calM_s}_B-\rho_B\right\|_1,
}
where $P^{\calM_s}$ is the projection to the subspace $\calH^{\calM_s}$.
Then, the following relation holds:
\eq{
\mathrm{Prob}_{U\sim H^{\calM_{all}}_{\times}}\left[x_{A'}(\rho,\rho_B,U)\approx_{t+(2N+k)\epsilon_{B,\calM_s}}(x_A(\rho)+x_B(\rho_B))\times(1-\gamma)\right]\ge 1-2\exp\left(-\frac{(_{N+k}C_s-2)t^2}{48(N+k)^2}\right)\label{C1-Z}
}
\end{corollary}

\begin{proofof}{Corollary \ref{corollary:master_ex}}
For an arbitrary unitary $U$ on $AB$ satisfying the conservation law, we obtain
\eq{
|x_{A'}(\rho,\rho_B,U)-x_{A'}(\rho,\rho^{\calM_s}_B,U)|&=|\Tr[X_{A'}(U\rho\otimes\rho_BU^\dagger-U\rho\otimes\rho^{\calM_s}_BU^\dagger)]|\non
&\le\|X_{A'}\|_\infty\epsilon_{B,\calM_s}\non
&\le(N+k)\times\epsilon_{B,\calM_s}\label{C1-A}
}
Therefore, $x_{A'}(\rho,\rho_B,U)\approx_{(N+k)\times\epsilon_{B,\calM_s}}x_{A'}(\rho,\rho^{\calM_s}_B,U)$ always holds.
We also obtain
\eq{
|x_{B}(\rho_B)-x_{B}(\rho^{\calM_s}_B)|&=|\Tr[X_B(\rho_B-\rho^{\calM_s}_B)]|\non
&\le\|X_B\|_\infty\times\epsilon_{B,\calM_s}\non
&\le N\times\epsilon_{B,\calM_s}.\label{C1-B}
}

Since the support of $\rho\otimes\rho^{\calM_s}_{B}$ is on $\calH^{\calM_s}$, Theorem \ref{ev-prob-t} and \eqref{alleq} imply
\eq{
\mathrm{Prob}_{U\sim H^{\calM_{all}}_{\times}}\left[x_{A'}(\rho,\rho^{\calM_s}_B,U)\approx_{t}(x_A(\rho)+x_B(\rho^{\calM_s}_B))\times(1-\gamma)\right]\ge 1-2\exp\left(-\frac{(_{N+k}C_s-2)t^2}{48(N+k)^2}\right).\label{C1-C}
}
Here we used $l^2\le(N+k)^2$.
Combining \eqref{C1-A}--\eqref{C1-C}, we obtain \eqref{C1-Z}.
\end{proofof}

\begin{derivationof}{Propositions 1 and 2}
As shown below, Theorem \ref{ev-prob-t} and Corollary \ref{corollary:master_ex} imply Propositions 1 and 2, i.e., \eqref{item-1} and \eqref{item-2}. In other words, under natural settings, for the Haar random unitary with the conservation law of $X$, the assumption \eqref{ex-sc-eq} used in Theorem \ref{ex-sc-t1} holds with very high probability.
First, we derive \eqref{item-1}.
We set the initial state $\rho_B$ on $B$ is the maximally mixed state and set the parameters $N$, $k$ and $\epsilon$ satisfying $N\ge10^3$, $k\le10N$ and $1/(N+k)^2\ge\epsilon\ge1/(N+k)^3$.
We also take an arbitrary state $\rho$ on $A$ satisfying $M\ge 1/(N+k)$.
To derive \eqref{item-1} from \eqref{C1-Z}, let us set $s$ and $t$ in \eqref{C1-Z} as $s=23$ and $t=\frac{M\gamma\epsilon}{4}$, and evaluate $\frac{(_{N+k}C_{s}-2)t^2}{48(N+k)^2}$ and $\epsilon_{B,\calM_s}$. We remark that to obtain \eqref{item-1}, we only have to derive the following inequalities:
\eq{
\epsilon_{B,\calM_s}&\le \frac{M\gamma\epsilon}{8(N+k)},\label{DP1-1}\\
\frac{(_{N+k}C_{s=23}-2)t^2}{48(N+k)^2}-\log 2&\ge(N+k)^2.\label{DP1-2}
}
Let us derive \eqref{DP1-1}. Due to the definition of $\epsilon_{B,\calM_s}$, $\epsilon_{B,\calM_s}=2\sum^{22}_{m=0}\frac{_{N}C_m}{2^{N}}$.
Because of $M\ge\frac{1}{N+k}$, $\gamma\ge\frac{1}{N+k}$, $\epsilon\ge\frac{1}{(N+k)^3}$ and $k\le10N$, we obtain $\frac{M\gamma\epsilon}{8(N+k)}\ge\frac{1}{8\times11^6\times N^6}$.
Therefore, to obtain \eqref{DP1-1}, it is sufficient to show 
\eq{
\sum^{22}_{m=0}\frac{_{N}C_m}{2^{N-1}}\le\frac{1}{8\times11^6\times N^6}.\label{DP1-3}
}
for $N\ge10^3$. Because $8\times11^6\times N^6\times\sum^{22}_{m=0}\frac{_{N}C_m}{2^{N-1}}$ is monotonically decreasing with $N$ for $N\ge10^3$, and because $(8\times11^6\times N^6\times\sum^{22}_{m=0}\frac{_{N}C_m}{2^{N-1}})|_{N=10^3}$ is approximately equal to $5\times10^{-230}$, we obtain \eqref{DP1-3}, and thus the inequality \eqref{DP1-1} holds.

Next, we derive \eqref{DP1-2}.
Because of $M\ge\frac{1}{N+k}$, $\gamma\ge\frac{1}{N+k}$, $\epsilon\ge\frac{1}{(N+k)^3}$, the inequality $\frac{(_{N+k}C_s-2)t^2}{48(N+k)^2}\ge\frac{{}_{N+k}C_s-2}{48\times16\times(N+k)^{12}}$ holds.
Therefore, to show \eqref{DP1-2}, it is sufficient to show
\eq{
1\le\frac{{}_{N+k}C_{s=23}-2}{48\times16\times(N+k)^{14}}-\frac{\log2}{(N+k)^2}.\label{DP1-4}
}
Because the right-hand side of the above is monotonically increasing with $N+k$ for $N+k\ge1001$, and because of $(\frac{{}_{N+k}C_s-2}{48\times16\times(N+k)^{14}}-\frac{\log2}{(N+k)^2})|_{N+k=1001}\approx39$, we obtain \eqref{DP1-4} and thus \eqref{DP1-2} holds.
Therefore, we have shown \eqref{item-1}.

Next, we show \eqref{item-2}. 
Let us set the parameters $N$, $k$ and $\epsilon$ satisfying $N\ge10^3$, $47< k\le10N$ and $1/(N+k)^3\le\epsilon\le1/(N+k)^2$.
We also take an arbitrary state $\rho_B$ on $B$ and an arbitrary state $\rho_A$ such that $M\ge\frac{1}{N+k}$ and the support of $\rho_A$ is included in $\calH^{23\le x\le k-23}_A$.
To show that \eqref{item-2} holds under this condition, we remark that when the support of $\rho_A$ is included in $\calH^{23\le x\le k-23}_A$, the support of $\rho_A\otimes\rho_B$ is included in $\calH^{\calM_{s=23}}$.
Therefore, to obtain \eqref{item-2} from Theorem \ref{ev-prob-t}, it is sufficient to set $t=M\gamma\epsilon/4$ and $s=23$ and show the inequality 
\eq{
\frac{(_{N+k}C_{s=23}-2)t^2}{48l^2}-\log 2&\ge(N+k)^2.
}
Due to $l\le N+k$, this inequality is easily obtained from \eqref{DP1-2}.
Therefore, we obtain \eqref{item-2}.
\end{derivationof}

Now, let us show Theorem \ref{ev-prob-t}.
To show it, we introduce two definitions and a theorem.
\begin{definition}
Let $f$ is a real-valued function on a metric space $(X,d)$.
When $f$ satisfies the following relation for a real positive constant $L$, then $f$ is called $L$-Lipschitz:
\begin{align}
L=\sup_{x\ne y \in X}\frac{|f(x)-f(y)|}{d(x,y)}.
\end{align}
\end{definition}

\begin{definition}
Let $\calU^{M}_{\times}$ be a product of unitary groups $\times^{M}_{i=1}\calU(d_i)$, where each $\calU(d_i)$ is the unitary group of all unitary operations on a $d_i$-dimensional Hilbert space.
For $U=\oplus^{M}_{i=1}U_{i}\in\calU^{M}_{\times}$ and $V=\oplus^{M}_{i=1}V_{i}\in\calU^{M}_{\times}$, the $L_2$-sum $D(U,V)$ of the Hilbert-Schmidt norms on $\calU^{M}_{\times}$ is defined as 
\begin{align}
D(U,V):=\sqrt{\sum^{M}_{i=1}\|U_i-V_i\|^2_2},
\end{align}
where $\|...\|_2:=\sqrt{\Tr[(...)(...)^\dagger]}$. 
\end{definition}

\begin{theorem}[Corollary 3.15 in Ref. \cite{Meckes}]\label{Eli}
Let $\calU^{M}_{\times}$ be a product of unitary groups $\times^{M}_{i=1}\calU(d_i)$, where each $\calU(d_i)$ is a unitary group of all unitary operations on a $d_i$-dimensional Hilbert space.
Let $\calU^{M}_{\times}$ be equipped with the $L_2$-sum of Hilbert-Schmidt norms, and $H^{M}_{\times}:=\times^{M}_{i=1}H_i$ where each $H_i$ is the Haar measure on $\calU(d_i)$.
Suppose that a real-valued function $f$ on $\calU^{M}_{\times}$ is $L$-Lipschitz.
Then, for arbitrary $t>0$,
\begin{align}
\mathrm{Prob}[f(U)\ge\overline{f(U)}+t]\le \exp\left(-\frac{(d_{\min}-2)t^2}{12L^2}\right),
\end{align} 
where $d_{\min}:=\min\{d_{1},...,d_{M}\}$.
\end{theorem}

From Theorem \ref{Eli}, we can easily derive Theorem \ref{ev-prob-t}:
\begin{proofof}{Theorem \ref{ev-prob-t}}
Since the support of $\rho\otimes\rho_B$ is included in the subspace $\calH^{\calM_s}:=\otimes_{m\in\calM_s}\calH^{(m)}$, the following relation holds for arbitrary $U\in\calU^{\calM_{all}}_{\times}$:
\begin{align}
x_{A'}(\rho,\rho_B,U)=x_{A'}(\rho,\rho_B,\tilde{U}),
\end{align}
where $\tilde{U}$ defined from $U$ by \eqref{tilV}.
Therefore, we only have to show 
\begin{align}
\mathrm{Prob}_{\tilde{U}\sim H^{\calM_{s}}_{\times}}\left[|x_{A'}(\rho,\rho_B,\tilde{U})-\overline{x_{A'}(\rho,\rho_B,\tilde{U})}|_{H^{\calM_s}_{\times}}|>t\right]\le 2\exp\left(-\frac{(_{N+k}C_s-2)t^2}{48l^2}\right).\label{ev-prob'}
\end{align} 
Note that $\min_{m\in\calM_s}\mathrm{dim}\calH^{(m)}=_{N+k}C_s$.
Therefore, due to Theorem \ref{Eli}, to show \eqref{ev-prob'}, it is sufficient to show that $x_{A'}(\rho,\rho_B,\tilde{U})$ is $2l$-Lipschitz.

To show that $x_{A'}(\rho,\rho_B,\tilde{U})$ is $2l$-Lipschitz, let us take two unitary operations $\hat{U}\in\calU^{\calM_s}_{\times}$ and $\hat{V}\in\calU^{\calM_s}_{\times}$.
We evaluate $|x_{A'}(\rho,\rho_B,\hat{U})-x_{A'}(\rho,\rho_B,\hat{V})|$ as follows:
\begin{align}
|x_{A'}(\rho,\rho_B,\hat{U})-x_{A'}(\rho,\rho_B,\hat{V})|
&=|\Tr[X_{A'}(\hat{U}(\rho\otimes\rho_B)\hat{U}^{\dagger}-\hat{V}(\rho\otimes\rho_B)\hat{V}^{\dagger})]|\nonumber\\
&\le\calD_{X_{A'}}\|\hat{U}(\rho\otimes\rho_B)\hat{U}^{\dagger}-\hat{V}(\rho\otimes\rho_B)\hat{V}^{\dagger}\|_{1}\nonumber\\
&\le l\|\hat{U}(\rho\otimes\rho_B)\hat{U}^{\dagger}-\hat{V}(\rho\otimes\rho_B)\hat{V}^{\dagger}\|_{1}.
\end{align}
Therefore, in order to show that $x_{A'}(\rho,\rho_B,\tilde{U})$ is $2l$-Lipschitz, we only have to show 
\begin{align}
\|\hat{U}(\rho\otimes\rho_B)\hat{U}^{\dagger}-\hat{V}(\rho\otimes\rho_B)\hat{V}^{\dagger}\|_{1}\le2D(\hat{U},\hat{V}).\label{legofinal}
\end{align}
To show \eqref{legofinal}, we take a purification of $\rho\otimes\rho_B$, and refer to it as $\ket{\psi_{ABQ}}$.
Due to the monotonicity of the 1 norm and $\|\phi-\psi\|_1=2D_F(\phi,\psi)$ for any pure $\phi$ and $\psi$, 
\begin{align}
\|\hat{U}(\rho\otimes\rho_B)\hat{U}^{\dagger}-\hat{V}(\rho\otimes\rho_B)\hat{V}^{\dagger}\|_{1}&\le\|\hat{U}\psi_{ABQ}\hat{U}^{\dagger}-\hat{V}\psi_{ABQ}\hat{V}^{\dagger}\|_{1}\nonumber\\
&=2\sqrt{1-F^2((\hat{U}\otimes1_{Q})\psi_{ABQ}(\hat{U}\otimes1_{Q})^{\dagger},(\hat{V}\otimes1_{Q})\psi_{ABQ}(\hat{V}\otimes1_{Q})^{\dagger})}\nonumber\\
&\le2\sqrt{2(1-F((\hat{U}\otimes1_{Q})\psi_{ABQ}(\hat{U}\otimes1_{Q})^{\dagger},(\hat{V}\otimes1_{Q})\psi_{ABQ}(\hat{V}\otimes1_{Q})^{\dagger}))}\nonumber\\
&=2\sqrt{2-2|\bra{\psi_{ABQ}}(\hat{U}\otimes1_{Q})^{\dagger}(\hat{V}\otimes1_{Q})\ket{\psi_{ABQ}}|}\nonumber\\
&\le2\sqrt{2-\bra{\psi_{ABQ}}(\hat{U}\otimes1_{Q})^{\dagger}(\hat{V}\otimes1_{Q})\ket{\psi_{ABQ}}-\bra{\psi_{ABQ}}(\hat{V}\otimes1_{Q})^{\dagger}(\hat{U}\otimes1_{Q})\ket{\psi_{ABQ}}}\nonumber\\
&=2\sqrt{\bra{\psi_{ABQ}}((\hat{U}\otimes1_{Q})-(\hat{V}\otimes1_{Q}))^{\dagger}((\hat{U}\otimes1_{Q})-(\hat{V}\otimes1_{Q}))\ket{\psi_{ABQ}}}\nonumber\\
&\le\|((\hat{U}\otimes1_{Q})-(\hat{V}\otimes1_{Q}))\ket{\psi_{ABQ}}\|_{2}\nonumber\\
&\le2\|\hat{U}-\hat{V}\|_2\|\rho\otimes\rho_B\|^{1/2}_{\infty}.\label{lego1}
\end{align}
In the final line, we use
\begin{align}
\|((\hat{U}\otimes1_{Q})-(\hat{V}\otimes1_{Q}))\ket{\psi_{ABQ}}\|^2_{2}
&=\Tr[((\hat{U}\otimes1_{Q})-(\hat{V}\otimes1_{Q}))\ket{\psi_{ABQ}}\bra{\psi_{ABQ}}((\hat{U}\otimes1_{Q})-(\hat{V}\otimes1_{Q}))^\dagger]\nonumber\\
&=\Tr[(\hat{U}-\hat{V})(\rho\otimes\rho_B)(\hat{U}-\hat{V})^\dagger]\nonumber\\
&\le\|\rho\otimes\rho_B\|_{\infty}\|(\hat{U}-\hat{V})(\hat{U}-\hat{V})^\dagger\|_1\nonumber\\
&\le\|\rho\otimes\rho_B\|_{\infty}\|\hat{U}-\hat{V}\|^2_2,\label{lego2}
\end{align}
where we use the H\"{o}lder inequality in the final line.
Due to $\|M_1\oplus M_2-M'_1\oplus M'_2\|^2_2=\|M_1-M'_1\|^2_2+\|M_2-M'_2\|^2_2$ and the definition of $L_2$-sum, we can show $\|\hat{U}-\hat{V}\|_2=D(\hat{U},\hat{V})$ as follows:
\begin{align}
\|\hat{U}-\hat{V}\|^2_2&=\sum_{m\in\calM_s}\|\hat{U}^{(m)}-\hat{V}^{(m)}\|^2_2\nonumber\\
&=D(\hat{U},\hat{V})^2,\label{lego3}
\end{align}
where $\hat{U}^{(m)}$ and $\hat{V}^{(m)}$ are defined as $\hat{U}=(\oplus_{m\in\calM_s}\hat{U}^{(m)})\oplus(\oplus_{m\not\in\calM_s}I^{(m)})$ and $\hat{V}=(\oplus_{m\in\calM_s}\hat{V}^{(m)})\oplus(\oplus_{m\not\in\calM_s}I^{(m)})$.
Combining \eqref{lego1}, \eqref{lego3} and $\|\rho\otimes\rho_B\|_{\infty}\le1$, we obtain \eqref{legofinal}.

\end{proofof}

Next, we prove that the assumption \eqref{RstarS} (=\eqref{Rstar} in the main text) is valid for a typical Haar random unitary with the energy conservation.
To do so, we use the following theorems, which correspond to the second-order variation of Theorem \ref{ex-sc-t2} and Theorem \ref{ev-prob-t}.
\begin{theorem}\label{ex-sc-tsec}
For arbitrary eigenstates $\ket{i,a}$ and $\ket{j,b}$ of $X_A$ and $X_B$ whose eigenvalues are $i$ and $j$, the following inequality hold:
\begin{align}
V_{\overline{\rho'_{\alpha'|i,a,j,b,U}}}(X_{\alpha'})\le\frac{\min\{(1-\gamma)(N+k),\gamma(N+k)\}}{4} \enskip (\alpha'=A',\enskip B'),\label{alleq_sec}
\end{align}
where $\overline{f(U)}$ is the average of the function $f$ with the product Haar measure $H^{\calM_{all}}_{\times}$.
Additionally, when $\ket{i,a}\otimes\ket{j,b}$ is included in the subspace $\calH^{\calM_s}:=\otimes_{m\in\calM_s}\calH^{(m)}$, the following equality holds:
\begin{align}
V_{\overline{\rho'_{\alpha'|i,a,j,b,\tilde{U}}}|_{H^{\calM_s}_{\times}}}(X_{\alpha'})\le\frac{\min\{(1-\gamma)(N+k),\gamma(N+k)\}}{4} \enskip (\alpha'=A',\enskip B'),\label{seq_sec}
\end{align}
where $\tilde{U}$ is a unitary which is described as $\tilde{U}=(\oplus_{m\in\calM_{s}}U^{(m)})\oplus(\oplus_{n\not\in\calM_s}I^{(m)})$ where $U^{(m)}\in\calU^{(m)}$, and $\overline{f(\tilde{U})}|_{H^{\calM_s}_{\times}}$ is the average of the function $f$ with the product Haar measure $H^{\calM_s}_{\times}$.
\end{theorem}

\begin{theorem}\label{ev-prob-tsec}
For an arbitrary positive number $t$, and arbitrary states $\rho$ and $\rho_B$ on $A$ and $B$ which satisfy that the support of $\rho\otimes\rho_B$ is included in the subspace $\calH^{\calM_s}:=\otimes_{m \in \calM_{s}}\calH^{(m)}$ for $\calM_{s}:=\{s,s+1,...,N+k-s\}$, the following relation holds:
\begin{align}
\mathrm{Prob}_{U\sim H^{\calM_{all}}_{\times}}\left[|x^{(2)}_{A'}(\rho,\rho_B,U)-\overline{x^{(2)}_{A'}(\rho,\rho_B,U)}|>t\right]&\le 2\exp\left(-\frac{(_{N+k}C_s-2)t^2}{48l^4}\right),\label{ev-prob2}\\
\mathrm{Prob}_{U\sim H^{\calM_{all}}_{\times}}\left[|x^{(2)}_{B'}(\rho,\rho_B,U)-\overline{x^{(2)}_{B'}(\rho,\rho_B,U)}|>t\right]&\le 2\exp\left(-\frac{(_{N+k}C_s-2)t^2}{48(\gamma (N+k))^4}\right)\label{ev-prob2B}
\end{align} 
Here $x^{(2)}_{\alpha'}(\rho,\rho_B,U):=\Tr_{\neg \alpha'}[X^2_{\alpha'}U\rho\otimes \rho_BU^\dagger]$ ($(\alpha',\neg\alpha')=(A',B') \mbox{ or } (B',A')$), and $\mathrm{Prob}_{U\sim H^{\calM_{all}}_{\times}}[...]$ is the probability that the event (...) happens when $U$ is chosen from $\calU^{\calM_{all}}_{\times}$ with the measure $H^{\calM_{all}}_{\times}$.
\end{theorem}

Theorems \ref{ex-sc-t2}, \ref{ex-sc-tsec}, and \ref{ev-prob-tsec} guarantee that when $U$ is a typical Haar random unitary with the conservation law of $X$, the assumption \eqref{RstarS} actually holds with very high probability, in the same way as Theorems \ref{ev-prob-t} guarantees that when $U$ is a typical Haar random unitary with the conservation law of $X$, the assumption \eqref{ex-sc-eq} holds with very high probability.
To be concrete, we can derive the following corollary:
\begin{corollary}\label{corollary:master_var}
Let $\ket{i,a}$ and $\ket{j,b}$ be eigenstates of $X_A$ and $X_B$ whose eigenvalues are $i$ and $j$.
When $N\ge 10^3$ and $18\le i+j\le N+k-18$ hold, the following relation holds for an arbitrary real number $\epsilon$ satisfying $1/(N+k)^3\le \epsilon\le1/(N+k)^2$:
\eq{
\mathrm{Prob}_{U\sim H^{\calM_{all}}_{\times}}\left[V_{\rho'_{\alpha'|i,a,j,b,U}}(X_{\alpha'})\le(1+\epsilon)\frac{\min\{(1-\gamma)(N+k),\gamma(N+k)\}}{4} \right]
\ge 1-4\exp\left(-10(N+k)\right),\enskip (\alpha'=A',\enskip B').
\label{CV-Z}
}
\end{corollary}
\begin{derivationof}{Proposition 3}
Now, let us derive Proposition 3 from Theorem \ref{ev-prob-t} and Corollary \ref{corollary:master_var}.
First, Corollary \ref{corollary:master_var} directly implies \eqref{item-3A}.
Let us derive \eqref{item-3B} from Theorem \ref{ev-prob-t}. 
We take the parameters $N$, $k$ and $\epsilon$ satisfying $N\ge10^3$ and $1/(N+k)^3\le\epsilon\le1/(N+k)^2$.
We also set a number $s$ as $s=18$.
To derive \eqref{item-3B}, we only have to set $s$ and $t$ in \eqref{ev-prob} as $s=18$ and $t=\epsilon$ and show
\eq{
(N+k)^2\le\frac{(_{N+k}C_s-2)\epsilon^2}{48l^2}-\log2.\label{DP3-1}
}
Due to $\epsilon\ge1/(N+k)^3$ and $l\le N+k$, to show \eqref{DP3-1}, we only have to show
\eq{
1\le\frac{(_{N+k}C_{18}-2)}{48(N+k)^{10}}-\frac{\log2}{(N+k)^2}.\label{DP3-2}
}
For $N+k\ge1001$, the right-hand side of \eqref{DP3-2} is monotonically increasing with $N+k$. Due to $(\frac{(_{N+k}C_{18}-2)}{48(N+k)^{10}}-\frac{\log2}{(N+k)^2})|_{N+k=1001}\approx2.8\times10^6$, \eqref{DP3-2} holds, and thus \eqref{item-3B} also holds.
\end{derivationof}

\begin{proofof}{Corollary \ref{corollary:master_var}}
Due to $18\le i+j\le N+k-18$, the state $\ket{i,a}\otimes\ket{j,b}$ is in $\calH^{\calM_{s=18}}$.
Therefore, from Theorem \ref{ev-prob-t} and Theorem \ref{ev-prob-tsec}, we obtain the following relations for arbitrary $t_1$ and $t_2$:
\begin{align}
\mathrm{Prob}_{U\sim H^{\calM_{all}}_{\times}}\left[|x^{(2)}_{\alpha'}(\ket{i,a},\ket{j,b},U)-\overline{x^{(2)}_{\alpha'}(\ket{i,a},\ket{j,b},U)}|>t_2\right]&\le 2\exp\left(-\frac{(_{N+k}C_{18}-2)t^2_2}{48(N+k)^4}\right),\\
\mathrm{Prob}_{U\sim H^{\calM_{all}}_{\times}}\left[|x_{\alpha'}(\ket{i,a},\ket{j,b},U)-\overline{x_{\alpha'}(\ket{i,a},\ket{j,b},U)}|>t_1\right]&\le 2\exp\left(-\frac{(_{N+k}C_{18}-2)t^2_1}{48(N+k)^2}\right)
\end{align} 
Therefore, 
\eq{
&\mathrm{Prob}_{U\sim H^{\calM_{all}}_{\times}}\left[
\left(x^{(2)}_{\alpha'}(\ket{i,a},\ket{j,b},U)\approx_{t_2}\overline{x^{(2)}_{\alpha'}(\ket{i,a},\ket{j,b},U)}\right)\land
\left(x_{\alpha'}(\ket{i,a},\ket{j,b},U)\approx_{t_1}\overline{x_{\alpha'}(\ket{i,a},\ket{j,b},U)}\right)
\right]\non
&\ge 1-2\left(\exp\left(-\frac{(_{N+k}C_{18}-2)t^2_2}{48(N+k)^4}\right)+\exp\left(-\frac{(_{N+k}C_{18}-2)t^2_1}{48(N+k)^2}\right)\right),
}
Due to $V_{\rho'_{\alpha'|i,a,j,b,U}}(X_{\alpha'})=x^{(2)}_{\alpha'}(\ket{i,a},\ket{j,b},U)-x_{\alpha'}(\ket{i,a},\ket{j,b},U)^2$ and $x_{\alpha'}(\ket{i,a},\ket{j,b},U)\le N+k$, 
we obtain
\eq{
&\left(x^{(2)}_{\alpha'}(\ket{i,a},\ket{j,b},U)\approx_{t_2}\overline{x^{(2)}_{\alpha'}(\ket{i,a},\ket{j,b},U)}\right)\land
\left(x_{\alpha'}(\ket{i,a},\ket{j,b},U)\approx_{t_1}\overline{x_{\alpha'}(\ket{i,a},\ket{j,b},U)}\right)\non
&\Rightarrow \left(V_{\rho'_{\alpha'|i,a,j,b,U}}(X_{\alpha'})\approx_{t_2+2t_1(N+k)+t^2_1}V_{\overline{\rho'_{\alpha'|i,a,j,b,U}}}(X_{\alpha'})\right)
}
Therefore, we obtain
\begin{align}
&\mathrm{Prob}_{U\sim H^{\calM_{all}}_{\times}}\left[V_{\rho'_{\alpha'|i,a,j,b,U}}(X_{\alpha'})\approx_{t_2+2t_1(N+k)+t^2_1}V_{\overline{\rho'_{\alpha'|i,a,j,b,U}}}(X_{\alpha'})\right]\non
&\ge 1-2\left(\exp\left(-\frac{(_{N+k}C_{18}-2)t^2_2}{48(N+k)^4}\right)+\exp\left(-\frac{(_{N+k}C_{18}-2)t^2_1}{48(N+k)^2}\right)\right),
\end{align} 
Here, let us define $t_1:=\frac{\epsilon}{16(N+k)}$ and $t_2:=\frac{\epsilon}{16}$.
Then, with using $1/(N+k)^2\ge\epsilon\ge 1/(N+k)^3$, we obtain
\begin{align}
\mathrm{Prob}_{U\sim H^{\calM_{all}}_{\times}}\left[V_{\rho'_{\alpha'|i,a,j,b,U}}(X_{\alpha'})\approx_{\epsilon/4}V_{\overline{\rho'_{\alpha'|i,a,j,b,U}}}(X_{\alpha'})\right]
\ge 1-4\exp\left(-\frac{_{N+k}C_{18}-2}{48\times16^2(N+k)^{10}}\right).
\end{align} 
Due to $N\ge10^3$, we obtain the following:
\eq{
\frac{_{N+k}C_{18}-2}{48\times16^2(N+k)^{11}}&\ge \frac{_{10^3}C_{18}-2}{748\times16^2(10)^{33}}\non
&\ge10.
}
Therefore, we obtain
\eq{
\frac{_{N+k}C_{18}-2}{48\times16^2(N+k)^{10}}\ge 10(N+k).
}
Therefore, we obtain
\begin{align}
\mathrm{Prob}_{U\sim H^{\calM_{all}}_{\times}}\left[V_{\rho'_{\alpha'|i,a,j,b,U}}(X_{\alpha'})\approx_{\epsilon/4}V_{\overline{\rho'_{\alpha'|i,a,j,b,U}}}(X_{\alpha'})\right]
\ge 1-4\exp\left(-10(N+k)\right).
\end{align} 
Combining the above inequality, \eqref{alleq_sec} and $\frac{\min\{(1-\gamma)(N+k),\gamma(N+k)\}}{4}\ge1/4$, we obtain \eqref{CV-Z}.

\end{proofof}

\begin{proofof}{Theorem \ref{ex-sc-tsec}}
We firstly remark that 
\eq{
\overline{U\ket{i,a}\bra{i,a}\otimes\ket{j,b}\bra{j,b}U^\dagger}=\rho^{\max}_{i+j}
}
where $\rho^{\max}_{i+j}$ is the maximally mixed state in eigenspace of $X_{A'}+X_{B'}$ whose eigenvalue is $i+j$.
Therefore, 
\eq{
\overline{\rho^{f}_{h,h'|i,a,j,b,U}}&:=\Tr_{\neg(h,h')}[\overline{U\ket{i,a}\bra{i,a}\otimes\ket{j,b}\bra{j,b}U^\dagger}]\non
&=\frac{(i+j)(i+j-1)}{(N+k)(N+k-1)}\ket{1}\bra{1}_h\otimes\ket{1}\bra{1}_{h'}
+\frac{(i+j)(N+k-(i+j))}{(N+k)(N+k-1)}\left(\ket{1}\bra{1}_h\otimes\ket{0}\bra{0}_{h'}+\ket{0}\bra{0}_h\otimes\ket{1}\bra{1}_{h'}\right)\non
&+\frac{(N+k-(i+j))(N+k-(i+j+1))}{(N+k)(N+k-1)}\ket{0}\bra{0}_h\otimes\ket{0}\bra{0}_{h'}
}
where $\Tr_{\neg(h,h')}$ is the partial trace of the qubits other than the $h$-th and $h'$-th qubits, and $\ket{1}_h$ and $\ket{0}_h$ are the excited state and the ground state of $X_h$.

Here we define $x^{(2)}_{\alpha'}(\rho,\rho_B,U):=\Tr_{\neg \alpha'}[X^2_{\alpha'}U\rho\otimes \rho_BU^\dagger]$ ($(\alpha',\neg\alpha')=(A',B') \mbox{ or } (B',A')$), and obtain
\begin{align}
\overline{x^{(2)}_{A'}(\ket{i,a}\bra{i,a},\ket{j,b}\bra{j,b},U)}&=\sum_{h,h'\in A', h\ne h'}\ex{X_h\otimes X_{h'}}_{\overline{\rho^{f}_{h,h'|i,a,j,b,U}}}+\sum_{h\in A'}\ex{X^2_h}_{\overline{\rho^{f}_{h,h'|i,a,j,b,U}}}\non
&=\frac{(i+j)(i+j-1)}{(N+k)(N+k-1)}l(l-1)+\frac{i+j}{N+k}l\non
&=p\left(p+\frac{p-1}{N+k-1}\right)l(l-1)+pl
\end{align}
Here we use $p:=(i+j)/(N+k)$.
Due to Theorem \ref{ex-sc-t2}, we obtain $\overline{x_{A'}(\ket{i,a}\bra{i,a},\ket{j,b}\bra{j,b},U)}=(1-\gamma)(i+j)=lp$.
Therefore, we obtain
\begin{align}
V_{\overline{\rho'_{A'|i,a,j,b,U}}}(X_{A'})&=p\left(p+\frac{p-1}{N+k-1}\right)l(l-1)+pl-p^2l^2\non
&=p(1-p)l-\frac{p(1-p)}{N+k-1}l(l-1)\non
&=p(1-p)l\left(1-\frac{l-1}{N+k-1}\right).
\end{align}
Since $0\le p\le1$ holds, the inequality $p(1-p)\le1/4$ holds.
Therefore, to obtain \eqref{alleq_sec}, we only have to show
\eq{
l\left(1-\frac{l-1}{N+k-1}\right)\le\min\{l,N+k-l\}.
}
Due to $\left(1-\frac{l-1}{N+k-1}\right)\le1$, we only have to show that the following inequality holds for $l\in\{1,2,..,N+k\}$:
\eq{
l\left(1-\frac{l-1}{N+k-1}\right)\le N+k-l.\label{deka1}
}
To show \eqref{deka1}, we show that the following inequality holds for $l\in\{1,2,..,N+k\}$:
\eq{
l\left(2-\frac{l-1}{N+k-1}\right)\le N+k.\label{deka2}
}
To show this, let us define $h(l):=l\left(2-\frac{l-1}{N+k-1}\right)$.
Then,
\eq{
\frac{d h(l)}{d l}&=\left(2-\frac{l-1}{N+k-1}\right)-\frac{l}{N+k-1}\non
&=2-\frac{2l-1}{N+k-1}.
}
Therefore, for $0\le l\le N+k-1$, the function $h(l)$ is monotonically increasing with $l$. Furthermore, $h(N+k-1)=h(N+k)=N+k$ holds. Therefore, \eqref{deka2} holds for $l\in\{1,2,..,N+k\}$.
Hence, we obtain \eqref{alleq_sec}.

We can also derive \eqref{seq_sec} in the same way. We only have to substitute $\tilde{U}$ and $\calU^{\calM_{s}}_{\times}$ for $U$ and $V\in\calU^{\calM_{all}}_{\times}$.
\end{proofof}

The proof of Theorem \ref{ev-prob-tsec} is almost the same as that of Theorem \ref{ev-prob-t}:
\begin{proofof}{Theorem \ref{ev-prob-tsec}}
Since the support of $\rho\otimes\rho_B$ is included in the subspace $\calH^{\calM_s}:=\otimes_{m\in\calM_s}\calH^{(m)}$, the following relation holds for arbitrary $U\in\calU^{\calM_{all}}_{\times}$:
\begin{align}
x_{A'}(\rho,\rho_B,U)=x_{A'}(\rho,\rho_B,\tilde{U}),
\end{align}
where $\tilde{U}$ defined from $U$ by $\tilde{V}=\left(\oplus_{m\in\calM_s}V^{(m)}\right)\oplus\left(\oplus_{n\not\in\calM_s}I^{(m)}\right)$.
Therefore, we only have to show 
\begin{align}
\mathrm{Prob}_{\tilde{U}\sim H^{\calM_{s}}_{\times}}\left[|x^{(2)}_{A'}(\rho,\rho_B,\tilde{U})-\overline{x^{(2)}_{A'}(\rho,\rho_B,\tilde{U})}|_{H^{\calM_s}_{\times}}|>t\right]\le 2\exp\left(-\frac{(_{N+k}C_s-2)t^2}{48l^4}\right).\label{ev-prob''}
\end{align} 
Note that $\min_{m\in\calM_s}\mathrm{dim}\calH^{(m)}=_{N+k}C_s$.
Therefore, due to Theorem \ref{Eli}, to show \eqref{ev-prob''}, it is sufficient to show that $x^{(2)}_{A'}(\rho,\rho_B,\tilde{U})$ is $2l^2$-Lipchitz.

To show that $x^{(2)}_{A'}(\rho,\rho_B,\tilde{U})$ is $2l^2$-Lipchitz, let us take two unitary operations $\hat{U}\in\calU^{\calM_s}_{\times}$ and $\hat{V}\in\calU^{\calM_s}_{\times}$.
We evaluate $|x_{A'}(\rho,\rho_B,\hat{U})-x_{A'}(\rho,\rho_B,\hat{V})|$ as follows:
\begin{align}
|x^{(2)}_{A'}(\rho,\rho_B,\hat{U})-x_{A'}(\rho,\rho_B,\hat{V})|
&=|\Tr[X^2_{A'}(\hat{U}(\rho\otimes\rho_B)\hat{U}^{\dagger}-\hat{V}(\rho\otimes\rho_B)\hat{V}^{\dagger})]|\nonumber\\
&\le\calD^2_{X_{A'}}\|\hat{U}(\rho\otimes\rho_B)\hat{U}^{\dagger}-\hat{V}(\rho\otimes\rho_B)\hat{V}^{\dagger}\|_{1}\nonumber\\
&\le l^2\|\hat{U}(\rho\otimes\rho_B)\hat{U}^{\dagger}-\hat{V}(\rho\otimes\rho_B)\hat{V}^{\dagger}\|_{1}.
\end{align}
Therefore, in order to show that $x^{(2)}_{A'}(\rho,\rho_B,\tilde{U})$ is $2l^2$-Lipchitz, we only have to show 
\begin{align}
\|\hat{U}(\rho\otimes\rho_B)\hat{U}^{\dagger}-\hat{V}(\rho\otimes\rho_B)\hat{V}^{\dagger}\|_{1}\le2D(\hat{U},\hat{V}).\label{legofinal2}
\end{align}
And we have already obtained \eqref{legofinal2} as \eqref{legofinal}.
Therefore, $x^{(2)}_{A'}(\rho,\rho_B,\tilde{U})$ is $2l^2$-Lipchitz, and we have obtained \eqref{ev-prob}.

We can show \eqref{ev-prob2B} by substituting $N+k-l$ for $l$ in the above argument.
\end{proofof}

\subsection{Other applications to Hayden-Preskill model with symmetry}
Other than \eqref{foggyineq}, there are several applications to Hayden-Preskill model.
For example, we can use \eqref{SIQ1} for non-maximally entangled states for the initial states $AR_A$ and $BR_B$.  Noting $\Delta_{+}\le(k+l)/2$, we obtain the following bound
\begin{align}
  \frac{1-\epsilon}{1+\epsilon}\times\frac{M \left(1-l/(N+k)\right)}{2 ( 
\sqrt{\calF}+2(k+l) )}
  \le \delta 
    .\label{foggyineq2pre}
\end{align}
To illustrate the meaning of this inequality, we consider the case of $M\propto k$.
Then, we obtain the lower bound \eqref{foggyineq2pre}:
\begin{align}
{\rm const.}\times\frac{1-l/(k+N)}{1+(2l+\sqrt{\calF})/(2k)}\le\delta.\label{foggyineq2}
\end{align}
Note that $\calF=4V_{\rho_{B}}(X_{B})$ where $\rho_{B}:=\Tr_{R_B}[\rho_{BR_B}]$. This inequality shows that when the fluctuation of the conserved quantity of the initial state of the black hole $B$ is not so large, in order to make $\delta$ small, we have to collect information from the Hawking radiation so that $l\gg k$ or $l\approx k+N$.
In other words, whenever the fluctuation of the conserved quantity of the black hole is small, then in order to recover the quantum data thrown into the black hole with good accuracy, we have to wait until the black hole is evaporated enough.
Note also that if $\sqrt{\calF}$ is small, the bound in \eqref{foggyineq2} does not become trivial even if $N$ is much larger than $k$.

\section{Lower bound of recovery error in the information recovery without using $R_B$}\label{S-woRB}
\begin{figure}[tb]
		\centering
		\includegraphics[width=.45\textwidth]{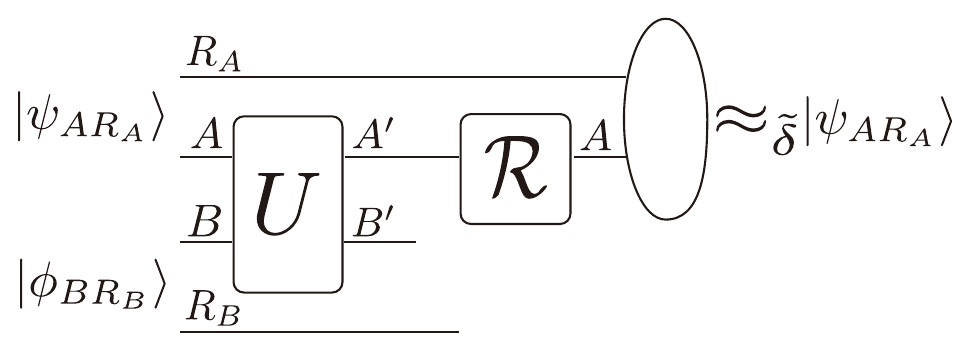}
		\caption{Schematic diagram of the information recovery without using $R_B$.}
		\label{setupSR2-S}
\end{figure}
The relations \eqref{SIQ1} and \eqref{SIQ2} in the main text describe the limitation of information recovery when one uses the quantum information of $R_B$. We can also discuss the case without using the information of $R_B$. The recovery operation $\calR$ in this case maps
the state on the system $A'$ to $A$\g{, as seen in} the schematic in Fig.~\ref{setupSR2-S}.  
We then define the recovery error as
\begin{align}
  \tilde{\delta}
  :=\min_{\mbox{$\calR$} \atop {A'\to A }} \!\!
  D_F\!\left( \rho_{AR_A},\mathrm{id}_{R_A}\otimes\calR\circ\calE(\rho_{AR_A})] \right) \, .\label{dwithoutRB}
\end{align}
Since $\tilde{\delta}\ge\delta$, \g{we} can substitute $\tilde{\delta}$ for $\delta$ in \eqref{SIQ1} and \eqref{SIQ2} to get a limitation of recovery in the present setup. Moreover, 
we can derive a tighter relation than this simple substitution as
\begin{align}
\frac{\SA}{2(\sqrt{\calF_B}+4\Delta_{+})}\le\tilde{\delta} \, , \label{SIQ1'-S}
\end{align}
where $\calF_{B}:=\calF_{\rho_B}(X_B)$. Note that $\calF_{B}\le\calF$ holds in general. The inequality (\ref{SIQ1'-S}) is the third main relation on the information recovery.

\begin{proofof}{\eqref{SIQ1'-S}}
We firstly take a quantum system $\tilde{B}$ whose dimension is the same as $B$, and a purification $\ket{\phi_{B\tilde{B}}}$ of $\rho_B:=\Tr_{R_B}[\phi_{B\tilde{B}}]$.
From $\ket{\phi_{B\tilde{B}}}$ and $U$, we define a set $\tilde{\calI}:=(\ket{\phi_{B\tilde{B}}}\otimes\ket{\eta_{R_B}},U\otimes 1_{\tilde{B}})$.
We take the Schmidt decomposition of $\ket{\phi_{B\tilde{B}}}$ as
\begin{align}
\ket{\phi_{B\tilde{B}}}=\sum_{l}\sqrt{r_l}\ket{l_{B}}\ket{l_{\tilde{B}}},
\end{align}
and define $X_{\tilde{B}}$ on $\tilde{B}$ as 
\begin{align}
X_{\tilde{B}}:=\sum_{ll'}\frac{2\sqrt{r_{l}r_{l'}}}{r_{l}+r_{l'}}\bra{l_{B}}X_{B}\ket{l'_{B}}\ket{l'_{\tilde{B}}}\bra{l_{\tilde{B}}}.
\end{align}
Then, due to \eqref{F=4V2} and \eqref{bestXR}, 
\begin{align}
\calF_{\rho_B}(X_B)=4V_{\ket{\phi_{B\tilde{B}}}}(X_{B}+X_{\tilde{B}}).\label{T-M2}
\end{align}

Note that $\tilde{\calI}$ is a Steinspring representation of $\calE$ and that $U\otimes 1_{\tilde{B}}(X_A+X_B+X_{\tilde{B}})(U\otimes 1_{\tilde{B}})^\dagger=X_{A'}+X_{B'}+X_{\tilde{B}}$.
Therefore, we  obtain the following inequality from \eqref{SIQ1}:
\begin{align}
\frac{\SA(\psi_{AR_A},\calE)}{2(\sqrt{\calF_{\ket{\phi_{B\tilde{B}}}\otimes\ket{\eta_{R_B}}}((X_{B}+X_{\tilde{B}})\otimes1_{R_B})}+4\Delta_{+})}\le
\delta(\psi_{AR_A},\tilde{\calI})
\end{align}
Since $\ket{\phi_{B\tilde{B}}}\otimes\ket{\eta_{R_B}}$ is a tensor product between $B\tilde{B}$ and $R_B$, the state of $B\tilde{B}R_B$ after $U$ is also another tensor product state between $B\tilde{B}$ and $R_B$.
Therefore, we obtain 
\begin{align}
\delta(\psi_{AR_A},\tilde{\calI})=\tilde{\delta}
\end{align}
Finally, from \eqref{T-M2}, we obtain 
\begin{align}
\calF_{\rho_B}(X_B)=4V_{\ket{\phi_{B\tilde{B}}}}(X_{B}+X_{\tilde{B}})
=4V_{\ket{\phi_{B\tilde{B}}}\otimes\ket{\eta_{R_B}}}((X_{B}+X_{\tilde{B}})\otimes1_{R_B})=\calF_{\ket{\phi_{B\tilde{B}}}\otimes\ket{\eta_{R_B}}}((X_{B}+X_{\tilde{B}})\otimes1_{R_B}).
\end{align}
Therefore, we obtain \eqref{SIQ1'-S}.
\end{proofof}

\section{Rederivation of approximated Eastin-Knill theorem as a corollary of \eqref{SIQ2}}\label{covariant-codes}
In this subsection, we rederive the approximate Eastin-Knill theorem from our trade-off relation \eqref{SIQ2} and/or \eqref{R-SIQ2}.
Following the setup for Theorem 1 in Ref. \cite{e-EKFaist-S}, we assume the following three:
\begin{itemize}
\item We assume that the code $\calC$ is covariant with respect to $\{U^{L}_{\theta}\}_{\theta\in\mathbb{R}}$ and $\{U^{P}_{\theta}\}_{\theta\in\mathbb{R}}$, where $U^{L}_{\theta}:=e^{i\theta X_L}$ and  $U^{P}_{\theta}:=e^{i\theta X_P}$. 
We also assume that the code $\calC$ is an isometry.
\item We assume that the physical system $P$ is a composite system of subsystems $\{P_{i}\}^{N}_{i=1}$, and that $X_{P}$ is written as $X_{P}=\sum_iX_{P_i}$. We also assume that the lowest eigenvalue of each $X_{P_i}$ is 0. (We can omit the latter assumption. See the section \ref{v-symmetry})
\item We assume that the noise $\calN$ is the erasure noise in which the location of the noise is known. To be concrete, the noise $\calN$ is a CPTP-map from $P$ to $P':=PC$ written as follows:
\begin{align}
\calN(...):=\sum_{i}\frac{1}{N}\ket{i_C}\bra{i_C}\otimes\ket{\tau_i}\bra{\tau_i}_{P_i}\otimes\Tr_{P_i}[...],\label{noisedef}
\end{align}
where the subsystem $C$ is the register remembering the location of error, and $\{\ket{i_C}\}$ is an orthonormal basis of $C$. The state $\ket{\tau_i}_{P_i}$ is a fixed state in $P_i$. 
\end{itemize}
In general, $\calN$ is not a covariant operation.
However, we can substitute the following covariant operation $\tilde{\calN}$ for $\calN$ without changing $\delta_C$:
\begin{align}
\tilde{\calN}(...):=\sum_{i}\frac{1}{N}\ket{i_C}\bra{i_C}\otimes\ket{0_i}\bra{0_i}_{P_i}\otimes\Tr_{P_i}[...]
\end{align}
where $\ket{0_i}$ is the eigenvector of $X_{P_i}$ whose eigenvalue is 0.
We can easily see that $\tilde{\calN}\circ\calC$ and $\calN\circ\calC$ are the same in the sense of $\delta_C$ by noting that we can convert the final state of $\tilde{\calN}\circ\calC$ to the final state of $\calN\circ\calC$ by the following unitary operation:
\begin{align}
W:=\sum_{i}\ket{i_C}\bra{i_C}\otimes U_{P_i}\otimes_{j:j\ne i} I_{P_j},
\end{align}
where $U_{P_i}$ is a unitary on $P_i$ satisfying $\ket{\tau_i}=U_{P_i}\ket{0_i}$.

Under the above setup, $\tilde{\calN}\circ\calC$ is covariant with respect to $\{U^{L}_{\theta}\}$ and $\{I_C\otimes U^{P}_{\theta}\}$.
Therefore, we can apply \eqref{SIQ1}, \eqref{SIQ2}, \eqref{R-SIQ1} and \eqref{R-SIQ2} to this situation.
Below, we derive the following approximated Eastin-Knill theorem from  \eqref{R-SIQ2}.
\begin{align}
\frac{\calD_{X_L}}{2\delta_C\calD_{\max}}\le N+\frac{\calD_{X_L}}{2\calD_{\max}}.\label{almostFaist}
\end{align}
Here $\calD_{\max}:=\max_{i}\calD_{P_i}$.
This inequality is the same as the inequality in Theorem 1 of \cite{e-EKFaist-S}, apart from the irrelevant additional term $\calD_{X_L}/2\calD_{\max}$. (In Theorem 1 of \cite{e-EKFaist-S}, $\frac{\calD_{X_L}}{2\delta_C\calD_{\max}}\le N$ is given.)
We can also derive a very similar inequality from \eqref{SIQ2}.
When we use \eqref{SIQ2} instead of \eqref{R-SIQ2}, the coefficient $1/2$ in the lefthand side of \eqref{almostFaist} becomes $1/4$.
We remark that although the bound \eqref{almostFaist} is little bit weaker than the bound in Theorem 1 of Ref.\cite{e-EKFaist-S}, it is still remarkable, because \eqref{almostFaist} is given as a corollary of more general inequality \eqref{R-SIQ2}.

\begin{proofof}{\eqref{almostFaist}}
We construct an implementation of $\tilde{\calN}\circ\calC$ by combining the following implementations of $\calC$ and $\tilde{\calN}$.
As the implementation of $\calC$, we take a system $B$ satisfying $LB=P$, a Hermitian operator $X_B$, a symmetric state $\rho_B$ on $B$ with respect to $X_B$, and a unitary $U$ on $LB$ satisfying 
\begin{align}
U(X_L+X_B)U^\dagger&=X_P,\label{c-law-P1}\\
[\rho_B,X_B]&=0\label{Bsym}.\\
\calC(...)&=U(...\otimes\rho_B) U^{\dagger}
\end{align}
The existence of such $B$ $X_B$, $U$, and $\rho_B$ is guaranteed since $\calC$ is an isometry and any covariant operation is realized by an invariant unitary and a symmetric state (see Appendix \ref{AppA} in the main text).

As an implementation of $\tilde{\calN}$, we take a composite system $B_1:=C\tilde{P_1}...\tilde{P_N}$ where each $\tilde{P_i}$ is a copy system of $P_i$ which has $\tilde{X}_{P_i}$ that is equal to $X_{P_i}$.
We also define a state $\rho_{B_1}$ on $B_1$ and a unitary $V$ on $PB_1$ as follows
\begin{align}
\rho_{B_1}&:=\frac{1}{N}\sum^{N}_{j=1}\ket{j}\bra{j}_{C}\otimes(\otimes^{N}_{i=1}\ket{0_i}\bra{0_i}_{\tilde{P}_i})\\
V&:=\sum_{k}\ket{k}\bra{k}_{C}\otimes S_{\tilde{P_k}P_k}\otimes(\otimes_{j:j\ne k}I_{\tilde{P_j}P_j}),
\end{align} 
where $S_{\tilde{P_k}P_k}$ is the swap unitary between $\tilde{P_k}$ and $P_k$ and $I_{\tilde{P_j}P_j}$ is the identity operator on $\tilde{P_j}P_j$.
Then, $\rho_{B_1}$ and $V$ satisfies
\begin{align}
V(X_{P}\otimes I_{\tilde{P}}\otimes I_C+I_{P}\otimes X_{\tilde{P}}\otimes I_C)V^\dagger&=X_{P}\otimes I_{\tilde{P}}\otimes I_C+I_{P}\otimes X_{\tilde{P}}\otimes I_C,\label{c-law-P2}\\
[\rho_{B_1},X_{\tilde{P}}\otimes I_C]&=0,\\
\tilde{\calN}(...)&=\Tr_{\tilde{P}}[V(...\otimes\rho_{B_1})V^\dagger]
\end{align}
where $\tilde{P}=\tilde{P_1}...\tilde{P_N}$ and $X_{\tilde{P}}=\sum^{N}_{j=1}X_{\tilde{P}_j}$.

For the above implementation, from \eqref{R-SIQ2} and $\delta_C\ge\max_{\ket{\psi_{LR_L}}}\delta$, we obtain the following relation for an arbitrary $\ket{\psi_{LR_L}}$:
\begin{align}
\frac{\SA_2}{\delta_C}\le2\sqrt{V_{\rho^{f}_{\tilde{P}}}(X_{\tilde{P}})}+\Delta_{\max},\label{TS-appEN-Spre}
\end{align}
where $\rho^f_{\tilde{P}}$ is the final state of $\tilde{P}$.

To derive \eqref{almostFaist} from \eqref{R-SIQ2}, let us evaluate $\SA_2$, $\Delta_{\max}$ and $V_{\rho^{f}_{\tilde{P}}}(X_{\tilde{P}})$ for the following $\ket{\psi_{LR_L}}$:
\begin{align}
\ket{\psi_{LR_L}}:=\frac{\ket{0_L}\ket{0_{R_L}}+\ket{1_L}\ket{1_{R_L}}}{\sqrt{2}},
\end{align}
where $\ket{0_L}$ and $\ket{1_L}$ are the maximum and minimum eigenvectors of $X_L$.
Due to the definition of $\SA_2$, we obtain
\begin{align}
\SA_2\ge\frac{1}{2}\sum^{1}_{i=0}\left|\left(\ex{X_L}_{\ket{j_L}\bra{j_L}}-\ex{X_{P}\otimes I_C}_{\calE(\ket{j_L}\bra{j_L})}\right)-\left(\ex{X_L}_{(\ket{0_L}\bra{0_L}+\ket{1_L}\bra{1_L})/2}-\ex{X_{P}\otimes I_C}_{\calE((\ket{0_L}\bra{0_L}+\ket{1_L}\bra{1_L})/2)}\right)\right|
\end{align}
Due to \eqref{noisedef} and \eqref{c-law-P1}, for any $\rho_L$ on $L$,
\begin{align}
\ex{X_{P}\otimes I_C}_{\calE(\rho_L)}
&=\left(1-\frac{1}{N}\right)\left((\ex{X_L}_{\rho_L}+\ex{X_B}_{\rho_B}\right)+\frac{1}{N}\sum^{N}_{i=1}\ex{X_{P_i}}_{\ket{0_i}\bra{0_i}}\nonumber\\
&=\left(1-\frac{1}{N}\right)\left((\ex{X_L}_{\rho_L}+\ex{X_B}_{\rho_B}\right).
\end{align}
Therefore, we obtain
\begin{align}
\SA_2&\ge\frac{1}{2N}\sum^{1}_{j=0}|\ex{X_L}_{\ket{j_L}\bra{j_L}}-\ex{X_L}_{(\ket{0_L}\bra{0_L}+\ket{1_L}\bra{1_L})/2}|\nonumber\\
&=\frac{\calD_{X_L}}{2N}.
\end{align}
By definition of $\Delta_{\max}$, we obtain
\begin{align}
\Delta_{\max}&=\max_{\rho\mbox{ on the support of $(\ket{0_L}\bra{0_L}+\ket{1_L}\bra{1_L})/2$}}\frac{1}{N}\left|\ex{X_L}_{\rho}-\ex{X_L}_{(\ket{0_L}\bra{0_L}+\ket{1_L}\bra{1_L})/2}\right|\nonumber\\
&\le\frac{\calD_{X_L}}{2N}.
\end{align}
To evaluate $V_{\rho^{f}_{\tilde{P}}}(X_{\tilde{P}})$, we note that 
\begin{align}
\rho^{f}_{\tilde{P}}=\frac{1}{N}\sum^{N}_{h=1}\rho^{f}_{h}\otimes(\otimes_{i:i\ne h}\ket{0_{i}}\bra{0_{i}})
\end{align}
where $\rho^{f}_{h}:=\Tr_{\lnot P_h}[\calC((\ket{0_L}\bra{0_L}+\ket{1_L}\bra{1_L})/2)]$.
Therefore, 
\begin{align}
\ex{X^2_{\tilde{P}}}_{\rho^{f}_{\tilde{P}}}&=\frac{\sum_{h}\ex{X^2_{P_h}}_{\rho^{f}_{h}}}{N}\\
\ex{X_{\tilde{P}}}_{\rho^{f}_{\tilde{P}}}&=\frac{\sum_{h}\ex{X_{P_h}}_{\rho^{f}_{h}}}{N}.
\end{align}
With using the above, we evaluate $V_{\rho^{f}_{\tilde{P}}}(X_{\tilde{P}})$ as follows:
\begin{align}
V_{\rho^{f}_{\tilde{P}}}(X_{\tilde{P}})&=\ex{X^2_{\tilde{P}}}_{\rho^{f}_{\tilde{P}}}-\ex{X_{\tilde{P}}}_{\rho^{f}_{\tilde{P}}}^2\nonumber\\
&=\frac{\sum_{h}\ex{X^2_{P_h}}_{\rho^{f}_{h}}}{N}-\left(\frac{\sum_{h}\ex{X_{P_h}}_{\rho^{f}_{h}}}{N}\right)^2\nonumber\\
&=V^{c}_{Q}(x)\nonumber\\
&\le\frac{\calD^2_{\max}}{4}
\end{align}
where $V^{c}_{Q}(x)$ is the variance of a classical distribution of $Q$ on  a set of real numbers $\calX$ defined as follows:
\begin{align}
Q(x)&:=\sum^{N}_{h=1}\frac{P_h(x)}{N}\\
P_h(x)&:=\begin{cases}
    \bra{x_h}\rho^{f}_{h}\ket{x_h} & (x\in \calX_h) \\
    0 & (otherwise)
  \end{cases}\\
\calX_h&:=\{x|\mbox{eigenvalues of $X_{P_h}$}\}\\
\calX&:=\bigcup^{N}_{h=1}\calX_h
\end{align}
where $\ket{x_h}$ is an eigenvector of $X_{P_h}$ whose eigenvalue is $x$.

Combining the above, we obtain \eqref{almostFaist}.
\end{proofof}

\section{Generalization of main results to the case of general Lie group symmetry}\label{g-symmetry}

In this section, we generalize the results in the main text to the case of general Lie group symmetries.
In the first subsection, we derive a variation of the main results (\eqref{SIQ1} and \eqref{SIQ2} in the main text) for the case of the conservation law of $X$, as preliminary. 
In the variation, we use $\SA_V$ which represents the variance of the change of local conserved quantity $X$ instead of $\SA$.
In the second subsection, we extend the variation to the case of general symmetries.

\subsection{Variance-type lower bound of recovery error for the cases of  $U(1)$ and $\mathbb{R}$}

In this subsection, we derive a variation of the main results for the case of the conservation law of $X$. 
We consider Setup 1 with the conservation law of $X$: $X_A+X_B=U^\dagger(X_{A'}+X_{B'}U)$.
For an arbitrary decomposition of $\rho_A:=\sum_jp_j\rho_{j,A}$, we define the following quantity:
\begin{align}
\SA_V(\{p_j,\rho_{j,A}\},\calE):=\sum_{j}p_j\Delta^2_j.
\end{align} 
Hereafter, we abbreviate $\SA_V(\{p_j,\rho_{j,A}\},\calE)$ as $\SA_V$.
We remark that the quantity $\SA_V$ depends on the decomposition of $\rho_A$, unlike $\SA$. 

For $\SA_V$, the following trade-off relation holds:
\begin{align}
\frac{\SA_V}{8\delta^2}&\le\calF+\calB,\label{VSIQ1}\\
\frac{\SA_V}{8\delta^2}&\le\calF_f+\calB,\label{VSIQ2}
\end{align}
where $\delta$, $\calF$ and $\calF_f$ are the same as in \eqref{SIQ1} and \eqref{SIQ2}, and $\calB$ is defined as follows:
\begin{align}
\calB&:=\frac{\sum_j\Delta^2_j}{2}+8(V_{\rho_A}(X_A)+V_{\calE(\rho_A)}(X_{A'})).
\end{align}

\begin{proofof}{\eqref{VSIQ1} and \eqref{VSIQ2}}
To derive \eqref{VSIQ1} and \eqref{VSIQ2}, we use the following mean-variance-distance trade-off relation which holds for arbitrary states $\rho$ and $\sigma$ and an arbitrary Hermitian operator $X$ \cite{Nishiyama}:
\begin{align}
\Tr[(\rho-\sigma)X]^2\le D_F(\rho,\sigma)^2((\sqrt{V_{\rho}(X)}+\sqrt{V_{\sigma}(X)})^2+\Tr[(\rho-\sigma)X]^2).\label{N-R-C}
\end{align}
With using \eqref{N-R-C}, Lemma \ref{SCL} and \eqref{Eeq}--\eqref{Vineq2-S}, we derive \eqref{VSIQ1} and \eqref{VSIQ2}, in the very similar way to \eqref{SIQ1} and \eqref{SIQ2}.

Let us take an arbitrary decomposition of $\rho_A$ as $\rho_A=\sum_jp_j\rho_{j,A}$.
Then, the following equation follows from \eqref{Eeq}:
\begin{align}
|\Delta_j|=|\ex{X_{B'}}_{\rho^{f}_{j,B'}}-\ex{X_{B'}}_{\rho^{f}_{B'}}|.\label{Eeq2-S}
\end{align}
We firstly evaluate $\SA_V$ as follows:
\begin{align}
\SA_V&\stackrel{(a)}{=}\sum_{j}p_j(\ex{X_{B'}}_{\rho^{f}_{j,B'}}-\ex{X_{B'}}_{\rho^{f}_{B'}})^2\nonumber\\
&\stackrel{(b)}{\le}\sum_{j}q_jD_{F}(\rho^{f}_{j,B'},\rho^{f}_{B'})^2\left((\sqrt{V_{\rho^{f}_{j,B'}}(X_{B'})}+\sqrt{V_{\rho^{f}_{B'}}(X_{B'})})^2+\Delta^2_{j}\right)\label{totyu}
\end{align}
Here we use \eqref{Eeq2-S} in (a), \eqref{RCR2} in (b).

Second, we evaluate $(\sqrt{V_{\rho^{f}_{j,B'}}(X_{B'})}+\sqrt{V_{\rho^{f}_{B'}}(X_{B'})})^2$ in \eqref{totyu} as follows:
\begin{align}
\left(\sqrt{V_{\rho^{f}_{j,B'}}(X_{B'})}+\sqrt{V_{\rho^{f}_{B'}}(X_{B'})}\right)^2
&\le4\left(\sqrt{V_{\rho_B}(X_{B})}+(\sqrt{V_{\rho_A}(X_A)}+\sqrt{V_{\rho_{A'}}(X_{A'})})\right)^2
\nonumber\\
&\le4\left(2V_{\rho_B}(X_{B})+4(V_{\rho_A}(X_A)+V_{\rho_{A'}}(X_{A'}))\right)\nonumber\\
&=2(\calF+8(V_{\rho_A}(X_A)+V_{\rho_{A'}}(X_{A'})))\label{totyu2-A}
\end{align}
Here we use \eqref{Vineq} and $(x+y)^2\le 2(x^2+y^2)$.
By combining \eqref{totyu}, \eqref{totyu2-A}, Lemma \ref{SCL} and $\Delta^2_j\le\sum_j\Delta^2_j$, we obtain \eqref{VSIQ1}:
\begin{align}
\SA_V\le8\delta^2\left(\calF+8(V_{\rho_A}(X_A)+V_{\rho_{A'}}(X_{A'}))+\frac{\sum_j\Delta^2_j}{2}\right).
\end{align}

To derive \eqref{VSIQ2}, we evaluate $(\sqrt{V_{\rho^{f}_{j,B'}}(X_{B'})}+\sqrt{V_{\rho^{f}_{B'}}(X_{B'})})^2$ in \eqref{totyu} in a different way:
\begin{align}
\left(\sqrt{V_{\rho^{f}_{j,B'}}(X_{B'})}+\sqrt{V_{\rho^{f}_{B'}}(X_{B'})}\right)^2
&\le\left(\sqrt{V_{\rho_{B}}(X_{B})}+\sqrt{V_{\rho^{f}_{B'}}(X_{B'})}+\sqrt{V_{\rho_A}(X_A)}+\sqrt{V_{\rho_{A'}}(X_{A'})}\right)^2
\nonumber\\
&\le4\left(\sqrt{V_{\rho^{f}_{B'}}(X_{B'})}+(\sqrt{V_{\rho_A}(X_A)}+\sqrt{V_{\rho_{A'}}(X_{A'})})\right)^2
\nonumber\\
&\le4\left(2V_{\rho^f_{B'}}(X_{B'})+4(V_{\rho_A}(X_A)+V_{\rho_{A'}}(X_{A'}))\right)\nonumber\\
&=2(\calF+8(V_{\rho_A}(X_A)+V_{\rho_{A'}}(X_{A'})))\label{totyu2-B}
\end{align}
Here we use \eqref{Vineq}, \eqref{Vineq2-S} and $(x+y)^2\le 2(x^2+y^2)$.

By combining \eqref{totyu}, \eqref{totyu2-B}, Lemma \ref{SCL} and $\Delta^2_j\le\sum_j\Delta^2_j$, we obtain \eqref{VSIQ2}:
\begin{align}
\SA_V\le8\delta^2\left(\calF_f+8(V_{\rho_A}(X_A)+V_{\rho_{A'}}(X_{A'}))+\frac{\sum_j\Delta^2_j}{2}\right).
\end{align}
\end{proofof}

\subsection{Main results for general symmetry: Limitations of recovery error for general Lie groups }\label{MSIQwithRB}

Now, we introduce the generalized version of the main results.
We consider Setup 1, and assume that $U$ is restricted by some Lie group symmetry.
To be more concrete, we take an arbitrary Lie group $G$ and its unitary representations $\{V_{g,\alpha}\}_{g\in G}$ ($\alpha=A,B,A',B'$).
We assume that $U$ satisfies the following relation:
\begin{align}
U(V_{g,A}\otimes V_{g,B})=(V_{g,A'}\otimes V_{g,B'})U,\enskip g\in G.\label{g-law}
\end{align}
Let $\{X^{(a)}_{\alpha}\}$ $(\alpha=A,B,A',B')$ be an arbitrary basis of Lie algebra corresponding to $\{V_{g,\alpha}\}_{g\in G}$.
Then, for an arbitrary decomposition $\rho_A=\sum_jp_j\rho_{j,A}$, the following matrix inequalities hold:
\begin{align}
\frac{\wh{\SA_V}}{8\delta^2}&\preceq \wh{\calF}+\wh{\calB},\label{MSIQ1}\\
\frac{\wh{\SA_V}}{8\delta^2}&\preceq \wh{\calF_f}+\wh{\calB},\label{MSIQ2}
\end{align}
where $\preceq$ is the inequality for matrices, and $\wh{\SA_V}$ and $\wh{\calB}$ are matrices whose components are defined as follows:
\begin{align}
\wh{\SA_V}_{ab}&:=\sum_jp_j\Delta^{(a)}_{j}\Delta^{(b)}_{j}\\
\Delta^{(a)}_{j}&:=\left(\ex{X{(a)}_A}_{\rho_j}-\ex{X^{(a)}_{A'}}_{\calE(\rho_j)}\right)-\left(\ex{X^{(a)}_A}_{\rho_A}-\ex{X^{(a)}_{A'}}_{\calE(\rho_A)}\right)\\
\wh{\calB}_{ab}&:=8(Cov_{\rho_A}(X^{(a)}_{A}:X^{(b)}_{A})+Cov_{\calE(\rho_{A})}(X^{(a)}_{A'}:X^{(b)}_{A'}))+\frac{\sum_j\Delta^{(a)}_{j}\Delta^{(b)}_{j}}{2}.
\end{align}
and $\wh{\calF}$ and $\wh{\calF_f}$ are the Fisher information matrices
\begin{align}
\wh{\calF}&:=\wh{\calF}_{\phi_{BR_B}}(\{X^{(a)}_{B}\otimes1_{R_B}\})\\
\wh{\calF_f}&:=\wh{\calF}_{\phi^{f}_{B'R_{B'}}}(\{X^{(a)}_{B}\otimes1_{R_B}\}),
\end{align}
where the Fisher information matrix $\wh{\calF}_{\xi}(\{X^{(a)}\})$ is defined as follows for a given state $\xi$ and given Hermite operators $\{X^{(a)}\}$:
\begin{align}
\wh{\calF}_{\xi}(\{X^{(a)}\})_{ab}=2\sum_{i,i'}\frac{(r_i-r_{i'})^2}{r_i+r_{i'}}X^{(a)}_{ii'}X^{(b)}_{i'i} \,  \label{afi}
\end{align}
Here, $r_i$ is the $i$-th eigenvalue of the density matrix $\xi$ with the eigenvector $\psi_i$, and $X^{(a)}_{ii'}:=\bra{\psi_i}X^{(a)}\ket{\psi_{i'}}$. 

\begin{proofof}{\eqref{MSIQ1} and \eqref{MSIQ2}}
We first show \eqref{MSIQ1}.
Since $\wh{\SA_V}$, $\wh{\calF}$ and $\wh{\calB}$ are real symmetric matrices, we only have to show the following relation holds for arbitrary real vector $\vec{\lambda}$:
\begin{align}
\vec{\lambda}^{T}\frac{\wh{\SA_V}}{8\delta^2}\vec{\lambda}\le\vec{\lambda}^{T}(\wh{\calF}+\wh{\calB})\vec{\lambda}.\label{MSIQ1pre}
\end{align}
By definition of  $\wh{\SA_V}$, $\wh{\calF}$ and $\wh{\calB}$, the inequality \eqref{MSIQ1pre} is equivalent to \eqref{VSIQ1} whose $X_A$, $X_{A'}$ and $X_{B}$ are substituted by $X_{\alpha,\vec{\lambda}}=\sum_{a}\lambda_aX^{(a)}_{\alpha}$  ($\alpha=A,A',B$ and $\{\lambda_{a}\}$ are the components of $\vec{\lambda}$).
Therefore, we only have to show that the following equality holds for arbitrary $\vec{\lambda}$:
\begin{align}
U(X_{A,\vec{\lambda}}+X_{B,\vec{\lambda}})U^{\dagger}=X_{A',\vec{\lambda}}+X_{B',\vec{\lambda}}.\label{MSIQ1pre2}
\end{align}
Due to \eqref{g-law}, for any $a$, the following relation holds:
\begin{align}
U(X^{(a)}_{A}+X^{(a)}_{B})=(X^{(a)}_{A'}+X^{(a)}_{B'})U.
\end{align}
Therefore, \eqref{MSIQ1pre2} holds, and thus we obtain \eqref{MSIQ1}.
We can obtain \eqref{MSIQ2} in the same way.

\end{proofof}

\subsection{Limitations of recovery error for general symmetry in information recovery without using $R_B$}

In this subsection, we extend \eqref{MSIQ1} and \eqref{MSIQ2} to the case of information recoveries without using $R_B$.
Let us consider the almost same setup as in the subsection \ref{MSIQwithRB}.
The difference between the present setup and the setup in the subsection \ref{MSIQwithRB} is that the recovery operation $\calR$ is a CPTP-map $A'$ to $A$.
Then, the recovery error is $\tilde{\delta}$ which is defined in \eqref{dwithoutRB}.

As is explained in the section \ref{S-woRB}, since the inequality $\tilde{\delta}\ge\delta$ holds in general, we can substitute $\tilde{\delta}$ for $\delta$ in \eqref{MSIQ1} and \eqref{MSIQ2}.
Moreover, we can derive the following more strong inequality from \eqref{VSIQ1}:
\begin{align}
\frac{\wh{\SA_V}}{8\tilde{\delta}^2}\preceq \wh{\calF_B}+\wh{\calB},\label{MSIQ1'}
\end{align}
where $\wh{\calF_B}:=\wh{\calF}_{\rho_B}(\{X^{(a)}_{B}\})$. 

The proof of \eqref{MSIQ1'} is very similar to the proof of \eqref{SIQ1'-S}:
\begin{proofof}{\eqref{MSIQ1'}}
As in the proof of \eqref{MSIQ1}, since $\wh{\SA_V}$, $\wh{\calF_B}$ and $\wh{\calB}$ are real symmetric matrices, we only have to show the following inequality for an arbitrary real vector $\vec{\lambda}$:
\begin{align}
\vec{\lambda}^{T}\frac{\wh{\SA_V}}{8\tilde{\delta}^2}\vec{\lambda}\le\vec{\lambda}^{T}(\wh{\calF_B}+\wh{\calB})\vec{\lambda}.
\end{align}

We take a quantum system $\tilde{B}$ whose dimension is the same as $B$, and a purification $\ket{\phi_{B\tilde{B}}}$ of $\rho_B:=\Tr_{R_B}[\phi_{BR_B}]$.
From $\ket{\phi_{B}\tilde{B}}$ and $U$, we define a set $\tilde{\calI}:=(\ket{\phi_{B\tilde{B}}}\otimes\ket{\eta_{R_B}},U\otimes 1_{\tilde{B}})$.
We take the Schmidt decomposition of $\ket{\phi_{B\tilde{B}}}$ as
\begin{align}
\ket{\phi_{B\tilde{B}}}=\sum_{l}\sqrt{r_l}\ket{l_{B}}\ket{l_{\tilde{B}}},
\end{align}
and define $\{X^{(a)}_{\tilde{B}}\}$ on $\tilde{B}$ corresponding to $\{X^{(a)}_{B}\}$ as 
\begin{align}
X^{(a)}_{\tilde{B}}:=\sum_{ll'}\frac{2\sqrt{r_{l}r_{l'}}}{r_{l}+r_{l'}}\bra{l_{B}}X^{(a)}_{B}\ket{l_{B}}\ket{l'_{\tilde{B}}}\bra{l_{\tilde{B}}}.
\end{align}

Note that $\tilde{\calI}$ is a Steinspring representation of $\calE$ and that $U\otimes 1_{\tilde{B}}(X^{(a)}_A+X^{(a)}_B+X^{(a)}_{\tilde{B}})(U\otimes 1_{\tilde{B}})^\dagger=X^{(a)}_{A'}+X^{(a)}_{B'}+X^{(a)}_{\tilde{B}}$ for any $a$.
Therefore, we obtain the following inequality from \eqref{VSIQ1} by substituting
 $X^{(\vec{\lambda})}_A:=\sum_a\lambda_aX^{(a)}_{A}$ for $X_A$, $X^{(\vec{\lambda})}_{B\tilde{B}}:=\sum_a\lambda_a(X^{(a)}_{B}+X^{(a)}_{\tilde{B}})$ for $X_B$, $X^{(\vec{\lambda})}_{A'}:=\sum_a\lambda_aX^{(a)}_{A'}$ for $X_{A'}$, and $X^{(\vec{\lambda})}_{B'\tilde{B}}:=\sum_a\lambda_a(X^{(a)}_{B'}+X^{(a)}_{\tilde{B}})$ for $X_{B'}$:
\begin{align}
\frac{\SA^{(\vec{\lambda})}_{V}(\psi_{AR_A},\tilde{\calI})}{8\delta(\psi_{AR_A},\tilde{\calI})}\le\calF^{(\vec{\lambda})}_{\ket{\phi_{B\tilde{B}}}\otimes\ket{\eta_{R_B}}}+\calB^{(\vec{\lambda})}.
\end{align}
Here $\SA^{(\vec{\lambda})}_{V}(\psi_{AR_A},\tilde{\calI})$, $\calF^{(\vec{\lambda})}_{\ket{\phi_{B\tilde{B}}}\otimes\ket{\eta_{R_B}}}$ and $\calB^{(\vec{\lambda})}$ are $\SA_V$, $\calF$ and $\calB$ for $(\ket{\phi_{B\tilde{B}}}\otimes\ket{\eta_{R_B}},U\otimes1_{R_B})$ and $X^{(\vec{\lambda})}_{\alpha}$ ($\alpha=A,B\tilde{B},A',B'\tilde{B}$).

Since both of $\calI$ and $\tilde{\calI}$ gives the same CPTP-map $\calE$, and due to the definitions of $\SA^{(\vec{\lambda})}_{V}(\psi_{AR_A},\tilde{\calI})$ and $\calB^{(\vec{\lambda})}$,
\begin{align}
\SA^{(\vec{\lambda})}_V(\psi_{AR_A},\tilde{\calI})&=\vec{\lambda}^{T}\SA_V\vec{\lambda},\\
\calB^{(\vec{\lambda})}&=\vec{\lambda}^{T}\wh{\calB}\vec{\lambda}.
\end{align}
Similarly due to \eqref{F=4V2}, 
\begin{align}
\vec{\lambda}^{T}\wh{\calF}_{\rho_B}(\{X^{(a)}_{B}\})\vec{\lambda}=\calF_{\rho_B}(\sum_{a}\lambda_{a}X^{(a)}_B)=4V_{\ket{\phi_{B\tilde{B}}}}(\sum_{a}\lambda_{a}(X^{(a)}_{B}+X^{(a)}_{\tilde{B}}))=\calF^{(\vec{\lambda})}_{\ket{\phi_{B\tilde{B}}}\otimes\ket{\eta_{R_B}}}.\label{T-M2-M}
\end{align}
Moreover, since $\ket{\phi_{B\tilde{B}}}\otimes\ket{\eta_{R_B}}$ is a tensor product between $B\tilde{B}$ and $R_B$, the state of $B\tilde{B}R_B$ after $U$ is also another tensor product state between $B\tilde{B}$ and $R_B$.
Therefore, we obtain 
\begin{align}
\delta(\psi_{AR_A},\tilde{\calI})=\tilde{\delta}
\end{align}
Combining the above, we obtain \eqref{MSIQ1'}.
\end{proofof}

\subsection{Applications of the limitations of recovery error for general symmetries}
As the cases of $U(1)$ and $\mathbb{R}$, we can use the inequalities \eqref{MSIQ1}, \eqref{MSIQ2} and \eqref{MSIQ1'} (and \eqref{MSIQ2} whose $\delta$ is substituted by $\tilde{\delta}$) to various phenomena.
\begin{itemize}
\item As \eqref{SIQ1} and \eqref{SIQ2}, we can apply \eqref{MSIQ1} and \eqref{MSIQ2} to information recovery from scrambling with general symmetry.
\item As \eqref{SIQ1'-S}, we can apply \eqref{MSIQ1} to implementation of general unitary dynamics and covariant error correcting codes with covariant errors. 
With using $\delta_U$ and $\delta_C$, we obtain
\begin{align}
\frac{\wh{\SA_V}}{8\delta^2_U}&\preceq \wh{\calF_B}+\wh{\calB}\\
\frac{\wh{\SA_V}}{8\delta^2_C}&\preceq \wh{\calB}
\end{align}
\end{itemize}

\section{limitations of recovery error for the case where the conservation law is weakly violated}\label{v-symmetry}
In this section, we consider the case where the conservation law of $X$ is violated.
We show that our results hold even in such cases.
We consider Setup 1 with the following violated global conservation law:
\begin{align}
X_{A}+X_{B}=U^{\dagger}(X_{A'}+X_{B'})U+Z.\label{vio-c-law}
\end{align}
Here $Z$ is some perturbation term which describes the strength of the violation of global conservation law.
Then, the following two relations hold:
\begin{align}
\frac{\SA-\SA_Z}{2(\sqrt{\calF}+2(\sqrt{V_{\rho_A}(X_A)}+\sqrt{V_{\rho^{f}_{A'}}(X_{A'})})+\SA^{(2)}+\SA^{(2)}_{Z}+2\sqrt{V_Z})}&\le\delta,\label{SIQ-v1}\\
\frac{\SA-\SA_Z}{2(\sqrt{\calF_f}+\SA^{(2)}+\calA^{(2)}_Z)}&\le\delta.\label{SIQ-v2}
\end{align}
Here $V_Z:=V_{\rho_A\otimes\rho_B}(Z)$ and 
\begin{align}
\SA_Z&:=\max_{\{p_j,\rho_{j,A}\}}\sum_jp_j|\ex{Z}_{\rho_{j,A}\otimes\rho_B}-\ex{Z}_{\rho_A\otimes\rho_B}|,\\
\SA^{(2)}_Z&:=\max_{\{p_j,\rho_{j,A}\}}\sqrt{\sum_jp_j|\ex{Z}_{\rho_{j,A}\otimes\rho_B}-\ex{Z}_{\rho_A\otimes\rho_B}|^2},\\
\SA^{(2)}&:=\max_{\{p_j,\rho_{j,A}\}}\sqrt{\sum_{j}p_j|\Delta_j|^2},
\end{align}
where $\{p_j,\rho_{j,A}\}$ runs $\rho_A=\sum_jp_j\rho_{j,A}$.

To simplify \eqref{SIQ-v1} and \eqref{SIQ-v2}, we can use the following relations (we prove them in the end of this section):
\begin{align}
\SA_Z&\le\SA^{(2)}_{Z}\le\sqrt{V_Z},\label{ur}\\
\SA^{(2)}&\le\Delta_{\max}\le2\Delta_{+},\label{ur2}\\
\sqrt{V_{\rho_A}(X_A)}+\sqrt{V_{\rho^{f}_{A'}}(X_{A'})}&\le\Delta_{+}\label{ur3}\\
\SA_Z&\le M_{\rho_A}(Z_A)\label{ur4}\\
\SA^{(2)}_Z&\le\sqrt{V_{\rho_A}(Z_A)}.\label{ur5}
\end{align}
where $Z_A:=\Tr_B[(1_A\otimes\rho_B)Z]$ and $M_{\rho_A}(Z_A):=\ex{|Z_A-\ex{Z_A}_{\rho_A}|}_{\rho_A}$.
For example, by using \eqref{ur}, \eqref{ur2} and \eqref{ur5}, we obtain the following inequalities from \eqref{SIQ-v1} and \eqref{SIQ-v2}:
\begin{align}
\frac{\SA-\sqrt{V_Z}}{2(\sqrt{\calF}+4\Delta_{+}+3\sqrt{V_Z})}&\le\delta,\label{SIQ-v1-sim}\\
\frac{\SA-\sqrt{V_Z}}{2(\sqrt{\calF_f}+\Delta_{\max}+\sqrt{V_Z})}&\le\delta.\label{SIQ-v2-sim}
\end{align}
We remark that we have introduced \eqref{SIQ-v1-sim} in the section \ref{mainA} of the main text.

Similarly, the following relations also hold:
\begin{align}
\frac{\SA_2-\SA_Z}{\sqrt{\calF}+2(\sqrt{V_{\rho_A}(X_A)}+\sqrt{V_{\rho^{f}_{A'}}(X_{A'})})+\SA^{(2)}+\SA^{(2)}_{Z}+2\sqrt{V_Z}}&\le\delta,\label{R-SIQ-v1}\\
\frac{\SA_2-\SA_Z}{\sqrt{\calF_f}+\SA^{(2)}+\calA^{(2)}_Z}&\le\delta.\label{R-SIQ-v2}
\end{align}

These inequalities have two important messages.
First, when $Z=\mu I$ where $\mu$ is an arbitrary real number, the inequalities \eqref{SIQ1} and \eqref{SIQ2} are valid, since in that case $\SA_Z=V_Z=V_{\rho_A}(Z_A)=0$ holds.
Therefore, we can omit the assumption that the lowest eigenvalue of $X_{P_i}$ is 0, which is used in the re-derivation of the approximate Eastin-Knill theorem in the section \ref{covariant-codes}.
Second, our trade-off relations become trivial only when $\SA\le \SA_Z$. 
As we show in the section 3 in the main text, the inequality $\SA\ge  M\gamma(1-\epsilon)$ holds in the Hayden-Preskill black hole model.
Therefore, when $M_Z$ is not so large, our message on black holes does not radically change. 
Even when the global conservation law is weakly violated, black holes are foggy mirrors.

\begin{proofof}{\eqref{SIQ-v1}, \eqref{SIQ-v2}, \eqref{R-SIQ-v1} and \eqref{R-SIQ-v2}}
Hereafter we use the abbreviation $X_{AB}=X_A+X_B$ and $X_{A'B'}=X_{A'}+X_{B'}$.
Then, for an arbitrary state $\xi$ on $AB$, we can transform $V_{U\xi U^\dagger}(X_{A'B'})$ as follows
\begin{align}
V_{U\xi U^\dagger}(X_{A'B'})&=\ex{X^2_{A'B'}}_{U\xi U^\dagger}-\ex{X_{A'B'}}^2_{U\xi U^\dagger}\nonumber\\
&=\ex{(U^\dagger X_{A'B'}U)^2}_{\xi}-\ex{U^\dagger X_{A'B'}U}^2_{U\xi U^\dagger}\nonumber\\
&=\ex{(X_{AB}-Z)^2}_{\xi}-\ex{X_{AB}-Z}^2_{\xi}\nonumber\\
&=V_{\xi}(X_{AB}-Z)\nonumber\\
&=V_{\xi}(X_{AB})-2Cov_{\xi}(X_{AB}:Z)+V_{\xi}(Z).
\end{align}
Due to $|Cov_{\xi}(X_{AB}:Z)|\le\sqrt{V_{\xi}(X_{AB})}\sqrt{V_{\xi}(Z)}$, we obtain 
\begin{align}
\left(\sqrt{V_{\xi}(X_{AB})}-\sqrt{V_{\xi}(Z)}\right)^2
\le
V_{U\xi U^\dagger}(X_{A'B'})
\le
\left(\sqrt{V_{\xi}(X_{AB})}+\sqrt{V_{\xi}(Z)}\right)^2\label{ev-v1}
\end{align}

Now, let us set $\xi=\xi_A\otimes\xi_B$, $\xi^{f}_{A'}:=\Tr_{B'}[U(\xi_A\otimes\xi_B) U^\dagger]$ and $\xi^{f}_{B'}:=\Tr_{A'}[U(\xi_A\otimes\xi_B) U^\dagger]$.
Then,
\begin{align}
V_{U\xi U^\dagger}(X_{A'B'})=V_{\xi^{f}_{A'}}(X_{A'})+2Cov_{U\xi U^\dagger}(X_{A'}:X_{B'})+V_{\xi^{f}_{B'}}(X_{B'}).
\end{align}
Due to $|Cov_{U\xi U^\dagger}(X_{A'}:X_{B'})|\le\sqrt{V_{\xi^{f}_{A'}}(X_{A'})}\sqrt{V_{\xi^{f}_{B'}}(X_{B'})}$.
Therefore, we obtain
\begin{align}
\left(\sqrt{V_{\xi^{f}_{A'}}(X_{A'})}-\sqrt{V_{\xi^{f}_{B'}}(X_{B'})}\right)^2
\le 
V_{U\xi U^\dagger}(X_{A'B'})
\le
\left(\sqrt{V_{\xi^{f}_{A'}}(X_{A'})}+\sqrt{V_{\xi^{f}_{B'}}(X_{B'})}\right)^2\label{ev-v2}
\end{align}
Substituting $\xi=\xi_{A}\otimes\xi_{B}$ into \eqref{ev-v1} and combining it with \eqref{ev-v2}, we obtain
\begin{align}
\sqrt{V_{\xi^{f}_{B'}}(X_{B'})}&\le\sqrt{V_{\xi_B}(X_B)}+\sqrt{V_{\xi_A}(X_A)}+\sqrt{V_{\xi^f_{A'}}(X_{A'})}+\sqrt{V_{\xi_A\otimes\xi_{B}}(Z)}.\label{Vineq-v}
\end{align}

Due to \eqref{vio-c-law}, we obtain
\begin{align}
\ex{X_A}_{\xi_A}-\ex{X_{A'}}_{\xi^f_{A'}}=-\ex{X_B}_{\xi_B}+\ex{X_{B'}}_{\xi^f_{B'}}+\ex{Z}_{\xi_A\otimes\xi_B}.
\end{align}
Therefore, for the decomposition $\rho_A=\sum_jp_j\rho_j$ such that $\SA=\sum_jp_j|\Delta_j|$, we obtain
\begin{align}
|\ex{X_{B'}}_{\rho^f_{j,B'}}-\ex{X_{B'}}_{\rho^f_{B'}}|-|\ex{Z}_{(\rho_{j,A}-\rho_{A})\otimes\rho_B}|\le|\Delta_j|\le|\ex{X_{B'}}_{\rho^f_{j,B'}}-\ex{X_{B'}}_{\rho^f_{B'}}|+|\ex{Z}_{(\rho_{j,A}-\rho_A)\otimes\rho_B}|\label{Eeq-v}
\end{align}
By using \eqref{Vineq-v} and \eqref{Eeq-v} instead of \eqref{Vineq} and \eqref{Eeq}, we obtain \eqref{SIQ-v1} by the same way as \eqref{SIQ1}.
We choose an ensemble $\{p_j,\rho_{j,A}\}$ satisfying $\calA=\sum_{j}p_j|\Delta_j|$.
Then, we obtain
\begin{align}
\SA&=\sum_{j}p_j|\Delta_j|\nonumber\\
&\le\sum_{j}p_j(|\ex{X_{B'}}_{\rho^f_{j,B'}}-\ex{X_{B'}}_{\rho^f_{B'}}|+|\ex{Z}_{(\rho_{j,A}-\rho_A)\otimes\rho_B}|)\nonumber\\
&\le\sum_{j}p_j|\ex{Z}_{(\rho_{j,A}-\rho_A)\otimes\rho_B}|
+\sum_{j}p_jD_{F}(\rho^{f}_{j,B'},\rho^{f}_{B'})\left(\sqrt{V_{\rho^{f}_{j,B'}}(X_{B'})}+\sqrt{V_{\rho^{f}_{B'}}(X_{B'})}+|\ex{X_{B'}}_{\rho^f_{j,B'}}-\ex{X_{B'}}_{\rho^f_{B'}}|\right)\nonumber\\
&\le\SA_Z
+\sum_{j}p_jD_{F}(\rho^{f}_{j,B'},\rho^{f}_{B'})\left(2\sqrt{V_{\rho_{B}}(X_{B})}+
\sqrt{V_{\rho_{j,A}}(X_A)}+\sqrt{V_{\rho^{f}_{j,A'}}(X_{A'})}+\sqrt{V_{\rho_A}(X_A)}+\sqrt{V_{\rho^{f}_{A'}}(X_{A'})}\right.\nonumber\\
&\left.+2\sqrt{V_Z}+|\Delta_j|+|\ex{Z}_{(\rho_{j,A}-\rho_A)\otimes\rho_B}|\right)\nonumber\\
&\le\SA_Z+2\delta\left(\sqrt{\calF}+\sqrt{V_{\rho_A}(X_A)}+\sqrt{V_{\rho^{f}_{A'}}(X_{A'})}+2\sqrt{V_Z}\right)\nonumber\\
&+\sqrt{\sum_{j}p_jD_{F}(\rho^{f}_{j,B'},\rho^{f}_{B'})^2}\left(\sqrt{\sum_{j}p_jV_{\rho_{j,A}(X_A)}}+\sqrt{\sum_{j}p_jV_{\rho^{f}_{j,A'}(X_{A'})}}+\sqrt{\sum_jp_j|\Delta_{j}|^2}+\sqrt{\sum_jp_j|\ex{Z}_{(\rho_{j,A}-\rho_A)\otimes\rho_B}|^2}\right)\nonumber\\
&\le\SA_Z
+2\delta\left(\sqrt{\calF}+2(\sqrt{V_{\rho_A}(X_A)}+\sqrt{V_{\rho^{f}_{A'}}(X_{A'})})+2\sqrt{V_Z}+\SA^{(2)}+\SA^{(2)}_{Z}\right).
\end{align}

Similarly, we derive \eqref{SIQ-v2} as follows:
\begin{align}
\SA&=\sum_{j}p_j|\Delta_j|\nonumber\\
&\le\sum_{j}p_j(|\ex{X_{B'}}_{\rho^f_{j,B'}}-\ex{X_{B'}}_{\rho^f_{B'}}|+|\ex{Z}_{(\rho_{j,A}-\rho_A)\otimes\rho_B}|)\nonumber\\
&\le\sum_{j}p_j|\ex{Z}_{(\rho_{j,A}-\rho_A)\otimes\rho_B}|
+\sum_{j}p_jD_{F}(\rho^{f}_{j,B'},\rho^{f}_{B'})\left(\sqrt{V_{\rho^{f}_{j,B'}}(X_{B'})}+\sqrt{V_{\rho^{f}_{B'}}(X_{B'})}+|\ex{X_{B'}}_{\rho^f_{j,B'}}-\ex{X_{B'}}_{\rho^f_{B'}}|\right)\nonumber\\
&\le\SA_Z+2\delta\sqrt{V_{\rho^f_{B'}}(X_{B'})}\nonumber\\
&+\sqrt{\sum_{j}p_jD_{F}(\rho^{f}_{j,B'},\rho^{f}_{B'})^2}\left(\sqrt{\sum_{j}p_jV_{\rho^{f}_{j,B'}}(X_{B'})}+\sqrt{\sum_jp_j|\Delta_j|^2}+\sqrt{\sum_jp_j|\ex{Z}_{(\rho_{j,A}-\rho_A)\otimes\rho_B}|^2}\right)\nonumber\\
&\le\SA_Z
+2\delta\left(\sqrt{\calF_f}+\SA^{(2)}+\SA^{(2)}_{Z}\right).
\end{align}

We can show \eqref{R-SIQ-v1} and \eqref{R-SIQ-v2} in the same way.

\end{proofof}

Finally, we prove \eqref{ur}--\eqref{ur5}.
\begin{proofof}{\eqref{ur}--\eqref{ur5}}
The inequalities \eqref{ur2} and \eqref{ur3} are easily obtained from the definition. So, we prove \eqref{ur}, \eqref{ur4} and \eqref{ur5}.
We firstly show \eqref{ur} and \eqref{ur5}.
Since the square of the average is smaller than the average of square, the inequality $\SA_Z\le\SA^{(2)}_{Z}$ in \eqref{ur} clearly holds.
We can easily derive the remaining parts of \eqref{ur} and \eqref{ur5} from the following inequality holds for arbitrary Hermitian $Y$ and state $\xi$ and its decomposition $\xi=\sum_{l}q_l\xi_l$:
\begin{align}
\sum_{l}q_l\left(\ex{Y}_{\xi_l}-\ex{Y}_{\xi}\right)^2\le V_{\xi}(Y)\label{A2Vpre}
\end{align}
We obtain \eqref{A2Vpre} as follows
\begin{align}
V_{\xi}(Y)&=\ex{Y^2}_{\xi}-\ex{Y}^2_{\xi}\nonumber\\
&=\sum_lq_l\ex{Y^2}_{\xi_l}-\left(\sum_lq_l\ex{Y}_{\xi_{l}}\right)^2\nonumber\\
&\ge\sum_lq_l\ex{Y}^2_{\xi_l}-\left(\sum_lq_l\ex{Y}_{\xi_{l}}\right)^2\nonumber\\
&=\sum_{l}q_l\left(\ex{Y}_{\xi_l}-\sum_{l'}q_{l'}\ex{Y}_{\xi_{l'}}\right)^2\nonumber\\
&=\sum_{l}q_l\left(\ex{Y}_{\xi_l}-\ex{Y}_{\xi}\right)^2.
\end{align}

Similarly, we can easily derive \eqref{ur4} from the following inequality holds for arbitrary Hermitian $Y$ and state $\xi$ and its decomposition $\xi=\sum_{l}q_l\xi_l$:
\begin{align}
\sum_{l}q_l\left|\ex{Y}_{\xi_l}-\ex{Y}_{\xi}\right|\le M_{\xi}(Y)\label{A2Mpre}
\end{align}
We obtain \eqref{A2Mpre} as follows:
\begin{align}
M_{\xi}(Y)&=\left<|Y-\ex{Y}_{\xi}|\right>_{\xi}\nonumber\\
&=\sum_lq_l\ex{|Y-\ex{Y}_{\xi}|}_{\xi_{l}}\nonumber\\
&\stackrel{(a)}{\ge}\sum_{l}q_l|\ex{Y-\ex{Y}_{\xi}}_{\xi_l}|\nonumber\\
&=\sum_{l}q_l|\ex{Y}_{\xi_l}-\ex{Y}_{\xi}|.
\end{align}
where we use the inequality $|\ex{H}_{\zeta}|\le\ex{|H|}_{\zeta}$ which holds for arbitrary Hermitian $H$ and state $\zeta$ in (a).

\end{proofof}

\section{Derivations of \eqref{check_ex1}--\eqref{check_ex3} and \eqref{upper_ex} in the main text}\label{Stightness}
In this section, we derive \eqref{check_ex1}--\eqref{check_ex3} and \eqref{upper_ex} in the main text.

\subsection{Model}
For the reader's convenience, we write down our model again. 
Let us consider four qubits, $A$, $R_A$, $B_1$ and $R_{B_1}$.
We also take a natural number $b$ and a $b+2$-dimensional system $B_2$.
For $A$, $B_1$, $B_2$, we define Hermitian operators $X_A$, $X_{B_1}$ and $X_{B_2}$ as follows:
\eq{
X_A&:=\ket{1}\bra{1}_{A},\\
X_{B_1}&:=\ket{1}\bra{1}_{B_1},\\
X_{B_2}&:=\sum^{b+2}_{x=1}2x\ket{x}\bra{x}_{B_2},
}
where $\{\ket{x}_{A}\}_{x=0,1}$, $\{\ket{x}_{B_1}\}_{x=0,1}$, $\{\ket{x}_{B_2}\}_{x=0,...,d+2}$ are orthogonal basis on $A$, $B_1$, and $B_2$, respectively.

On the above system, we prepare the following initial states:
\eq{
\ket{\psi_{AR_A}}&:=\frac{\ket{00}_{AR_A}+\ket{11}_{AR_A}}{\sqrt{2}},\\
\ket{\phi_{B_1R_{B_1}}}&:=\frac{\ket{00}_{B_1R_{B_1}}+\ket{11}_{B_1R_{B_1}}}{\sqrt{2}},\\
\ket{\phi_{B_2}}&:=\frac{1}{\sqrt{b}}\sum^{b}_{x=1}\ket{x}_{B_2}.
} 
And we prepare the following unitary $U_{AB_1B_2}$ on $AB_1B_2$:
\eq{
U_{AB_1B_2}:=&\sum_{1\le k\le b+1}
 \left(\ket{11}\bra{00}_{AB_1}\otimes\ket{k-1}\bra{k}
+\ket{10}\bra{01}_{AB_1}\otimes\ket{k}\bra{k}
+\ket{01}\bra{10}_{AB_1}\otimes\ket{k}\bra{k}
+\ket{00}\bra{11}_{AB_1}\otimes\ket{k}\bra{k-1}\right)\nonumber\\
&+\ket{00}\bra{00}_{AB_1}\otimes\ket{0}\bra{0}
+\ket{01}\bra{01}_{AB_1}\otimes\ket{0}\bra{0}
+\ket{10}\bra{10}_{AB_1}\otimes\ket{0}\bra{0}
+\ket{11}\bra{11}_{AB_1}\otimes\ket{b+1}\bra{b+1}.
}
After the unitary $U$, we perform a CPTP map on $AR_{B_1}$ to recover  the initial state $\ket{\psi_{AR_A}}$ on $AR_A$.
Then, following our framework, we define the minimum recovery error as follows:
\eq{
\delta=\min_{\calR}D_F(\psi_{AR_A},id_{R_A}\otimes\calR\circ\calE_{A\rightarrow AR_{B_1}}(\psi_{AR_A})).
}
Here $\calE_{A\rightarrow AR_{B_1}}(...):=\Tr_{B_1B_2}[U_{AB_1B_2}(...\otimes\phi_{B_2}\otimes\phi_{BR_{B_{1}}}) U^\dagger_{AB_1B_2}]$.

Since $U_{AB_1B_2}$ conserves $X_{A}+X_{B_1}+X_{B_2}$ and $\rho_A:=\Tr_{R_A}[\psi_{AR_A}]=\frac{\ket{0}\bra{0}+\ket{1}\bra{1}}{2}$, our bound is applicable to the above model:
\eq{
\delta\ge\frac{\calA_2}{\sqrt{\calF_f}+\Delta_{\max}}\label{Slower_a}
}
Here, to check our bound, we calculate each term in the RHS of \eqref{Slower_a}, and give upper bounds of $\delta$ with analytical and numerical ways.

\begin{figure}[tb]
		\centering
		\includegraphics[width=.45\textwidth]{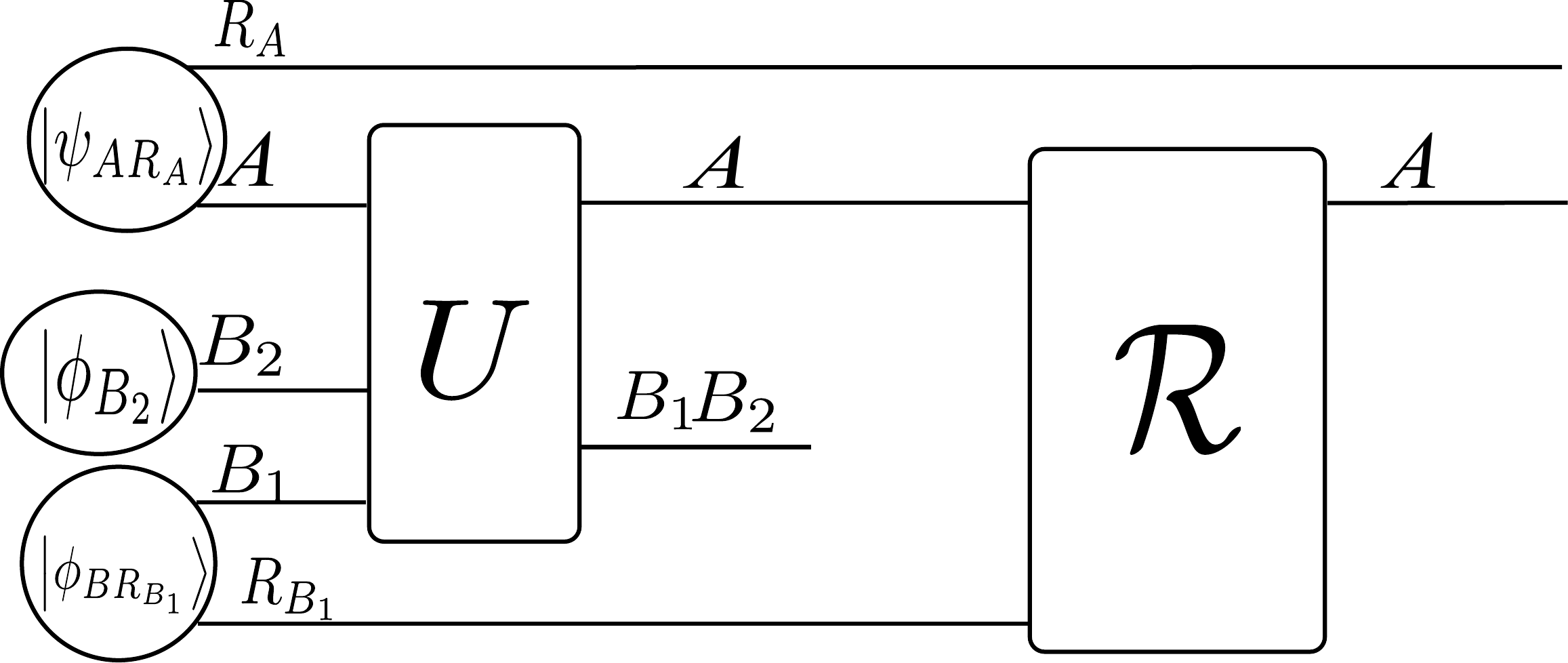}
		\caption{Schematic diagram of the concrete model for numerical check.}
		\label{Scheck_fig}
	\end{figure}

\subsubsection{analytical lower bound: Derivations of \eqref{check_ex1}--\eqref{check_ex3}}
Let us calculate $\calA_2$, $\calF_f$ and $\Delta_{\max}$ in the RHS of \eqref{Slower_a}. The calculations correspond to the derivations of \eqref{check_ex1}--\eqref{check_ex3}.
We firstly calculate $\calA_2$. The definition of $\calA_2$ is as follows:
\eq{
\calA_2:=\frac{\sum_{i=0,1}|\ex{X_A}_{\ket{i}\bra{i}}-\ex{X_A}_{\calE_{A\rightarrow A}(\ket{i}\bra{i})}|}{2},
}
where $\calE_{A\rightarrow A}(...):=\Tr_{B_1B_2R_{B_1}}[U_{AB_1B_2}(...\otimes\phi_{B_2}\otimes\phi_{BR_{B_{1}}}) U^\dagger_{AB_1B_2}]$.
We remark that
\eq{
\calE_{A\rightarrow A}(\ket{0}\bra{0})&=\ket{1}\bra{1},\\
\calE_{A\rightarrow A}(\ket{1}\bra{1})&=\ket{0}\bra{0}.
}
Therefore, we obtain
\eq{
\calA_2=1.
}

Next, we evaluate $\Delta_{\max}$. Due to $\Delta_{\max}:=\max_{\rho\mbox{ on }\mathrm{supp}(\rho_A)}|\ex{X_A}_{\rho}-\ex{X_A}_{\calE_{A\rightarrow A}(\rho)}|$ and $\Delta_{X_A}=1$, we obtain
\eq{
\Delta_{\max}\le1.
}

Next, we evaluate $\calF_f$. Note that
\eq{
\calF_f=4V_{\rho'_{B_1B_2}}(X_{B_1}+X_{B_2}),
}
where $\rho'_{B_1B_2}:=\Tr_{AR_AR_{B_1}}[1_{R_A}\otimes U_{AB_1B_2}(\psi_{AR_A}\otimes\phi_{B_2}\otimes\phi_{BR_{B_{1}}}) 1_{R_A}\otimes U^\dagger_{AB_1B_2}]$.
To evaluate $V_{\rho'_{B_1B_2}}(X_{B_1}+X_{B_2})$, we firstly note that
\eq{
&U_{AB_1B_2}\ket{\psi_{AR_A}}\otimes\ket{\phi_{B_2}}\otimes\ket{\phi_{BR_{B_{1}}}}\non
&=\frac{1}{2}(\ket{0101}_{R_AAR_{B_1}B_1}\ket{\phi^{(-1)}_{B_2}}+\ket{0110}_{R_AAR_{B_1}B_1}\ket{\phi_{B_2}}+\ket{1001}_{R_AAR_{B_1}B_1}\ket{\phi_{B_2}}+\ket{1010}_{R_AAR_{B_1}B_1}\ket{\phi^{(+1)}_{B_2}}),\label{Shukusen1}
}
where $\ket{\phi^{(-1)}_{B_2}}:=\frac{1}{\sqrt{b}}\sum^{b-1}_{x=0}\ket{x}_{B_2}$ and $\ket{\phi^{(+1)}_{B_2}}:=\frac{1}{\sqrt{b}}\sum^{b+1}_{x=2}\ket{x}_{B_2}$.
Therefore, we obtain
\eq{
\rho'_{B_1B_2}=\frac{1}{4}(\ket{1}\bra{1}\otimes\phi^{(-1)}_{B_2}+\ket{0}\bra{0}\otimes\phi_{B_2}+\ket{1}\bra{1}\otimes\phi_{B_2}+\ket{0}\bra{0}\otimes\phi^{(+1)}_{B_2}).
}
Therefore, with using $V_{\sum_xq_x\rho_x}(W)=V_{\{q_x\}}(\{\ex{W}_{\rho_x}\})+\sum_{x}q_xV_{\rho_x}(X)$, we obtain
\eq{
V_{\rho'_{B_1B_2}}(X_{B_1}+X_{B_2})&=\frac{5}{4}+V_{\phi_{B_2}}(X_{B_2})\non
&=\frac{5}{4}+\frac{1}{b}\sum^{b}_{k=1}k^2-\frac{(\sum^{b}_{k=1}k)^2}{b^2}\non
&=\frac{5}{4}+\frac{(b+1)(2b+1)}{6}-\frac{(b+1)^2}{4}\non
&=\frac{b^2+14}{12}
}
Therefore, we obtain
\eq{
\calF_f=\frac{b^2+14}{3}.
}
Hence, $\delta$ is bounded as follows:
\eq{
\delta&\ge\frac{1}{\sqrt{\calF_f}+1}\non
&=\frac{1}{\sqrt{\frac{b^2+14}{3}}+1}
}

\subsubsection{Analytical upper bound:}
Next, we give \eqref{upper_ex} in the main text, an analytical upper bound for $\delta$.
To derive the upper bound, we define the following $V_{AR_{B_1}}$ on $AR_{B_1}$.
\eq{
V_{AR_{B_1}}:=\ket{11}\bra{00}_{AR_{B_1}}+\ket{10}\bra{01}_{AR_{B_1}}+\ket{01}\bra{10}_{AR_{B_1}}+\ket{00}\bra{11}_{AR_{B_1}}
}
Since the unitary operation $\calR^{V}_{AR_{B_1}\rightarrow A}(...):=\Tr_{R_{B_1}}[V_{AR_{B_1}}(...)V^\dagger_{AR_{B_1}}]$ is a CPTP map from $AR_{B_1}$ to $A$, we obtain
\eq{
\delta\le D_F(\psi_{AR_A},\mathrm{id}_{R_A}\otimes\calR^{V}_{AR_{B_1}\rightarrow A}\circ\calE_{A\rightarrow AR_{B_1}}(\psi_{AR_A})).
}
To evaluate the RHS of the above inequality, note the following relations:
\eq{
&\mathrm{id}_{R_A}\otimes\calR^{V}_{AR_{B_1}\rightarrow A}\circ\calE_{A\rightarrow AR_{B_1}}(\psi_{AR_A})=\Tr_{R_{B_1}B_1B_2}[V_{AR_{B_1}}U_{AB_1B_2}\psi_{AR_A}\otimes\phi_{B_2}\otimes\phi_{BR_{B_{1}}}U^\dagger_{AB_1B_2}V^\dagger_{AR_{B_1}}],\\
&V_{AR_{B_1}}U_{AB_1B_2}\ket{\psi_{AR_A}}\otimes\ket{\phi_{B_2}}\otimes\ket{\phi_{BR_{B_{1}}}}\non
&\stackrel{(a)}{=}\frac{1}{2}(\ket{0011}_{R_AAR_{B_1}B_1}\ket{\phi^{(-1)}_{B_2}}+\ket{0000}_{R_AAR_{B_1}B_1}\ket{\phi_{B_2}}+\ket{1111}_{R_AAR_{B_1}B_1}\ket{\phi_{B_2}}+\ket{1100}_{R_AAR_{B_1}B_1}\ket{\phi^{(+1)}_{B_2}})\non
&=\sqrt{\frac{b-2}{b}}\ket{\psi_{AR_A}}\ket{\phi_{BR_{B_1}}}\ket{\phi'_{B_2}}
\non
&+\frac{1}{2\sqrt{b}}(\ket{0011}_{R_AAR_{B_1}B_1}(\ket{0}+\ket{1})+\ket{0000}_{R_AAR_{B_1}B_1}(\ket{1}+\ket{b})+\ket{1111}_{R_AAR_{B_1}B_1}(\ket{1}+\ket{b})+\ket{1100}_{R_AAR_{B_1}B_1}(\ket{b}+\ket{b+1}))
}
where $\ket{\phi'_{B_2}}:=\sqrt{\frac{1}{b-2}}\sum^{b-1}_{x=2}\ket{x}$ and we use \eqref{Shukusen1} in $(a)$.

Using the above, we evaluate $F(\psi_{AR_A},\mathrm{id}_{R_A}\otimes\calR^{V}_{AR_{B_1}\rightarrow A}\circ\calE_{A\rightarrow AR_{B_1}}(\psi_{AR_A}))$ as follows:
\eq{
&F^2(\psi_{AR_A},\mathrm{id}_{R_A}\otimes\calR^{V}_{AR_{B_1}\rightarrow A}\circ\calE_{A\rightarrow AR_{B_1}}(\psi_{AR_A}))\non
&=\bra{\psi_{AR_A}}
\mathrm{id}_{R_A}\otimes\calR^{V}_{AR_{B_1}\rightarrow A}\circ\calE_{A\rightarrow AR_{B_1}}(\psi_{AR_A})\ket{\psi_{AR_A}}\non
&=\bra{\psi_{AR_A}}\Tr_{R_{B_1}B_1B_2}[V_{AR_{B_1}}U_{AB_1B_2}\psi_{AR_A}\otimes\phi_{B_2}\otimes\phi_{BR_{B_{1}}}U^\dagger_{AB_1B_2}V^\dagger_{AR_{B_1}}]\ket{\psi_{AR_A}}\non
&=\bra{\psi_{AR_A}}\sum^{4(b+2)}_{j=1}\bra{\phi_j}V_{AR_{B_1}}U_{AB_1B_2}\psi_{AR_A}\otimes\phi_{B_2}\otimes\phi_{BR_{B_{1}}}U^\dagger_{AB_1B_2}V^\dagger_{AR_{B_1}}\ket{\phi_j}\ket{\psi_{AR_A}}\non
&\ge\bra{\psi_{AR_A}}\bra{\phi_0}V_{AR_{B_1}}U_{AB_1B_2}\psi_{AR_A}\otimes\phi_{B_2}\otimes\phi_{BR_{B_{1}}}U^\dagger_{AB_1B_2}V^\dagger_{AR_{B_1}}\ket{\phi_0}\ket{\psi_{AR_A}}\non
&=\frac{b-2}{b}.
}
where $\{\phi_j\}$ is a orthonormal basis of $R_{B_1}B_1B_{2}$ satisfying $\phi_0=\phi_{B_1R_{B_1}}\otimes\phi'_{B_2}$.
Due to $D_F=\sqrt{1-F^2}$, we obtain
\eq{
D_F(\psi_{AR_A},\mathrm{id}_{R_A}\otimes\calR^{V}_{AR_{B_1}\rightarrow A}\circ\calE_{A\rightarrow AR_{B_1}}(\psi_{AR_A}))\le\sqrt{\frac{2}{b}}
}
Therefore, we obtain
\eq{
\delta\le\sqrt{\frac{2}{b}}.
}

\end{widetext}


\begin{thebibliography}{10}
\bibitem{footnote-contribution}
We remark that with using HP model with conservation laws, Refs.\cite{Yoshida-soft,JLiu,Nakata} physically argue a similar effect indicating that there would be a delay in the escape of information from a black hole, when the sizes of the black hole and the diary are comparable. 
We stress here that our general theorems rigorously prove that 1) a significant delay exists even when the diary is much smaller than the black hole, and 2) when the sizes of the black hole and the diary are comparable, the escape of the information from a black hole is not merely delayed but does not end until the black hole has completely evaporated.
These two messages clearly show that as long as the HP model is correct (i.e., the total system can be described as qubits with  unitary dynamics), conservation laws have a non-negligible effect on the speed of information escape, regardless of the size of the black hole.
\bibitem{footnote-size} There are two situations that $k$ is comparable with $N$ while $N$ is macroscopically large.
The first situation is the one so called “feeding black hole” \cite{Lloyd-Preskill}.
In this scenario, we continuously throw matter into the black hole, so as the energy leak from the black hole by the Hawking radiation is equal to the amount of the matter thrown into the black hole.
By throwing $A$ into $B$ in this way, we can make $k$ comparable with $N$.
The second situation is the one so called “black hole cannibalism” \cite{canni}.
In this scenario, we just consider $A$ as another black hole, and thus $k$ can be comparable with $N$.
Note that since our framework is more general than the original Hayden-Preskill setup, we can treat the above two situations in our framework.
\bibitem{footnoteA2}When $\rho_A$ can be expressed as $\rho_A=(\rho_0+\rho_1)/2$, we can give tighter versions of \eqref{SIQ1}, \eqref{SIQ2}, and \eqref{SIQ1'-Smain} by substituting $2\calA_2$ for $\calA$, where $\calA_2:=\max_{\rho_0,\rho_1:\rho_A=(\rho_0+\rho_1)/2}\sum_{j=0,1}\frac{|\Delta_j|}{2}$. See Appendix \ref{tighterSIQ}.
\end{thebibliography}
\end{document}